\documentclass[preprint,3p]{elsarticle}

\usepackage{graphicx}
\usepackage{epsfig}
\usepackage{amsmath}
\usepackage{epstopdf}
\usepackage{multirow}
\usepackage{rotating}  
\usepackage{booktabs}
\usepackage{caption}
\usepackage{subcaption}
\usepackage{ifpdf}
\usepackage{color}
\usepackage{amssymb}

\usepackage{amsthm}
\usepackage[numbers]{natbib}
\usepackage[colorlinks=true, linkcolor=blue, citecolor=blue, urlcolor=blue]{hyperref}
\usepackage{lineno}
\usepackage{xcolor}
\usepackage{algorithm2e}
\usepackage{bm}
\usepackage{mathrsfs}

\newtheorem{proposition}{Proposition}[section]

\newtheorem{remark}{Remark}[section]

\usepackage{lscape}
\usepackage{needspace}

\usepackage{float}
\usepackage{adjustbox}  
\numberwithin{equation}{section}

\journal{Physica A: Statistical Mechanics and its Applications}

\begin{document}

\begin{frontmatter}

\title{Markovian Shock-Source Tracing and Multidimensional Asset Roles in Exchange Rates, Gold Futures, and Bitcoin}

\author[1,4]{Seung Ho Choi\fnref{equal}}

\author[2]{Seoin Jang\fnref{equal}}

\author[5]{Hyunwoo Lee\corref{cor1}}
\ead{hyunwoolee@kias.re.kr}

\author[1]{Hayoung Choi\corref{cor1}}
\ead{hayoung.choi@knu.ac.kr}

\address[1]{Department of Mathematics, Kyungpook National University, 
80, Daehak-ro, Buk-gu, 41566 Daegu, Republic of Korea}
\address[2]{Department of Statistics, Kyungpook National University, 
80, Daehak-ro, Buk-gu, 41566 Daegu, Republic of Korea}
\address[4]{Daegu Metropolitan Office of Education, 
150, Geomdan-ro, Buk-gu, 41521 Daegu, Republic of Korea}

\address[5]{Korea Institute for Advanced Study, 
85 Hoegiro, Dongdaemun-gu, 02455 Seoul, Republic of Korea}

\cortext[cor1]{Co-corresponding authors.}
\fntext[equal]{These authors contributed equally to this work.}


\begin{abstract}
This study examines cross-asset connectedness in an international financial network of major exchange rates, gold futures, and Bitcoin. Moving beyond the conventional net transmitter--receiver classification, we characterize asset roles through three complementary dimensions: direct spillover transmission, stationary source-tracing dynamics, and multistep upstream connectivity. Return spillovers are estimated using VAR generalized forecast-error variance decomposition (VAR-GFEVD). The positive pairwise net-spillover structure is then mapped into a row-stochastic Markov kernel whose transitions trace dominant net-spillover sources in the reverse direction of the original transmission edges. A stationary departure flux describes long-run movement in this source-tracing chain, while Viral Centrality is evaluated by the deterministic probability-propagation algorithm of Fink et al. to approximate multistep upstream reach. Empirically, gold futures emerge as the dominant direct net transmitter and the leading stationary source-tracing node, but do not have the largest Viral Centrality. Several exchange-rate nodes classified as direct net receivers have relatively high Viral Centrality, indicating broad conditional access to upstream source nodes. Bitcoin occupies an intermediate role. An auxiliary specification including the U.S. Dollar Index yields qualitatively similar role differentiation. The results show that direct connectedness, stationary source tracing, and approximate multistep upstream reach cannot be inferred from NET spillovers alone. All diffusion interpretations are descriptive and conditional on the
estimated network, row normalization, and restart closure.

\end{abstract}


\end{frontmatter}

\section{Introduction}
\label{sec:introduction}

A central issue in financial connectedness analysis is how to characterize the role of each asset in a system of interacting markets. In the conventional spillover literature, this role is usually described through the distinction between net transmitters and net receivers. Assets with positive net spillovers are interpreted as shock transmitters, whereas assets with negative net spillovers are interpreted as shock receivers. This interpretation has been useful for measuring directional connectedness, systemic risk, and crisis-period spillovers in financial markets~\cite{diebold2009measuring,diebold2012better,diebold2014network,choi2022volatility}. However, complex financial systems may involve richer forms of shock propagation than can be captured by a single transmitter--receiver classification.

International financial markets are increasingly understood as weighted and directed networks rather than as collections of independent assets. Exchange rates, precious metals, digital assets, and broad dollar-related factors interact through global liquidity conditions, monetary-policy expectations, inflation uncertainty, risk aversion, and portfolio rebalancing. In such a system, shocks may not only move directly from one asset to another, but may also persist through repeated network transitions or diffuse indirectly through multistep cascade-like paths. From the perspective of complex financial networks, cross-asset interactions can therefore be studied not only as pairwise relationships but also as networked diffusion processes~\cite{mantegna1999hierarchical,onnela2004clustering,tumminello2005tool,diebold2014network}.

This perspective suggests a limitation of conventional spillover measures. Although net spillovers identify the dominant direction of direct transmission, they do not fully describe the topology obtained by recursively tracing the sources of those spillovers~\cite{diebold2014network,choi2022volatility}. A direct net transmitter need not be the node most frequently visited by a reverse source-tracing walk. Conversely, a direct net receiver may be connected to several upstream sources through indirect paths. These considerations motivate a distinction among direct transmission, stationary source-tracing dynamics, and multistep upstream connectivity~\cite{masuda2017random,fink2023centrality}.

The novelty of this study lies in characterizing financial connectedness through several network summaries rather than through a binary transmitter--receiver classification alone. We distinguish direct transmission, stationary occupation and departure in a reverse source-tracing chain, and Viral Centrality on the same conditional transition topology. This role-based perspective separates direct net sources from nodes having broad multistep access to upstream sources. It also highlights that assets with similar NET spillovers may occupy different positions in the normalized source-tracing topology.

To operationalize this idea, we construct a directed network from return spillovers estimated through a VAR-GFEVD framework. The VAR-GFEVD framework provides ordering-invariant forecast-error variance shares and is widely used to measure directional connectedness among financial assets~\cite{koop1996impulse,pesaran1998generalized,diebold2012better,diebold2014network}. The positive pairwise net-spillover matrix is transformed into a row-stochastic Markov kernel. A transition from asset $i$ to asset $j$ traces a dominant net-spillover supplier $j$ of asset $i$ and therefore reverses the original transmission edge. The transition probabilities are conditional relative weights; row normalization removes the absolute incoming magnitude. The resulting kernel is thus a stochastic representation of source-tracing topology, not a structural law of economic shock propagation~\cite{lovasz1993random,masuda2017random,seabrook2023tutorial}.

This Markovian representation enables the analysis to distinguish three complementary dimensions of asset roles. The first is direct spillover transmission, measured by conventional TO, FROM, and NET spillovers. The second is stationary source tracing, summarized by the stationary distribution and the probability of departing from each node in the reverse chain. The third is approximate multistep upstream reach, measured by the deterministic Viral Centrality algorithm of Fink et al.~\cite{fink2023centrality}. Viral Centrality propagates activation probabilities iteratively and is distinct from a Monte Carlo estimate of independent-cascade spread. Combining these dimensions extends the binary transmitter--receiver classification while retaining a descriptive, noncausal interpretation.

This framework is applied to an international financial asset network consisting of major exchange rates, gold futures, and Bitcoin. These assets are not selected because they perform the same financial function, but because they occupy closely related positions in the international financial system. Major exchange rates reflect currency-market adjustment through capital flows, trade-related price changes, monetary-policy expectations, and global risk sentiment. Gold futures are commonly associated with safe-haven demand, inflation expectations, and real interest-rate conditions, while Bitcoin has increasingly been discussed as an alternative digital asset whose behavior may be sensitive to global liquidity and risk appetite~\cite{ciner2013hedges,reboredo2013gold,dyhrberg2016bitcoin,urquhart2019bitcoin}. Because exchange rates, gold, and Bitcoin are jointly linked to dollar-related liquidity conditions, portfolio rebalancing, and shifts in market sentiment, their interactions provide a suitable empirical setting for examining heterogeneous shock-diffusion roles.

The baseline network consists of JPY/USD, GBP/USD, EUR/USD, CNY/USD, gold futures, and BTC/USD. The U.S. Dollar Index is excluded from the baseline specification not because dollar-related conditions are irrelevant, but because the objective is to examine the internal diffusion topology among exchange rates, gold, and Bitcoin without allowing a broad dollar factor to enter the network as an explicit and potentially dominant node. This choice is also motivated by the fact that bilateral exchange rates share common currency risk factors and dollar-related systematic components~\cite{lustig2011common,verdelhan2018share}. Accordingly, the U.S. Dollar Index is treated as an auxiliary specification rather than as the main object of analysis. This design allows us to examine whether the baseline diffusion-role patterns remain broadly consistent when a broad dollar factor is introduced into the network.

The empirical findings suggest that asset roles are multidimensional under the baseline network. Gold futures emerge as the dominant direct net transmitter and the leading stationary source-tracing node, but do not have the largest Viral Centrality. By contrast, some exchange-rate nodes that are direct net receivers have relatively high Viral Centrality, indicating broad conditional access to upstream source nodes in the reverse topology. Bitcoin occupies an intermediate position. These findings show that an asset's descriptive network role depends on the measure considered.

This study makes three main contributions. First, it links VAR-GFEVD connectedness to a reverse source-tracing Markov representation of the positive pairwise net-spillover network. Second, it distinguishes direct spillover roles, stationary source-tracing roles, and approximate multistep upstream connectivity without proposing a new VAR estimator or centrality measure. Third, it applies this framework to major exchange rates, gold futures, and Bitcoin and documents heterogeneous positions across these conditional network summaries.

Overall, this paper combines VAR-GFEVD connectedness, a reverse source-tracing Markov chain, stationary departure probabilities, Viral Centrality, and an auxiliary dollar-index comparison. The resulting role classification complements conventional net-transmitter and net-receiver measures while avoiding structural causal claims.

The remainder of this paper is organized as follows. Section~\ref{sec:literature} reviews the related literature on cross-asset connectedness, spillover networks, safe-haven assets, complex financial networks, and diffusion-based network measures. Section~\ref{sec:data_methodology} describes the VAR-GFEVD framework, the reverse source-tracing kernel, stationary departure, deterministic Viral Centrality, and state indicators. Section~\ref{sec:results} presents the direct-spillover, stationary source-tracing, upstream-connectivity, and auxiliary dollar-index results. Sections~\ref{sec:discussion} and~\ref{sec:conclusion} discuss the implications and limitations and conclude the paper.

\section{Literature review}
\label{sec:literature}

\subsection{Cross asset connectedness and safe haven assets}

International financial assets are increasingly understood as a connected system rather than as independent markets linked only by pairwise correlations. Exchange rates, commodities, equity related assets, gold, and digital assets are jointly affected by global liquidity conditions, portfolio rebalancing, inflation expectations, risk aversion, and hedging demand. From a complex systems perspective, these assets can be interpreted as nodes in a financial network through which shocks propagate over time ~\cite{mantegna1999hierarchical}. Prior studies show that the relationships among oil prices, exchange rates, and stock markets are dynamic and state dependent. Basher et al.~\cite{basher2012oil} analyze oil prices, exchange rates, and emerging stock markets within a structural VAR framework, while Tsai~\cite{tsai2012relationship} shows that the relationship between stock price indices and exchange rates in Asian markets varies across exchange rate conditions. Basher et al.~\cite{basher2016impact} further demonstrate, using a Markov switching approach, that the effects of oil shocks on real exchange rates differ across regimes. These findings suggest that cross asset transmission is not fixed, but varies depending on market conditions, shock origins, and macro financial environments.

Gold has also been widely examined as a hedge and safe haven asset. Ciner et al.~\cite{ciner2013hedges} show that gold may hedge exchange rate fluctuations and function as a safe haven under extreme currency market conditions. Reboredo~\cite{reboredo2013gold}, using a copula approach, finds that gold can reduce downside risk in currency portfolios and act as a hedge against U.S. dollar depreciation. Bitcoin has also been discussed as an alternative financial asset whose role may differ from that of traditional safe haven assets and whose behavior can be sensitive to global liquidity and risk appetite ~\cite{dyhrberg2016bitcoin}. More recent studies extend the literature from correlation and dependence analysis to spillover and connectedness frameworks. Antonakakis and Kizys~\cite{antonakakis2015dynamic} examine dynamic return and volatility spillovers between commodity and currency markets, showing that the transmitter and receiver roles of assets change over time. Liu et al.~\cite{liu2024risk} analyze risk spillovers among oil, gold, stock, and foreign exchange markets in G20 economies from a network perspective, emphasizing the coexistence of risk contagion and diversification effects.

These studies imply that exchange rates, gold, and Bitcoin should not be treated as separate asset classes with independent dynamics. Exchange rates reflect international capital flows, trade related price adjustments, and monetary policy expectations, while gold responds to safe haven demand, inflation expectations, and real interest rate conditions. Bitcoin, in turn, has been discussed as an alternative asset whose behavior may overlap partly with gold and partly with risk sensitive financial assets~\cite{dyhrberg2016bitcoin,urquhart2019bitcoin}. Therefore, the relationships among exchange rates, gold, and Bitcoin are likely to involve not only direct pairwise dependence but also indirect and time varying transmission channels. This provides a network based motivation for analyzing these assets jointly as interacting nodes in a directed and weighted financial system.

\subsection{Spillover networks and transmitter receiver classification}

Methodologically, this study builds on the spillover and network connectedness literature. Diebold and Yilmaz~\cite{diebold2009measuring} propose a VAR based spillover index using forecast error variance decomposition, which allows system wide connectedness to be summarized by the total connectedness index and decomposed into directional TO, FROM, and NET spillovers. Koop et al.~\cite{koop1996impulse} and Pesaran and Shin~\cite{pesaran1998generalized} provide the foundation for generalized impulse response and generalized variance decomposition, which are useful because they avoid the variable ordering problem inherent in Cholesky based approaches. Diebold and Yilmaz~\cite{diebold2014network} further reinterpret variance decomposition as a weighted directed network, where nodes represent financial assets and edges represent shock transmission.

Within this framework, assets are commonly classified as net transmitters or net receivers according to whether they transmit more shocks to the system than they receive from it. This interpretation has been widely used in empirical studies of financial spillovers and crisis period connectedness. For example, Choi~\cite{choi2022volatility} applies the Diebold and Yilmaz spillover index to stock market volatility connectedness among Northeast Asia and the United States and identifies time varying net transmitter and receiver roles during the global financial crisis and the COVID 19 pandemic. Such studies demonstrate the usefulness of directional spillover measures for describing the hierarchy of shock transmission across financial markets.

However, the transmitter--receiver dichotomy does not describe every topological property of a connected network. A direct net transmitter need not be the most frequently occupied node in a source-tracing random walk. Conversely, a direct net receiver may be linked to several upstream suppliers through indirect paths. This motivates comparison of direct spillovers with stationary and multistep summaries whose orientation and normalization are stated explicitly.

\subsection{Markovian source tracing and deterministic network roles}

Building on the spillover-network framework, the present study constructs a Markov transition structure from the estimated spillover matrix. Related applications link spillover networks to Markov chains; for example, Zhang et al.~\cite{zhang2020detecting} model intra-urban housing-market spillovers through a spatial Markov structure. Once a matrix is transformed into a row-stochastic kernel, stationary distributions and flows characterize the selected walk~\cite{seabrook2023tutorial,masuda2017random}. Their substantive meaning, however, depends on the direction assigned to each transition. In this study the walk is explicitly reverse source tracing.

Fink et al.~\cite{fink2023centrality} propose Viral Centrality as a deterministic approximation motivated by expected independent-cascade reach on weighted directed networks. The recursion is exact for acyclic networks but can overestimate spread in cyclic networks. We apply that deterministic recursion to the full augmented reverse kernel. Consequently, the score is used to summarize approximate multistep upstream connectivity, rather than to label receiver nodes as forward cascade conduits.

\subsection{Financial network contagion, state-dependent spillovers, and diffusion-based role interpretation}
\label{subsec:network_contagion_diffusion_roles}

The need to move beyond a simple transmitter--receiver classification is also supported by the broader literature on financial network contagion, systemic risk, and diffusion-based centrality. Financial spillovers are not necessarily constant over time, because both the magnitude and direction of risk transmission may change across market conditions. Adams et al.~\cite{adams2014spillover} show that spillover effects among financial institutions are state dependent and differ across tranquil, normal, and volatile market states. This finding is closely related to the present study, which uses rolling TCI, rolling directional spillovers, and rolling diffusion measures to examine how cross-asset connectedness changes over time rather than assuming a fixed spillover structure.

The financial network literature further emphasizes that shock propagation depends on the structure of interdependence among nodes. Elliott et al.~\cite{elliott2014financial} show that financial contagion can arise through networked claims and obligations, and that the extent of failure cascades depends on the integration and diversification of the network. Similarly, Acemoglu et al.~\cite{acemoglu2015systemic} demonstrate that financial networks may exhibit a robust-yet-fragile property: greater connectivity can absorb small shocks, but may amplify large shocks once the system crosses a critical threshold. Gai and Kapadia~\cite{gai2010contagion} also show that contagion in financial networks is shaped not only by aggregate connectivity but also by idiosyncratic shocks, network structure, and market liquidity. These studies suggest that financial connectedness should be interpreted as a network-dependent propagation process rather than as a collection of isolated bilateral relationships.

This perspective is relevant for the empirical design because an asset's position cannot be inferred from direct NET spillover alone. Jackson and Pernoud~\cite{jackson2021systemic} emphasize that systemic risk can arise through direct exposures, common portfolios, fire sales, feedback effects, and changes in financial centrality. The present study does not analyze those structural channels, default contagion, or balance-sheet exposures; it uses that broader network perspective only as motivation for comparing descriptive connectedness summaries.

The importance of indirect paths is also consistent with the diffusion-centrality literature. Banerjee et al.~\cite{banerjee2013diffusion} show that the position of initially informed nodes affects subsequent diffusion of microfinance participation. In weighted directed networks, Fink et al.~\cite{fink2023centrality} propose Viral Centrality as a deterministic spread approximation. These studies motivate a multistep diagnostic, while the reverse orientation used here changes its application-specific interpretation to upstream-source connectivity.

Finally, the safe-haven literature suggests that gold, Bitcoin, and major currencies perform heterogeneous roles under different market conditions. Wang and Lee~\cite{wang2022gold} show that gold's hedging role varies across currencies and horizons. Feder-Sempach et al.~\cite{federsempach2024global} find that safe-haven properties differ across assets, markets, and crisis episodes. These findings support examining gold futures, exchange rates, and Bitcoin as non-equivalent nodes across direct and source-tracing network summaries.

Taken together, these studies justify comparing conventional NET spillovers with additional, explicitly oriented network measures. We combine VAR-GFEVD connectedness with a reverse source-tracing Markov kernel, stationary departure, and deterministic Viral Centrality. The resulting roles are conditional descriptive summaries, not distinct structural mechanisms established by the data.

\section{Data and Methodology}
\label{sec:data_methodology}

This section presents the data and the estimation pipeline in the order in
which it is applied: construction of daily returns
(Section~\ref{subsec:data_preprocessing}); VAR-GFEVD estimation of
directional spillovers (Section~\ref{subsec:var_gfevd}); transformation of
the positive net-spillover structure into a reverse source-tracing Markov
kernel and calculation of its stationary departure measure
(Section~\ref{subsec:markov_fout}); deterministic iterative evaluation of
Viral Centrality (Section~\ref{subsec:viral_centrality}); and the
rolling-window implementation together with distribution-based
diffusion-state indicators
(Section~\ref{subsec:rolling_implementation}).
Auxiliary specifications are collected in
Section~\ref{subsec:robustness_design}. All connectedness and diffusion
objects depend on the forecast horizon $H$; the argument $H$ is displayed in
definitions but suppressed in the surrounding text whenever no confusion
arises. Throughout, every measure introduced below is conditional on the
estimated VAR-GFEVD structure and is interpreted as statistical
forecast-error-variance connectedness, not as structural causality; this
caveat applies uniformly and is not repeated at each step.%

\subsection{Data, asset sets, and preprocessing}
\label{subsec:data_preprocessing}

Daily closing prices are obtained from Yahoo
Finance\footnote{\url{https://finance.yahoo.com/}} through the Python
package \texttt{yfinance}\footnote{\url{https://pypi.org/project/yfinance/}}
for the period September 18, 2014 to April 30, 2026. The baseline node set
is
\[
\mathcal{V}_0=
\{\mathrm{JPY/USD},\,\mathrm{GBP/USD},\,\mathrm{EUR/USD},\,\mathrm{CNY/USD},\,
\mathrm{Gold},\,\mathrm{BTC/USD}\},
\qquad N=6,
\]
and the auxiliary node set is
$\mathcal{V}_1=\mathcal{V}_0\cup\{\mathrm{U.S.\ Dollar\ Index}\}$ with $N=7$.
The U.S. Dollar Index enters only the auxiliary specification; the reasons
for excluding a broad dollar factor from the baseline network are given in
Section~\ref{subsec:robustness_design}. Table~\ref{tab:data_sources} reports
the tickers and quotation conventions. The raw \texttt{JPY=X} and
\texttt{CNY=X} series are quoted as local-currency units per U.S. dollar.
Before returns are calculated, each is transformed to its reciprocal so that
the analyzed series are JPY/USD and CNY/USD. A positive return therefore
denotes an increase in the displayed, possibly transformed, quotation.%

\begin{table}[!htbp]
\centering
\caption{Data sources and quotation conventions}
\label{tab:data_sources}
\scriptsize
\setlength{\tabcolsep}{5pt}
\renewcommand{\arraystretch}{1.15}
\begin{tabular}{llll}
\toprule
Asset label & Yahoo Finance ticker & Quotation used in returns & Asset set \\
\midrule
JPY/USD & \texttt{JPY=X} & USD per JPY (reciprocal) & Baseline \\
GBP/USD & \texttt{GBPUSD=X} & U.S. dollars per British pound & Baseline \\
EUR/USD & \texttt{EURUSD=X} & U.S. dollars per euro & Baseline \\
CNY/USD & \texttt{CNY=X} & USD per CNY (reciprocal) & Baseline \\
Gold futures & \texttt{GC=F} & U.S. dollars per troy ounce & Baseline \\
BTC/USD & \texttt{BTC-USD} & U.S. dollars per Bitcoin & Baseline \\
U.S. Dollar Index & \texttt{DX-Y.NYB} & Index level & Auxiliary \\
\bottomrule
\end{tabular}
\vspace{1mm}
\begin{minipage}{0.96\textwidth}
\scriptsize
\textit{Note}: \texttt{JPY=X} and \texttt{CNY=X} are downloaded as Japanese
yen and Chinese yuan per U.S. dollar, respectively. They are inverted before
return calculation and are labeled JPY/USD and CNY/USD. Thus, a positive
return means appreciation of the displayed numerator currency relative to the
U.S. dollar. The other series are used in their downloaded quotations.

\end{minipage}
\end{table}

For each asset $i\in\{1,\ldots,N\}$ the daily percentage return is%

\begin{equation}
r_{i,t}
=
100 \times \frac{P_{i,t}-P_{i,t-1}}{P_{i,t-1}},
\label{eq:return}
\end{equation}

where $P_{i,t}$ is the closing price of asset $i$ on date $t$. The vector
$\mathbf{y}_t=(r_{1,t},\ldots,r_{N,t})'$ collects the returns at time $t$,
stacked into the return matrix $\mathbf{Y}\in\mathbb{R}^{T\times N}$.

All series are aligned on the intersection of trading dates; dates with a
missing observation in any series are dropped and no interpolation is
applied. Because \texttt{GC=F} is a continuously linked nearby-futures
series, some daily return movements may reflect contract-roll adjustments
as well as changes in the underlying gold price. The resulting balanced
panel contains $T=2{,}914$ daily observations.

Table~\ref{tab:descriptive_stationarity} reports descriptive
statistics: the ADF and Phillips--Perron tests reject the unit-root null for
every return series, so the returns are suitable inputs for the stationary
VAR analysis below. The complete estimation pipeline is summarized in
Algorithm~\ref{alg:spillover_diffusion}.%

\begin{table}[!htbp]
\centering
\caption{Descriptive statistics and stationarity test results}
\label{tab:descriptive_stationarity}
\scriptsize
\setlength{\tabcolsep}{4pt}
\renewcommand{\arraystretch}{1.15}
\begin{tabular}{lcccccccc}
\toprule
Variable
& Obs.
& Mean
& Std. Dev.
& Skewness
& Kurtosis
& JB
& ADF
& PP $Z_{\rho}$ \\
\midrule
JPY/USD
& 2914 & -0.012 & 0.579 & 0.293 & 6.215
& 1296.68$^{***}$ & -13.58$^{***}$ & -2924.99$^{***}$ \\

GBP/USD
& 2914 & -0.005 & 0.582 & -0.827 & 15.917
& 20590.40$^{***}$ & -15.48$^{***}$ & -2736.90$^{***}$ \\

EUR/USD
& 2914 & -0.002 & 0.508 & 0.084 & 5.512
& 769.58$^{***}$ & -14.76$^{***}$ & -2836.42$^{***}$ \\

CNY/USD
& 2914 & -0.003 & 0.296 & 0.285 & 11.886
& 9626.63$^{***}$ & -12.59$^{***}$ & -3438.24$^{***}$ \\

Gold futures
& 2914 & 0.051 & 1.033 & -0.534 & 11.450
& 8807.94$^{***}$ & -14.59$^{***}$ & -2910.42$^{***}$ \\

BTC/USD
& 2914 & 0.264 & 4.189 & 0.032 & 9.611
& 5307.05$^{***}$ & -13.21$^{***}$ & -3067.50$^{***}$ \\

U.S. Dollar Index
& 2914 & 0.006 & 0.439 & -0.121 & 4.641
& 334.07$^{***}$ & -14.50$^{***}$ & -2772.83$^{***}$ \\
\bottomrule
\end{tabular}

\vspace{1mm}

\begin{minipage}{0.98\textwidth}
\scriptsize
\textit{Note}:
JB denotes the Jarque-Bera normality test statistic.
ADF denotes the Augmented Dickey--Fuller unit-root test statistic, and
PP $Z_{\rho}$ denotes the Phillips--Perron rho statistic.
For the JB test, statistical significance indicates rejection of the null hypothesis of normality.
For the ADF and PP tests, statistical significance indicates rejection of the unit-root null hypothesis.
$^{***}$, $^{**}$, and $^{*}$ denote significance at the 1\%, 5\%, and 10\% levels, respectively.
The ADF test was conducted with an intercept, and the lag length was selected by the Bayesian Information Criterion with a maximum lag of 20.
The PP $Z_{\rho}$ test was conducted with truncation lag 9.
\end{minipage}
\end{table}

\subsection{VAR-GFEVD connectedness}
\label{subsec:var_gfevd}

Directional return connectedness is estimated with the generalized VAR
framework of Diebold and Yilmaz~\cite{diebold2012better,diebold2014network}.
It builds on the generalized impulse-response framework of Koop et
al.~\cite{koop1996impulse} and the ordering-invariant generalized
forecast-error variance decomposition of Pesaran and
Shin~\cite{pesaran1998generalized}. The return vector follows a VAR($p$),%

\begin{equation}
\mathbf{y}_t
=
\mathbf{c}
+
\sum_{\ell=1}^{p}
\mathbf{A}_{\ell}\mathbf{y}_{t-\ell}
+
\boldsymbol{\varepsilon}_t,
\qquad
\mathbb{E}(\boldsymbol{\varepsilon}_t\mid\mathcal{F}_{t-1})=\mathbf{0},
\qquad
\mathbb{E}(\boldsymbol{\varepsilon}_t\boldsymbol{\varepsilon}_t^{\prime})
=\boldsymbol{\Sigma},
\label{eq:var_model}
\end{equation}

where $\mathbf{c}$ is the intercept, $\mathbf{A}_\ell$ are lag matrices,
and $\boldsymbol{\varepsilon}_t$ is a serially uncorrelated innovation with
covariance $\boldsymbol{\Sigma}$. The lag order minimizes the Bayesian
Information Criterion~\cite{schwarz1978estimating},%

\begin{equation}
\mathrm{BIC}(p)
=
\ln \left| \widehat{\boldsymbol{\Sigma}}_{p} \right|
+
\frac{\ln(T)}{T} k_{p},
\label{eq:bic}
\end{equation}

where $\widehat{\boldsymbol{\Sigma}}_p$ is the residual covariance under lag
length $p$, $T$ is the effective estimation-sample size used by the BIC
implementation, and $k_p$ is its corresponding system parameter count. As
Table~\ref{tab:bic_results} reports, the criterion selects $p=1$ for the
baseline six-asset network and $p=2$ for the auxiliary seven-asset network.
These orders are held fixed across all rolling windows in
Section~\ref{subsec:rolling_implementation}, so that time variation in the
estimated measures reflects changes in the spillover structure rather than
changes in the lag specification; they are model-selection results, not
evidence of structural causality.%

\begin{table}[!htbp]
\centering
\caption{BIC results for lag selection}
\label{tab:bic_results}
\scriptsize
\setlength{\tabcolsep}{4pt}
\renewcommand{\arraystretch}{1.10}
\begin{tabular}{lcccccccccc}
\toprule
\multicolumn{11}{c}{Panel A. Auxiliary asset set including the U.S. Dollar Index} \\
\midrule
Lag & 1 & 2 & 3 & 4 & 5 & 6 & 7 & 8 & 9 & 10 \\
\midrule
BIC & -7.36 & -7.50 & -7.47 & -7.39 & -7.30 & -7.19 & -7.08 & -6.96 & -6.84 & -6.74 \\
\midrule
\multicolumn{11}{c}{Panel B. Baseline asset set excluding the U.S. Dollar Index} \\
\midrule
Lag & 1 & 2 & 3 & 4 & 5 & 6 & 7 & 8 & 9 & 10 \\
\midrule
BIC & -3.98 & -3.90 & -3.81 & -3.73 & -3.65 & -3.56 & -3.49 & -3.39 & -3.30 & -3.23 \\
\bottomrule
\end{tabular}

\vspace{1mm}

\begin{minipage}{0.94\textwidth}
\small
\textit{Note}: The optimal lag length is selected by minimizing the Bayesian Information Criterion. Panel A reports the BIC values for the auxiliary asset set including the U.S. Dollar Index, while Panel B reports the results for the baseline asset set excluding the U.S. Dollar Index. Lower BIC values indicate a preferred model after accounting for both model fit and parameter complexity.
\end{minipage}
\end{table}

For a stable VAR, let
\[
\boldsymbol{\mu}
=
\left(\mathbf{I}_N-\sum_{\ell=1}^{p}\mathbf{A}_{\ell}\right)^{-1}
\mathbf{c}.
\]
Its mean-adjusted moving-average representation is
\[
\mathbf{y}_t-\boldsymbol{\mu}
=
\sum_{h=0}^{\infty}
\boldsymbol{\Phi}_h\boldsymbol{\varepsilon}_{t-h},
\qquad
\boldsymbol{\Phi}_0=\mathbf{I}_N.
\]
The $H$-step generalized forecast-error variance share associated with
innovation $j$ for affected asset $i$ is%

\begin{equation}
\theta_{ij}^{g}(H)
=
\frac{
\sigma_{jj}^{-1}
\sum_{h=0}^{H-1}
\left(
\mathbf{e}_{i}^{\prime}
\boldsymbol{\Phi}_{h}
\boldsymbol{\Sigma}
\mathbf{e}_{j}
\right)^{2}
}{
\sum_{h=0}^{H-1}
\mathbf{e}_{i}^{\prime}
\boldsymbol{\Phi}_{h}
\boldsymbol{\Sigma}
\boldsymbol{\Phi}_{h}^{\prime}
\mathbf{e}_{i}
},
\label{eq:gfevd}
\end{equation}

where $\sigma_{jj}$ is the $j$th diagonal element of $\boldsymbol{\Sigma}$
and $\mathbf{e}_i$ a selection
vector~\cite{koop1996impulse,pesaran1998generalized}. Because generalized
shares need not sum to one across $j$, each row is normalized,%

\begin{equation}
\widetilde{\theta}_{ij}^{g}(H)
=
\frac{
\theta_{ij}^{g}(H)
}{
\sum_{k=1}^{N}
\theta_{ik}^{g}(H)
},
\label{eq:normalized_gfevd}
\end{equation}

so that $\widetilde{\theta}_{ij}^{g}(H)$ is the share of the $H$-step
forecast-error variance of asset $i$ attributable to shocks in asset $j$:
rows index affected assets and columns index shock sources. The directional
measures are%

\begin{equation}
\mathrm{FROM}_{i}(H)=100\!\!\sum_{\substack{j=1 \\ j\ne i}}^{N}\!\widetilde{\theta}_{ij}^{g}(H),
\qquad
\mathrm{TO}_{i}(H)=100\!\!\sum_{\substack{j=1 \\ j\ne i}}^{N}\!\widetilde{\theta}_{ji}^{g}(H),
\qquad
\mathrm{NET}_{i}(H)=\mathrm{TO}_{i}(H)-\mathrm{FROM}_{i}(H),
\label{eq:to_from_net}
\end{equation}

and the total connectedness index is%

\begin{equation}
\mathrm{TCI}(H)
=
100
\times
\frac{
\sum_{i=1}^{N}
\sum_{\substack{j=1 \\ j \ne i}}^{N}
\widetilde{\theta}_{ij}^{g}(H)
}{N}.
\label{eq:tci}
\end{equation}

A positive $\mathrm{NET}_i(H)$ identifies asset $i$ as a net transmitter and
a negative value as a net receiver; the TCI is the average cross-asset
variance contribution and serves as the system-wide connectedness level.%

\subsection{From net spillovers to a reverse source-tracing Markov kernel}
\label{subsec:signed_network}
\label{subsec:markov_fout}

The normalized GFEVD matrix is next reinterpreted as a weighted directed
network in which assets are nodes and net variance contributions are
directed links~\cite{diebold2014network}. The pairwise net spillover from
asset $i$ to asset $j$ is%

\begin{equation}
\mathcal{N}_{ij}(H)
=
\widetilde{\theta}_{ji}^{g}(H)
-
\widetilde{\theta}_{ij}^{g}(H),
\label{eq:pairwise_net_spillover}
\end{equation}

and its positive part%

\begin{equation}
\mathcal{A}_{ij}^{+}(H)
=
\max
\left\{
\mathcal{N}_{ij}(H),0
\right\}
\label{eq:positive_net_spillover}
\end{equation}

retains, for each ordered pair, only the dominant direction of net
transmission. The matrix $\mathcal{A}^{+}(H)$ is therefore a directed
dominance network of pairwise net spillovers, not the gross spillover
network. The following identity shows that this reduction preserves exactly
the directional information summarized by NET.%

\begin{proposition}[Signed net spillovers and network imbalance]
\label{prop:net_imbalance}
Let \(\widetilde{\boldsymbol{\Theta}}^g(H)=[\widetilde{\theta}_{ij}^g(H)]_{i,j=1}^{N}\) be the row-normalized GFEVD matrix, and let \(\mathcal{N}(H)=[\mathcal{N}_{ij}(H)]_{i,j=1}^{N}\) be defined by Eq.~\eqref{eq:pairwise_net_spillover}. Then \(\mathcal{N}(H)\) is skew-symmetric and
\begin{equation}
\mathrm{NET}_{i}(H)
=
100\sum_{j=1}^{N}\mathcal{N}_{ij}(H).
\label{eq:net_as_signed_sum}
\end{equation}
Moreover, if \(\mathcal{A}_{ij}^{+}(H)=\max\{\mathcal{N}_{ij}(H),0\}\), then
\begin{equation}
\mathrm{NET}_{i}(H)
=
100
\left[
\sum_{j=1}^{N}\mathcal{A}_{ij}^{+}(H)
-
\sum_{j=1}^{N}\mathcal{A}_{ji}^{+}(H)
\right].
\label{eq:net_as_positive_imbalance}
\end{equation}
\end{proposition}

\begin{proof}
Skew-symmetry follows directly from
\[
\mathcal{N}_{ji}(H)
=
\widetilde{\theta}_{ij}^{g}(H)
-
\widetilde{\theta}_{ji}^{g}(H)
=
-\mathcal{N}_{ij}(H).
\]
Furthermore,
\[
\sum_{j=1}^{N}\mathcal{N}_{ij}(H)
=
\sum_{j=1}^{N}\widetilde{\theta}_{ji}^{g}(H)
-
\sum_{j=1}^{N}\widetilde{\theta}_{ij}^{g}(H).
\]
The diagonal terms cancel, so the right-hand side equals
\[
\sum_{j\ne i}\widetilde{\theta}_{ji}^{g}(H)
-
\sum_{j\ne i}\widetilde{\theta}_{ij}^{g}(H)
=
\frac{\mathrm{TO}_{i}(H)-\mathrm{FROM}_{i}(H)}{100}
=
\frac{\mathrm{NET}_{i}(H)}{100}.
\]
Finally, skew-symmetry implies
\(\mathcal{N}_{ij}(H)=\mathcal{A}_{ij}^{+}(H)-\mathcal{A}_{ji}^{+}(H)\). Summing over \(j\) gives Eq.~\eqref{eq:net_as_positive_imbalance}.
\end{proof}

For each asset define the incoming and outgoing positive net-spillover
strengths $s_i^{\mathrm{in}}(H)=\sum_{k}\mathcal{A}^{+}_{ki}(H)$ and
$s_i^{\mathrm{out}}(H)=\sum_{k}\mathcal{A}^{+}_{ik}(H)$; by
Proposition~\ref{prop:net_imbalance},
$\mathrm{NET}_i(H)=100\,[s_i^{\mathrm{out}}(H)-s_i^{\mathrm{in}}(H)]$, so a
net transmitter is an asset whose outgoing strength exceeds its incoming
strength. A node with
$s_i^{\mathrm{in}}(H)=0$ is called a \emph{pure net source} and a node with
$s_i^{\mathrm{out}}(H)=0$ a \emph{pure net sink}.

We define a random walk on the edge-reversed dominance network. Starting
from asset $i$, the walk moves toward assets that are upstream
net-spillover suppliers of $i$: the step $i\to j$ is weighted by the net
spillover that $j$ transmits into $i$. Thus, the walk follows the reverse
orientation of the transmission edges. The resulting kernel is a
source-tracing device, not a forward shock-propagation kernel. It is defined
as%

\begin{equation}
P_{ij}(H)
=
\begin{cases}
\dfrac{\mathcal{A}_{ji}^{+}(H)}
{s_i^{\mathrm{in}}(H)},
&
\text{if } s_i^{\mathrm{in}}(H)>0, \\[3mm]
\dfrac{1}{N},
&
\text{if } s_i^{\mathrm{in}}(H)=0.
\end{cases}
\label{eq:transition_matrix}
\end{equation}

Each row of $\mathbf{P}(H)=[P_{ij}(H)]_{i,j=1}^{N}$ is nonnegative and sums
to one, so $\mathbf{P}(H)$ is a row-stochastic Markov
kernel~\cite{lovasz1993random,masuda2017random,seabrook2023tutorial}. For
$s_i^{\mathrm{in}}(H)>0$, $P_{ij}(H)$ is the conditional share of upstream
supplier $j$ among the positive net-spillover suppliers of asset $i$. Row
normalization preserves these relative weights but removes the total
incoming magnitude $s_i^{\mathrm{in}}(H)$. For a pure net source, the
uniform row is a restart closure introduced only to keep the reverse kernel
row-stochastic. It adds no directional preference, but it creates augmented
transitions, including a self-transition, that are absent from
$\mathcal{A}^{+}(H)$. Consequently, all Markov and Viral Centrality results
computed from $\mathbf{P}(H)$ are conditional on row normalization and this
restart closure.
Absolute direct transmission remains measured by TO,
FROM, and NET.%

\begin{proposition}[Pure-sink bound under the uniform-restart closure]
\label{prop:source_concentration}
Let \(U=\{k:\,s_k^{\mathrm{in}}(H)=0\}\) denote the set of pure net sources, and let \(\boldsymbol{\pi}(H)\) be any stationary distribution of \(\mathbf{P}(H)\). If asset \(r\) is a pure net sink, that is, \(s_r^{\mathrm{out}}(H)=0\) (equivalently \(\mathcal{A}_{rk}^{+}(H)=0\) for all \(k\)), then
\begin{equation}
\pi_r(H)=\frac{1}{N}\sum_{k\in U}\pi_k(H)\le\frac{1}{N}.
\label{eq:sink_bound}
\end{equation}
Consequently, a pure net sink receives at most $1/N$ stationary mass under
the selected reverse orientation and uniform-restart closure. This bound
does not imply that every net source has high stationary mass, that mass
generally concentrates on sources, or that the ranking is invariant to an
alternative closure.
\end{proposition}

\begin{proof}
Write \(\pi_r(H)=\sum_{i=1}^{N}\pi_i(H)P_{ir}(H)\). For any \(i\notin U\) we have \(s_i^{\mathrm{in}}(H)>0\) and \(P_{ir}(H)=\mathcal{A}_{ri}^{+}(H)/s_i^{\mathrm{in}}(H)\); since \(r\) is a pure net sink, \(\mathcal{A}_{rk}^{+}(H)=0\) for all \(k\), hence \(\mathcal{A}_{ri}^{+}(H)=0\) and \(P_{ir}(H)=0\). For \(i\in U\) the row is uniform, so \(P_{ir}(H)=1/N\). Therefore
\[
\pi_r(H)=\frac{1}{N}\sum_{i\in U}\pi_i(H)\le\frac{1}{N}\sum_{i=1}^{N}\pi_i(H)=\frac{1}{N}.
\]
\end{proof}

Being finite and row-stochastic, $\mathbf{P}(H)$ admits at least one
stationary distribution $\boldsymbol{\pi}(H)$,%

\begin{equation}
\boldsymbol{\pi}^{\prime}(H)
=
\boldsymbol{\pi}^{\prime}(H)\mathbf{P}(H),
\qquad
\sum_{i=1}^{N}\pi_i(H)=1,
\label{eq:stationary_distribution}
\end{equation}

which is unique when the induced chain is irreducible; empirically,
$\boldsymbol{\pi}(H)$ is computed as a normalized nonnegative left
eigenvector of $\mathbf{P}(H)$ associated with eigenvalue one. The quantity
$\pi_i(H)$ is the long-run occupation probability of asset $i$ under the
reverse source-tracing walk. A large value means that the node is frequently
reached while recursively tracing upstream suppliers, conditional on row
normalization and the restart closure. It does not by itself imply large
outgoing spillover magnitude. The stationary departure measure of asset $i$
is defined as%

\begin{equation}
F_{\mathrm{out},i}(H)
:=
\pi_i(H)
\sum_{\substack{j=1 \\ j\ne i}}^{N}
P_{ij}(H)
=
\pi_i(H)
\left(
1-P_{ii}(H)
\right).
\label{eq:fout}
\end{equation}

\begin{remark}
\label{rem:fout_vs_pi}
Since $\mathcal{A}^{+}_{ii}(H)=0$, one has $P_{ii}(H)=0$ and hence
$F_{\mathrm{out},i}(H)=\pi_i(H)$ for every non-source asset, whereas a pure
net source has $P_{ii}(H)=1/N$ and
$F_{\mathrm{out},i}(H)=(1-1/N)\,\pi_i(H)$. Thus $F_{\mathrm{out}}$ coincides
with $\boldsymbol{\pi}$ up to a mild $(1-1/N)$ deflation applied only to
pure-source nodes; it is reported alongside $\boldsymbol{\pi}$ because
Proposition~\ref{prop:fout_flux} endows it with an exact stationary-flux
interpretation.
\end{remark}

\begin{proposition}[Stationary departure probability in the augmented source-tracing chain]
\label{prop:fout_flux}
Let \(\{X_t\}_{t\ge 0}\) be a finite Markov chain with transition matrix \(\mathbf{P}(H)\) and stationary distribution \(\boldsymbol{\pi}(H)\). Then
\begin{equation}
F_{\mathrm{out},i}(H)
=
\Pr_{\boldsymbol{\pi}}
\left(
X_t=i,\;X_{t+1}\ne i
\right).
\label{eq:fout_stationary_flux}
\end{equation}
Moreover,
\begin{equation}
\sum_{i=1}^{N}F_{\mathrm{out},i}(H)
=
\Pr_{\boldsymbol{\pi}}(X_{t+1}\ne X_t),
\label{eq:fout_total_flux}
\end{equation}
which is the total stationary probability of moving across distinct asset nodes in one transition.
\end{proposition}

\begin{proof}
Under stationarity, \(\Pr_{\boldsymbol{\pi}}(X_t=i)=\pi_i(H)\). Conditional on \(X_t=i\), the probability of leaving node \(i\) is
\[
\Pr(X_{t+1}\ne i\mid X_t=i)
=
\sum_{j\ne i}P_{ij}(H)
=
1-P_{ii}(H).
\]
Therefore,
\[
\Pr_{\boldsymbol{\pi}}(X_t=i,X_{t+1}\ne i)
=
\pi_i(H)(1-P_{ii}(H))
=
F_{\mathrm{out},i}(H).
\]
Summing this identity over all \(i\) gives Eq.~\eqref{eq:fout_total_flux}.
\end{proof}

Proposition~\ref{prop:fout_flux} shows that $F_{\mathrm{out},i}(H)$ is the
stationary frequency with which the augmented reverse source-tracing chain
occupies asset $i$ and then moves to a distinct asset. The subscript
``out'' refers to departure from a node of this reverse chain; it does not
denote forward economic shock outflow from the asset. For non-source nodes,
$F_{\mathrm{out},i}(H)=\pi_i(H)$ exactly, while a pure-source value also
reflects the selected restart closure.%

\subsection{Deterministic Viral Centrality on the source-tracing kernel}
\label{subsec:viral_centrality}

Viral Centrality (VC) is used as an auxiliary measure of multistep
connectivity on the weighted directed network. Adapting the deterministic
probability-propagation structure of Fink et
al.~\cite{fink2023centrality}, VC is designed to approximate the expected reach of an
independent-cascade process; it is not obtained by simulating random cascade
realizations. We apply the recursion directly to the full augmented kernel
$\mathbf{P}(H)$ used in the reported computation. Thus, a pure-source
restart row and any diagonal transition created by that closure are retained
in the VC input, and the resulting score is conditional on that convention.

We encode the selected seed as already activated and therefore no longer
susceptible at initialization. For a seed asset $s$, let
$S_i^{(t)}(s)$ denote the algorithm's approximate probability that asset
$i$ remains susceptible after iteration $t$, and let $A_i^{(t)}(s)$ denote
its approximate probability of new activation at iteration $t$. The
initialization is

\begin{equation}
S_i^{(0)}(s)
=
1-\mathbf{1}\{i=s\},
\qquad
A_i^{(0)}(s)
=
\mathbf{1}\{i=s\}.
\label{eq:vc_initialization}
\end{equation}

In the implementation, $\mathbf{P}(H)$ is stored in row-stochastic form,
where $P_{ij}(H)$ denotes the transition weight from node $i$ to node $j$.
For the probability-update equations below, the corresponding incoming
weight to node $i$ from node $j$ is therefore written as $P_{ji}(H)$.
At iteration $t$, the probability that asset $i$ receives no activation
from the nodes activated in the preceding iteration is

\begin{equation}
B_i^{(t)}(s)
=
\prod_{j=1}^{N}
\left[
1-P_{ji}(H)A_j^{(t-1)}(s)
\right].
\label{eq:vc_no_activation}
\end{equation}

The newly activated and remaining susceptible probabilities are then
updated recursively as

\begin{equation}
A_i^{(t)}(s)
=
\left[
1-B_i^{(t)}(s)
\right]
S_i^{(t-1)}(s),
\label{eq:vc_activation_update}
\end{equation}

and

\begin{equation}
S_i^{(t)}(s)
=
B_i^{(t)}(s)
S_i^{(t-1)}(s).
\label{eq:vc_susceptible_update}
\end{equation}

The corresponding cumulative approximate activation mass at iteration $t$ is

\begin{equation}
C_s^{(t)}(H)
=
\sum_{i=1}^{N}
\left[
1-S_i^{(t)}(s)
\right].
\label{eq:vc_total_activation}
\end{equation}

Iterations continue until the change in total activation mass satisfies

\begin{equation}
\left|
C_s^{(t)}(H)-C_s^{(t-1)}(H)
\right|
<
\varepsilon,
\label{eq:vc_convergence}
\end{equation}

or the maximum number of iterations is reached. In the implementation,
$\varepsilon=10^{-6}$ and the maximum number of iterations is 200.
The Viral Centrality score of seed asset $s$ is then

\begin{equation}
\mathrm{VC}_s(H)
=
C_s^{(*)}(H)-1
=
\sum_{i=1}^{N}
\left[
1-S_i^{(*)}(s)
\right]
-1,
\label{eq:viral_centrality}
\end{equation}

where $*$ denotes the converged iteration and the subtraction of one
excludes the seed asset itself.

The calculation is deterministic conditional on $\mathbf{P}(H)$:
no random cascade realizations are generated and no Monte Carlo sampling
error is involved. Accordingly, the reported VC values should be
interpreted as iterative Viral Centrality approximations implied by the
augmented transition network rather than as exact Monte Carlo expectations
or estimates of realized contagion. The Fink et al. recursion is exact on
acyclic networks but may overestimate independent-cascade reach when cycles
are present. Because $\mathbf{P}(H)$ is row-stochastic and hence has spectral
radius one, we use VC as an algorithm-defined topology score and do not claim
that it is numerically identical to independent-cascade expected spread.

Because each non-restart step of $\mathbf{P}(H)$ is oriented toward dominant
net-spillover sources, a high VC value indicates broad conditional
connectivity to upstream source nodes within the constructed reverse
topology. For a pure net source, however, VC also reflects the artificial
restart row and cannot be interpreted as literal upstream reach. In all
cases, VC is a descriptive network score rather than evidence of structural
or realized financial contagion.

\subsection{Rolling implementation and diffusion-state indicators}
\label{subsec:rolling_implementation}
\label{subsec:diffusion_states}

Time variation is captured by re-estimating the entire pipeline on rolling
windows of $W=100$ trading days,
$\mathcal{W}_t=\{\mathbf{y}_{t-W+1},\ldots,\mathbf{y}_t\}$, which yields
$T_W=T-W+1=2{,}815$ endpoint estimates of
$\mathrm{TCI}_t$, $\mathrm{TO}_{i,t}$, $\mathrm{FROM}_{i,t}$,
$\mathrm{NET}_{i,t}$, $F_{\mathrm{out},i,t}$, and $\mathrm{VC}_{i,t}$.

The rolling design is adopted in preference to time-varying-parameter (TVP)
VAR connectedness in the spirit of Antonakakis et
al.~\cite{antonakakis2020refined}, and this choice is deliberate. Rolling
estimation introduces no additional hyperparameters --- no forgetting
factors, decay rates, or prior tightness --- so each windowed estimate is an
ordinary least-squares object that is transparent and exactly reproducible
window by window. More importantly, the contribution of this paper lies in
the spillover-to-diffusion layer of Sections~\ref{subsec:markov_fout}
and~\ref{subsec:viral_centrality}, which operates on any estimated GFEVD
matrix: the reverse-kernel construction, stationary departure measure, and Viral
Centrality are orthogonal to the choice of time-varying estimator and could
equally be applied on top of TVP-VAR-based connectedness estimates.
Examining the framework under smoother TVP-VAR dynamics is left as a
natural extension (Section~\ref{sec:conclusion}).%

The baseline specification uses a forecast horizon of $H=10$ and a rolling
window of $W=100$.
Variation in the rolling-window size, $W\in\{52,100,150\}$, is reported in
Section~\ref{supsec:rolling_window_robustness} of the Supplementary Material.

Forecast-horizon variation, $H\in\{5,10,15\}$, is examined jointly with the
alternative window sizes in the PageRank robustness grid
(Table~\ref{tabS:stationary_pagerank_robustness}) and in the
diffusion-state robustness summary (Table~\ref{tabS:eds_lds_robustness_grid})
reported in the Supplementary Material.

Full-sample and rolling
estimates are distinguished notationally: for a generic asset-level measure
$M_i$,
$M_i^{\mathrm{full}}(H)$ denotes the full-sample estimate,
$M_{i,t}(H,W)$ the rolling estimate at endpoint $t$, and%

\begin{equation}
\overline{M}_{i}(H,W)
=
\frac{1}{T_W}
\sum_{t=1}^{T_W}
M_{i,t}(H,W)
\label{eq:rolling_average_measure}
\end{equation}

its time average. Tables based on rolling output report
$\overline{M}_i(H,W)$, whereas full-sample scatter plots use
$M_i^{\mathrm{full}}(H)$. Algorithm~\ref{alg:spillover_diffusion} summarizes
the complete procedure.%

To describe when individual diffusion roles become unusually strong or weak,
distribution-based state indicators are constructed from the rolling series
of the two diffusion dimensions,%

\begin{equation}
X_{i,t}
\in
\left\{
\mathrm{VC}_{i,t},
F_{\mathrm{out},i,t}
\right\},
\qquad
t=1,\ldots,T_W.
\label{eq:diffusion_dimensions}
\end{equation}

For each asset $i$ and each measure $X$, the asset-specific bounds are the
empirical quantiles of the full rolling sample,%

\begin{equation}
U_i^{X}
=
Q_{0.975}
\left(
\left\{
X_{i,t}
\right\}_{t=1}^{T_W}
\right),
\qquad
L_i^{X}
=
Q_{0.025}
\left(
\left\{
X_{i,t}
\right\}_{t=1}^{T_W}
\right),
\label{eq:diffusion_bounds}
\end{equation}

\begin{algorithm}[H]
\caption{Spillover-to-diffusion estimation procedure}
\label{alg:spillover_diffusion}
\KwIn{Return matrix \(\mathbf{Y}\), node set \(\mathcal{V}\), VAR lag \(p\),
forecast horizon \(H\), rolling-window size \(W\), VC convergence tolerance
\(\varepsilon\), and maximum number of VC iterations \(T_{\max}\).}
\KwOut{TCI, TO, FROM, NET, \(F_{\mathrm{out}}\), VC, and diffusion-state indicators.}
\For{the full sample or each rolling window}{
Estimate the VAR(\(p\)) model and compute the normalized GFEVD matrix \(\widetilde{\boldsymbol{\Theta}}^g(H)\)\;
Compute TO, FROM, NET, and TCI from \(\widetilde{\boldsymbol{\Theta}}^g(H)\)\;
Construct the signed pairwise net-spillover matrix \(\mathcal{N}(H)\)\;
Construct the positive net-spillover matrix \(\mathcal{A}^{+}(H)\)\;
Form the Markov transition kernel \(\mathbf{P}(H)\) from \(\mathcal{A}^{+}(H)\) by incoming-strength normalization (Eq.~\eqref{eq:transition_matrix})\;
Compute the stationary distribution \(\boldsymbol{\pi}(H)\) and stationary departure measure \(F_{\mathrm{out},i}(H)\)\;
Compute Viral Centrality \(\mathrm{VC}_i(H)\) deterministically using
Eqs.~\eqref{eq:vc_initialization}--\eqref{eq:viral_centrality}, with
\(\varepsilon=10^{-6}\) and \(T_{\max}=200\)\;
}
Construct diffusion-state indicators from the rolling series of \(F_{\mathrm{out}}\) and VC\;
\end{algorithm}

so that $[L_i^X,U_i^X]$ is a central $95\%$ empirical band, with approximately $2.5\%$ of observations in each tail.
Because the bounds use full-sample quantiles, all indicators below are in-sample
descriptive classifications of each asset relative to its own rolling
distribution; they are neither real-time signals nor externally validated
crisis dates. Upper- and lower-tail exceedances are%

\begin{equation}
I_{i,t}^{X,+}
=
\mathbf{1}
\left(
X_{i,t}
>
U_i^{X}
\right),
\qquad
I_{i,t}^{X,-}
=
\mathbf{1}
\left(
X_{i,t}
<
L_i^{X}
\right),
\label{eq:upper_indicator}
\end{equation}

and the joint and relaxed diffusion states are%

\begin{equation}
\mathrm{EDS}_{i,t}=I_{i,t}^{\mathrm{VC},+}\,I_{i,t}^{F_{\mathrm{out}},+},
\qquad
\mathrm{LDS}_{i,t}=I_{i,t}^{\mathrm{VC},-}\,I_{i,t}^{F_{\mathrm{out}},-},
\qquad
\mathrm{WeakEDS}_{i,t}=\max\bigl\{I_{i,t}^{\mathrm{VC},+},\,I_{i,t}^{F_{\mathrm{out}},+}\bigr\}.
\label{eq:eds}
\end{equation}

The elevated diffusion state EDS requires simultaneous upper-tail behavior
in \emph{both} diffusion dimensions and is intentionally conservative; the
low diffusion state LDS is its lower-tail counterpart; WeakEDS flags an
upper-tail event in at least one dimension and captures episodic behavior
that the joint criterion misses. A continuous intensity score complements
the binary indicators: standardizing each rolling series by its asset-specific
time-series mean $\overline{X}_i$ and standard deviation
$\widehat{\sigma}_{X,i}$ (written $\widehat{\sigma}$ to avoid collision with
the strength notation $s_i^{\mathrm{in}}$, $s_i^{\mathrm{out}}$),%

\begin{equation}
Z_{i,t}^{X}
=
\frac{X_{i,t}-\overline{X}_{i}}{\widehat{\sigma}_{X,i}},
\quad X\in\{\mathrm{VC},F_{\mathrm{out}}\},
\qquad
S_{i,t}
=
\frac{Z_{i,t}^{\mathrm{VC}}+Z_{i,t}^{F_{\mathrm{out}}}}{2},
\label{eq:diffusion_score}
\end{equation}

so that a high $S_{i,t}$ indicates a high combined standardized level
of the two diffusion diagnostics for asset $i$ relative to its own history,
not return performance; it does not require both diagnostics to be
individually high. Simultaneous upper-tail behavior is captured separately
by $\mathrm{EDS}_{i,t}$.
High and low
score states are defined analogously, but use the 95th and 5th empirical
quantiles of the score distribution,%

\begin{equation}
I_{i,t}^{S,+}
=
\mathbf{1}
\left(
S_{i,t}
>
U_i^{S}
\right),
\qquad
I_{i,t}^{S,-}
=
\mathbf{1}
\left(
S_{i,t}
<
L_i^{S}
\right),
\label{eq:score_indicators}
\end{equation}

with $U_i^{S}$ and $L_i^{S}$ denoting the 95th and 5th empirical
quantiles of $\{S_{i,t}\}_{t=1}^{T_W}$, respectively. Thus, the VC and
$F_{\mathrm{out}}$ indicators use 2.5th--97.5th percentile bounds,
whereas the score-based indicators use 5th--95th percentile bounds.
In the empirical figures, $I^{\mathrm{VC},+}$,
$I^{F_{\mathrm{out}},+}$, $\mathrm{WeakEDS}$, $I^{S,+}$, and
$\mathrm{EDS}$ are labeled ``VC upper exceedance'',
``$F_{\mathrm{out}}$ upper exceedance'', ``Weak EDS'',
``EDS\_score\_high'', and ``Joint EDS'', respectively.

State frequencies are computed as
$\widehat{p}_i^{\mathrm{EDS}}
=T_W^{-1}\sum_{t=1}^{T_W}\mathrm{EDS}_{i,t}$,
and analogously for the other indicators.

Finally, the calibration of the bounds is checked with the coverage
probability
\begin{equation}
\mathrm{CP}_{i}^{X}
=
\frac{1}{T_W}
\sum_{t=1}^{T_W}
\mathbf{1}
\left(
L_i^{X}
\leq
X_{i,t}
\leq
U_i^{X}
\right),
\qquad
X_{i,t}
\in
\left\{
\mathrm{VC}_{i,t},
F_{\mathrm{out},i,t}
\right\}.
\label{eq:coverage_probability}
\end{equation}

This coverage probability is used only as a consistency check for the
threshold computation. Since the bounds are defined from the in-sample
2.5th--97.5th percentile band, the coverage is expected to be approximately
95\% by construction, up to ties and discreteness. It is not interpreted as
evidence of predictive calibration or crisis-forecasting performance.

\subsection{Auxiliary specifications: dollar factor and PageRank regularization}
\label{subsec:robustness_design}

Two auxiliary designs assess robustness beyond the horizon and window
variations described above. First, the U.S. Dollar Index is added as a
seventh node. The baseline excludes it because bilateral dollar exchange
rates share a common dollar-related risk
component~\cite{lustig2011common,verdelhan2018share}: entering the index as
an explicit node mixes asset-specific spillover channels with this common
factor and allows a broad dollar factor to dominate the estimated topology.
The seven-node network (BIC-selected VAR(2)) is therefore used only to check
whether the qualitative conclusion of the baseline analysis --- that direct
transmission, stationary source tracing, and multistep upstream connectivity
are distinct role dimensions --- survives when a broad dollar factor is
modeled explicitly.

Second, the stationary distribution is replaced by a PageRank-type
score~\cite{brin1998anatomy,page1999pagerank}, which regularizes the random
walk with uniform teleportation~\cite{masuda2017random}:%

\begin{equation}
\mathbf{G}_{\alpha}(H)
=
\alpha \mathbf{P}(H)
+
(1-\alpha)\frac{1}{N}\mathbf{1}\mathbf{1}^{\prime},
\qquad 0<\alpha<1,
\label{eq:pagerank_transition}
\end{equation}

where $\alpha$ is the damping factor and
$(1/N)\mathbf{1}\mathbf{1}^{\prime}$ the uniform teleportation matrix.
Since $\mathbf{G}_\alpha(H)$ is strictly positive for $\alpha<1$, the
regularized chain is irreducible and aperiodic, so its stationary
distribution --- the PageRank vector --- exists and is unique:%

\begin{equation}
PR^{\prime}(H)
=
PR^{\prime}(H)\,\mathbf{G}_{\alpha}(H),
\qquad
\sum_{i=1}^{N} PR_i(H)=1.
\label{eq:pagerank_vector}
\end{equation}

For direct comparability with
$F_{\mathrm{out},i}(H)=\pi_i(H)(1-P_{ii}(H))$, we define a hybrid
PageRank-weighted original-kernel departure score,%

\begin{equation}
PR_{\mathrm{out},i}(H)
=
PR_i(H)\left(1-P_{ii}(H)\right).
\label{eq:pagerank_outflow}
\end{equation}

The quantity in Eq.~\eqref{eq:pagerank_outflow} weights the original-kernel
departure probability by PageRank occupancy. It is not the stationary
off-diagonal flux of the teleportation chain $\mathbf{G}_{\alpha}(H)$, which
would instead be $PR_i(H)[1-G_{\alpha,ii}(H)]$. The robustness check uses the conventional damping value
$\alpha=0.85$~\cite{page1999pagerank,masuda2017random} and asks whether the
stationary-departure ranking reported in Section~\ref{sec:results} is an
artifact of the exact stationary distribution of $\mathbf{P}(H)$ or remains
stable under teleportation-regularized diffusion weighting.%

\section{Results}
\label{sec:results}

\subsection{Preliminary results and baseline specification}
 
This study examines the spillover structure and diffusion roles in a six asset international financial network that excludes the U.S. Dollar Index from the baseline specification. The assets included in the baseline network are JPY/USD, GBP/USD, EUR/USD, CNY/USD, gold futures, and BTC/USD. This specification is designed to analyze the direct and indirect transmission structure formed among major exchange rates, gold, and Bitcoin without treating the broad dollar factor as an explicit network node.
 
As reported in Table~\ref{tab:descriptive_stationarity}, the ADF and Phillips--Perron tests reject the unit root null hypothesis for all return series included in the analysis. This indicates that the daily return series are stationary and suitable for VAR GFEVD based connectedness analysis. In addition, the BIC based lag selection reported in Table~\ref{tab:bic_results} selects a VAR(1) specification for the baseline six asset network. This lag choice provides a parsimonious representation of short run return dynamics, but it should be interpreted as a model selection result rather than as evidence of structural causality among individual assets.
 
Based on this baseline specification, the empirical analysis proceeds in several steps. First, the GFEVD-based TO, FROM, NET, and TCI measures are used to examine direct spillover connectedness. Second, the positive net-spillover matrix is transformed into a reverse source-tracing Markov kernel, from which the stationary departure measure is calculated. Third, Viral Centrality is computed deterministically as an approximate multistep-connectivity score on the augmented reverse kernel. Finally, distribution-based state indicators are constructed from rolling \(F_{\mathrm{out}}\) and rolling VC to identify periods in which each asset exhibits unusually strong or weak values relative to its own empirical distribution.
 
\subsection{Time varying connectedness}
 
Figure~\ref{fig:rolling_tci_window100} presents the rolling Total Connectedness Index (TCI) for the international financial asset network excluding the dollar index. The TCI is the cross-sectional average, across assets, of the normalized forecast-error variance percentage associated with innovations in the other assets. A higher TCI therefore indicates stronger system-wide connectedness. The figure shows that the network among major exchange rates, gold futures, and Bitcoin is not static, but evolves over time.
 
In the main analysis, the results based on a rolling window size of 100 days are presented as the representative setting. Additional results using window sizes of 52 and 150 days are reported in Section~\ref{supsec:rolling_window_robustness} of the Supplementary Material (Figure~\ref{fig:appendix_rolling_tci_window_comparison}) as robustness checks. These additional results examine whether the observed connectedness pattern is sensitive to the choice of rolling window length. Shorter windows are expected to capture temporary changes in connectedness more sensitively, whereas longer windows smooth out short run fluctuations and reflect more persistent network structures.
 
\begin{figure}[htbp]
\centering
\includegraphics[width=0.65\textwidth]{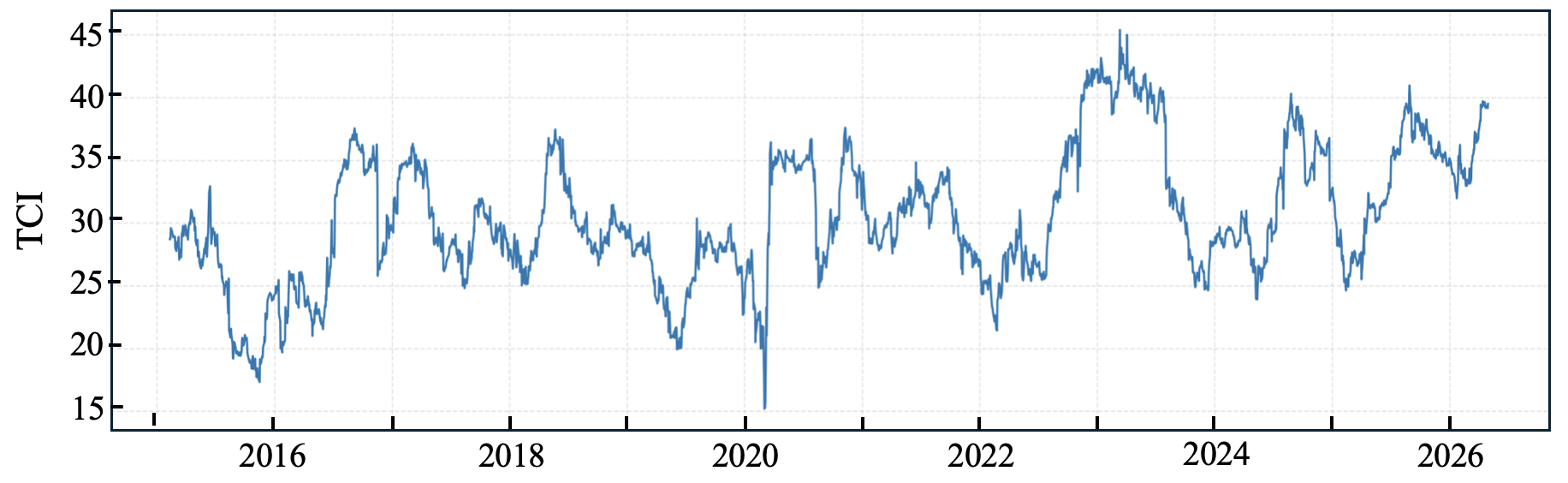}
\caption[Rolling TCI under the baseline window size]{
Rolling Total Connectedness Index (TCI) for the network excluding the dollar index. The figure reports the time varying system wide connectedness under the baseline setting of window size 100 and forecast horizon $H=10$. Higher TCI values indicate stronger overall spillover connectedness among the assets.
}
\label{fig:rolling_tci_window100}
\end{figure}
 
Under the representative setting of window size 100 and forecast horizon $H=10$, the rolling TCI has an average value of 30.44, a minimum of 15.20, a maximum of 45.38, and a standard deviation of 5.25. The maximum value is observed in the rolling window spanning from October 18, 2022 to March 13, 2023. These results indicate that, even after excluding the dollar index, the exchange rate, gold, and Bitcoin markets remain linked through a time-varying spillover structure. An average rolling TCI of approximately 30 means that, after averaging across assets and rolling windows, about 30\% of normalized forecast-error variance is associated with innovations in the other assets; it does not imply a 30\% contribution for every individual asset.
 
Several major financial episodes provide useful context for interpreting the observed TCI dynamics. The decline and recovery around 2015--2016 overlaps with the Chinese yuan devaluation and heightened concerns over China's financial markets~\cite{wang2016net,lai2019yuan}. Around the COVID--19 shock in March 2020, the TCI dropped to one of its lowest levels and then recovered rapidly, which is consistent with evidence of pandemic related financial contagion and increased global market risk~\cite{akhtaruzzaman2021financial,zhang2020financial}. The increase in early to mid 2022 coincides with the Russia--Ukraine war, energy price shocks, and heightened geopolitical uncertainty, which generated volatility spillovers across currency, commodity, stock, and energy markets~\cite{u2024russia}. The strongest connectedness episode, from late 2022 to March 2023, overlaps with aggressive U.S. monetary tightening, changing interest rate expectations, and banking sector stress associated with the collapse of Silicon Valley Bank~\cite{azmi2023svb}.
 
These event based interpretations should not be read as direct causal evidence, because the rolling TCI is computed over a 100 day window and captures system wide connectedness rather than the effect of a single event on a specific date. Instead, these episodes provide financial market context for understanding when and why the network may become more tightly connected. From an econophysics perspective, the rolling TCI suggests that the exchange rate, gold, and Bitcoin network undergoes time varying reorganization, with stronger shock propagation during periods of global financial stress.
 
Overall, the rolling TCI results show that cross asset connectedness remains economically and structurally meaningful even when the dollar index is not included as an explicit network node. Periods of elevated connectedness can be interpreted as regimes in which exchange rate risk, safe haven demand, hedging incentives, and global liquidity conditions interact more strongly within a connected financial network. This supports the central objective of this study, which is to examine not only direct spillovers but also the diffusion roles of individual assets in a time varying international financial network.

\subsection{Direct spillover roles}
 
Figure~\ref{fig:net_spillover} reports the net pairwise spillover network for the baseline six asset system excluding the dollar index. The figure provides a directed weighted representation of the asymmetric transmission structure among major exchange rates, gold futures, and Bitcoin. Positive net pairwise spillovers define the dominant direction of shock transmission between asset pairs.
 
The most notable feature is the central position of gold futures as a clear net transmitting node. Gold futures form outward transmission links mainly toward the major exchange rate variables, indicating that the gold futures node occupies a central position in the directed spillover topology even when the broad dollar factor is not included as an explicit node. This pattern suggests that gold futures do not merely exhibit isolated bilateral linkages, but occupy a directional position through which shocks are transmitted outward within the financial network.
 
\begin{figure}[H]
\centering
\includegraphics[width=0.40\textwidth]{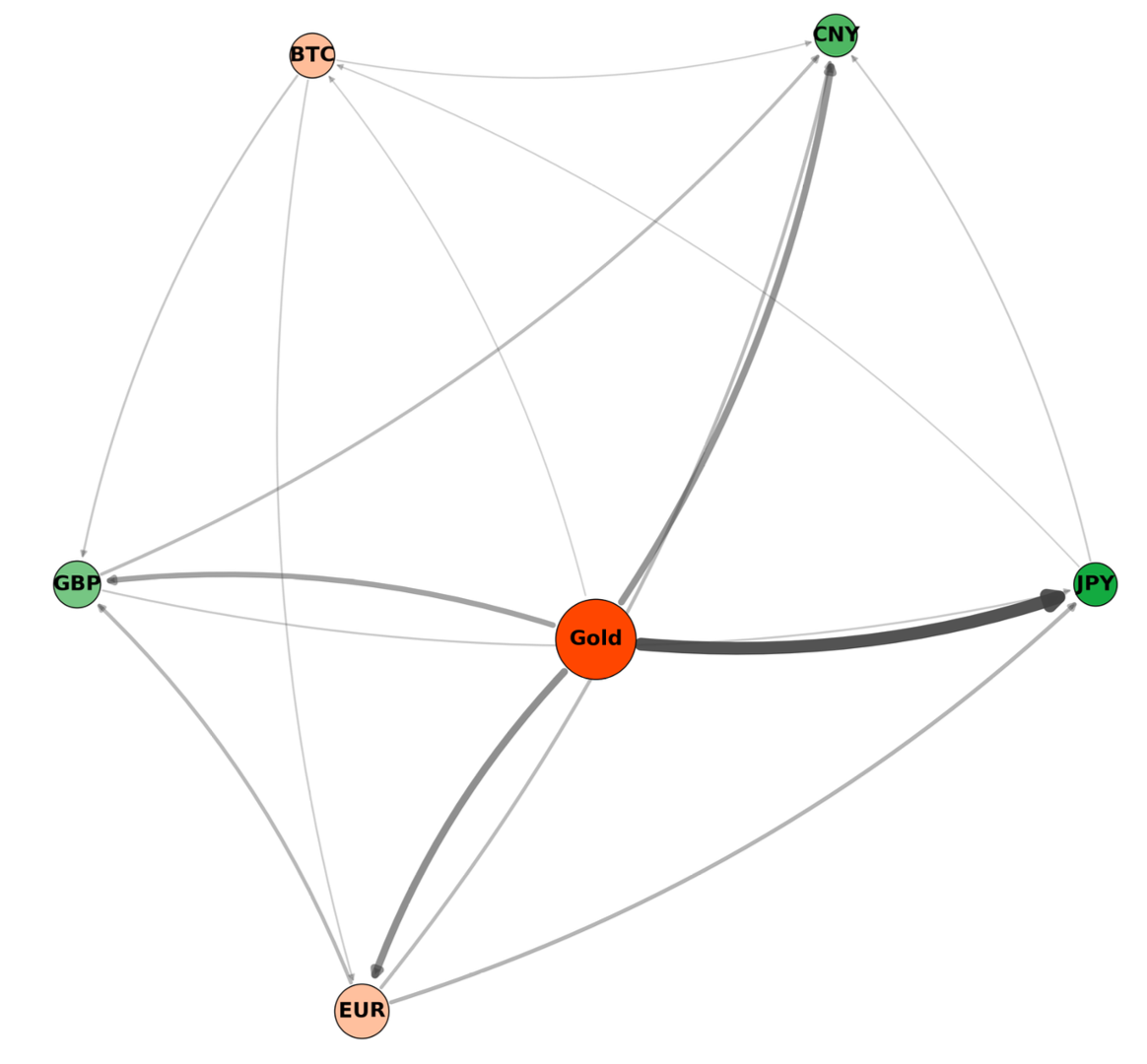}
\caption{Net pairwise spillover network among international financial assets excluding the dollar index. Warm-colored nodes indicate net transmitters, whereas green nodes indicate net receivers. Node size is scaled by total outward net spillover strength, and edge width and opacity represent the magnitude of positive net pairwise spillovers. Directed arrows show the dominant transmission direction between asset pairs. The network reveals a pronounced asymmetric spillover structure, with gold futures serving as the dominant source of net spillover transmission, especially toward major exchange rates and, to a lesser extent, Bitcoin.}
\label{fig:net_spillover}
\end{figure}
 
By contrast, the major exchange rate variables generally occupy net receiving or near balanced receiver positions. This does not mean that exchange rates are inactive in the network. Rather, it indicates that their incoming spillovers are generally larger than, or close to, their outgoing spillovers. In particular, CNY/USD and JPY/USD show clearer net receiving roles, while EUR/USD and GBP/USD are closer to balanced but still slightly receiver like. Bitcoin does not appear as a dominant net transmitting node comparable to gold futures, but it is not a purely peripheral receiver either. Its role is better interpreted as limited, time varying, and intermediate in the direct spillover structure.
 
\begin{table}[!htbp]
\centering
\caption{Directional spillover summary: TO, FROM, NET}
\label{tab:directional_spillover_summary}
\small
\setlength{\tabcolsep}{6pt}
\renewcommand{\arraystretch}{1.1}
\begin{tabular}{lccc}
\toprule
Asset & TO & FROM & NET\\
\midrule
Gold futures & 47.51 & 13.71 & 33.81\\
BTC/USD & 15.49 & 11.95 & 3.54\\
EUR/USD & 43.98 & 46.38 & -2.40\\
GBP/USD & 36.68 & 41.58 & -4.90\\
CNY/USD & 16.14 & 30.47 & -14.33\\
JPY/USD & 22.84 & 38.55 & -15.71\\
\bottomrule
\end{tabular}
 
\vspace{1mm}
\begin{minipage}{0.92\textwidth}
\small
\emph{Note}: TO denotes the total directional spillover transmitted by each asset to the other assets in the network, whereas FROM denotes the total directional spillover received by each asset from the other assets. NET is defined as TO minus FROM. A positive NET value indicates that
the asset acts as a net transmitter in the network, while a negative NET
value indicates that the asset acts as a net receiver. The values reported
in this table are based on the rolling results for the network excluding
the dollar index under the baseline specification with window size 100
and forecast horizon $H=10$. NET is computed from full-precision TO and
FROM estimates; therefore, minor discrepancies of 0.01 may arise when
subtracting the rounded TO and FROM values reported in the table.
\end{minipage}
\end{table}
 
Table~\ref{tab:directional_spillover_summary} confirms the asymmetric direct spillover structure. Gold futures record the largest positive NET spillover, with a TO value of 47.51, a FROM value of 13.71, and a NET value of 33.81. This indicates that gold futures transmit substantially more spillovers to other assets than they receive from them. Bitcoin also shows a positive NET spillover, but its magnitude is limited, with a TO value of 15.49, a FROM value of 11.95, and a NET value of 3.54. The major exchange rate variables generally play net receiving roles. CNY/USD and JPY/USD exhibit clear net receiving positions, with NET spillovers of -14.33 and -15.71, respectively. EUR/USD and GBP/USD are closer to balanced but still slightly receiver like, with NET spillovers of -2.40 and -4.90. 
The values in Table~\ref{tab:directional_spillover_summary} are time averages of the rolling estimates in the sense of Section~\ref{subsec:rolling_implementation}; the corresponding full-sample estimates are reported in Table~\ref{tab:full_sample_asset_roles}, which is why, for example, the NET spillover of gold futures is 33.81 here and 22.64 there.

\subsection{Stationary departure roles in the reverse source-tracing chain}
 
This subsection examines stationary departure in the reverse source-tracing Markov chain constructed from the positive net spillover matrix. While TO, FROM, and NET identify direct transmitter--receiver positions, $F_{\mathrm{out}}$ is the stationary probability of occupying an asset and then moving to a different node under the augmented reverse kernel. It is therefore a topology-based departure measure conditional on row normalization and the restart closure, not a measure of absolute forward economic shock outflow.
 
\begin{figure}[!htbp]
\centering
 
\begin{minipage}{0.32\textwidth}
\centering
\includegraphics[width=\textwidth,height=0.17\textheight,keepaspectratio]{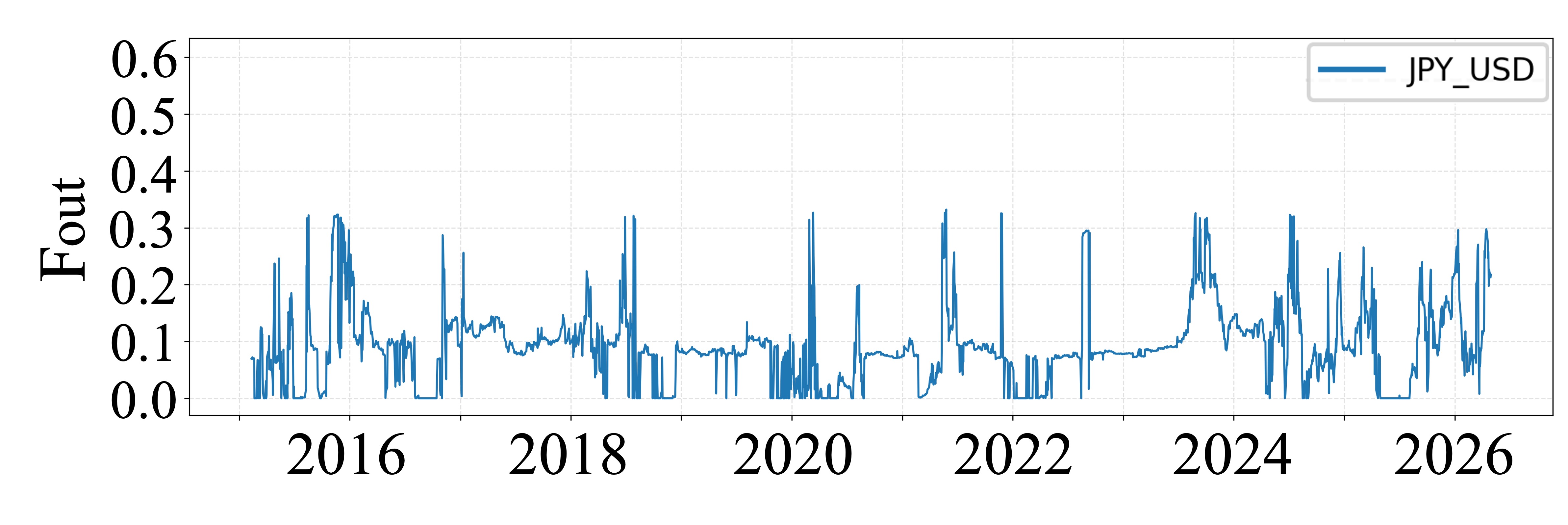}
 
{\small (a) JPY/USD}
\end{minipage}
\hfill
\begin{minipage}{0.32\textwidth}
\centering
\includegraphics[width=\textwidth,height=0.17\textheight,keepaspectratio]{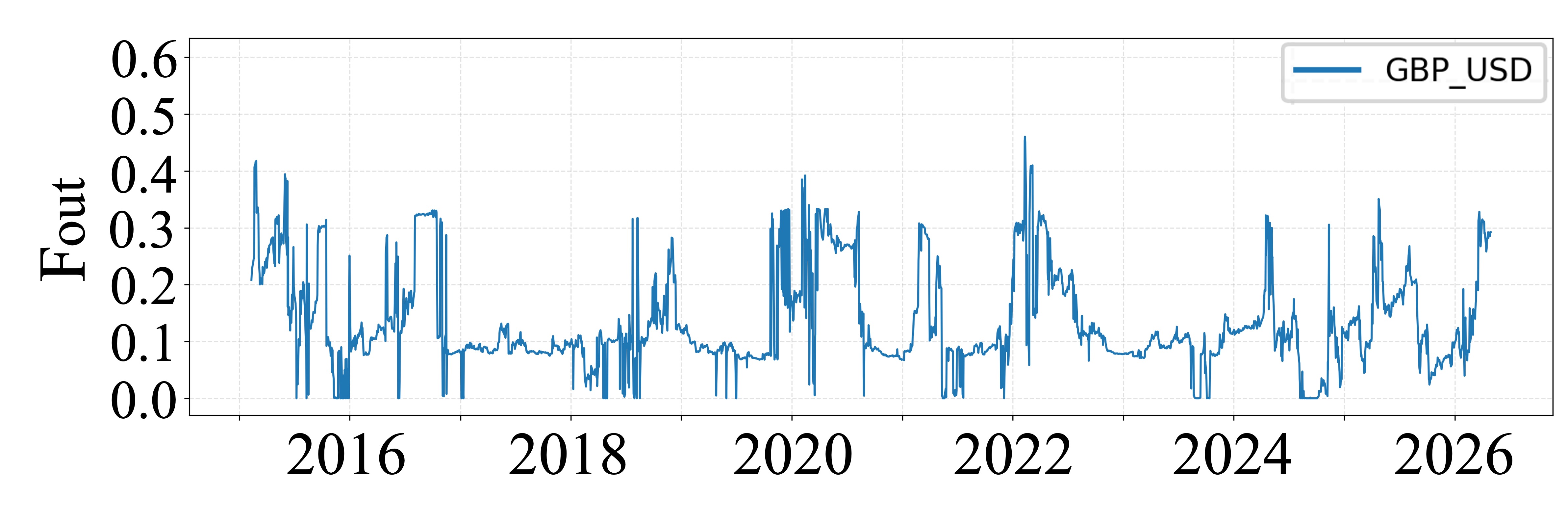}
 
{\small (b) GBP/USD}
\end{minipage}
\hfill
\begin{minipage}{0.32\textwidth}
\centering
\includegraphics[width=\textwidth,height=0.17\textheight,keepaspectratio]{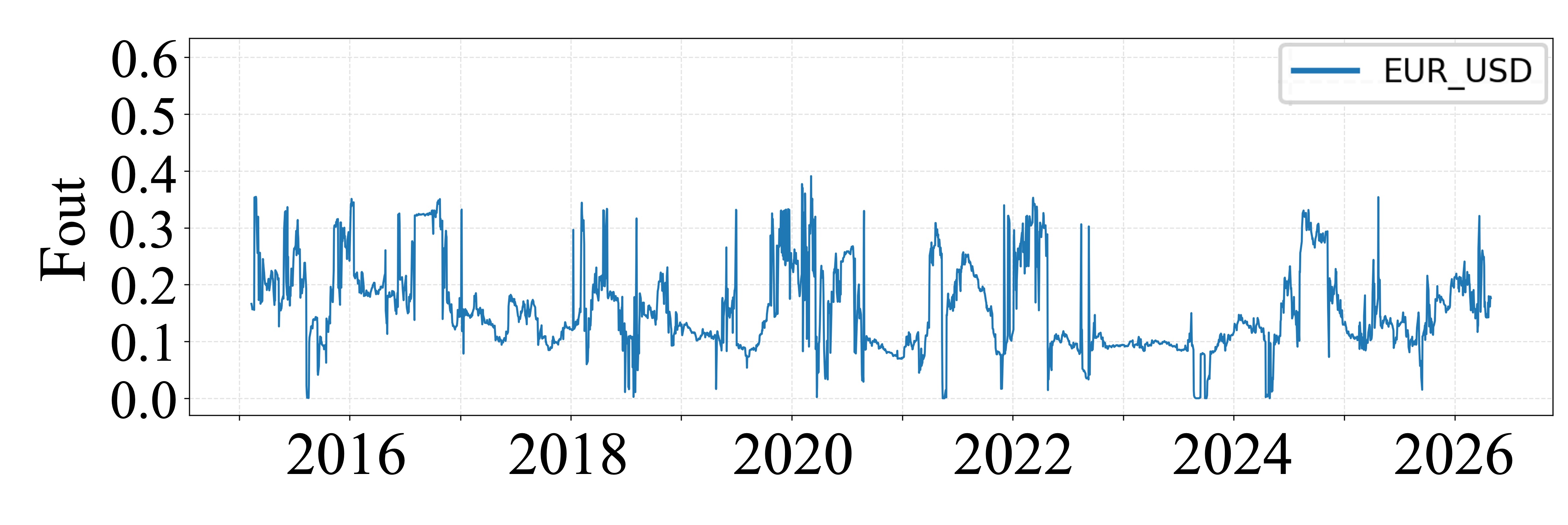}
 
{\small (c) EUR/USD}
\end{minipage}
 
\vspace{1mm}
 
\begin{minipage}{0.32\textwidth}
\centering
\includegraphics[width=\textwidth,height=0.17\textheight,keepaspectratio]{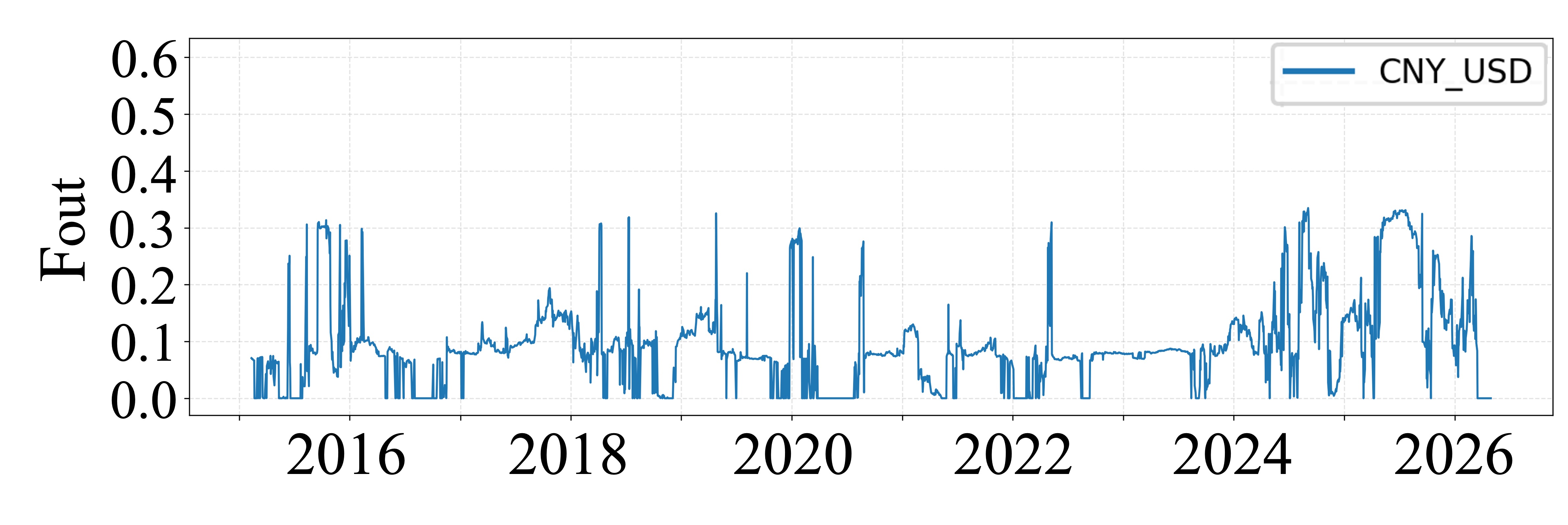}
 
{\small (d) CNY/USD}
\end{minipage}
\hfill
\begin{minipage}{0.32\textwidth}
\centering
\includegraphics[width=\textwidth,height=0.17\textheight,keepaspectratio]{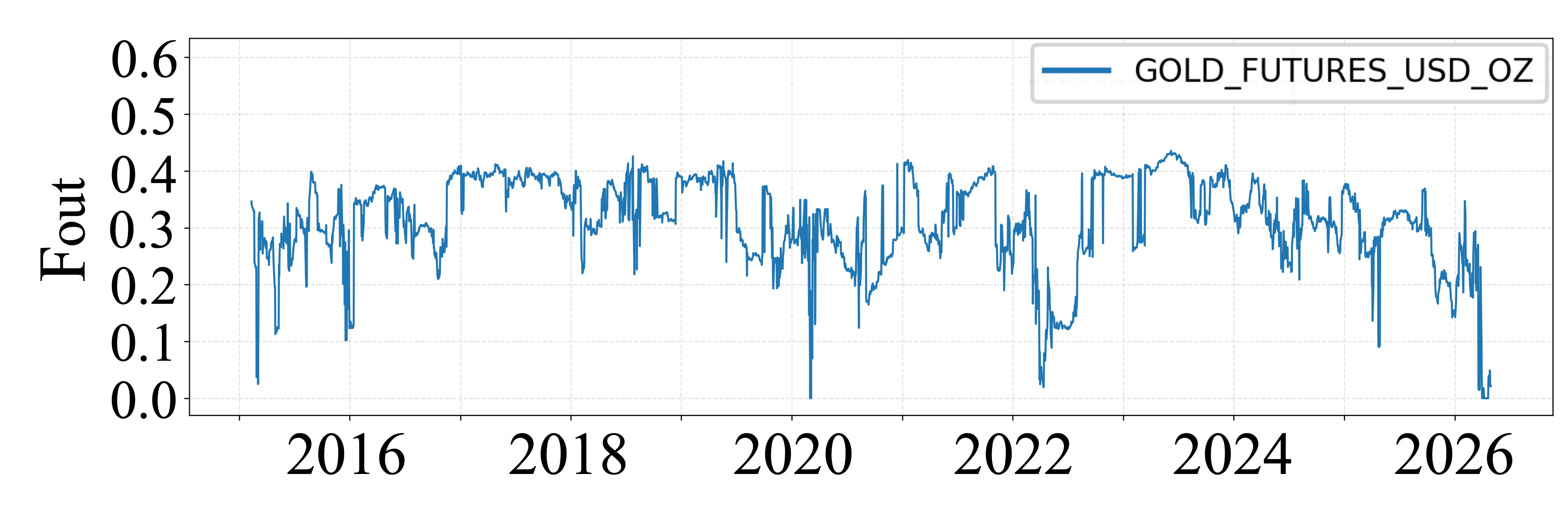}
 
{\small (e) Gold futures}
\end{minipage}
\hfill
\begin{minipage}{0.32\textwidth}
\centering
\includegraphics[width=\textwidth,height=0.17\textheight,keepaspectratio]{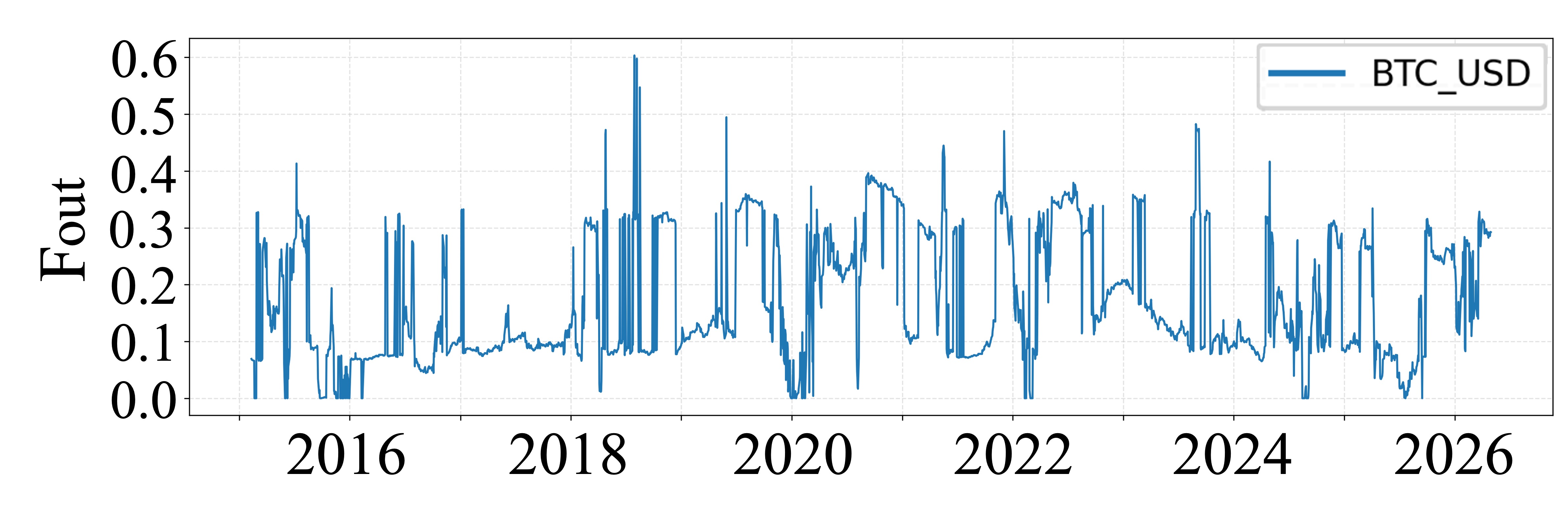}
 
{\small (f) BTC/USD}
\end{minipage}
 
\caption[Rolling stationary departure by asset]{
Rolling stationary departure measure, $F_{\mathrm{out}}$, by asset. The results are based on the reverse source-tracing kernel for the network excluding the dollar index, with window size 100 and forecast horizon $H=10$. Higher values indicate a greater stationary probability of occupying the asset and moving to a distinct node under the augmented kernel; they do not represent forward economic shock outflow.
}
\label{fig:rolling_fout_by_asset}
\end{figure}
 
Figure~\ref{fig:rolling_fout_by_asset} reports the rolling $F_{\mathrm{out}}$ series for each asset in the baseline network excluding the dollar index. Gold futures maintain the highest level over most of the sample period, indicating the greatest stationary departure frequency in the selected reverse source-tracing chain. This empirical ranking is consistent with the chosen reverse orientation and with the direct-spillover evidence identifying gold futures as the dominant net transmitter. It is not implied by Proposition~\ref{prop:source_concentration}, which only supplies an upper bound for a pure net sink under the uniform-restart closure.
 
The exchange-rate variables generally display lower $F_{\mathrm{out}}$ values than gold futures, although temporary increases are observed in some periods. Thus, their stationary departure frequencies in the reverse chain are more limited or episodic. JPY/USD and CNY/USD remain receiver-like in the direct spillover structure, which is a separate quantity from $F_{\mathrm{out}}$.
 
BTC/USD also shows a lower $F_{\mathrm{out}}$ level than gold futures but occasionally records increases. Under this reverse-chain measure, Bitcoin therefore has an intermediate and time-varying stationary departure role rather than the leading role.
 
Overall, Figure~\ref{fig:rolling_fout_by_asset} shows that stationary departure in the reverse source-tracing chain provides information not contained in direct spillover measures alone. Gold futures rank highly in both direct NET transmission and this conditional stationary topology, whereas the exchange-rate variables and BTC/USD show lower or more episodic stationary departure values.
 
Figure~\ref{fig:fout_vs_net} compares full-sample NET spillovers and $F_{\mathrm{out}}$ under forecast horizon $H=10$. The scatter plot examines whether direct net transmission and reverse-chain stationary departure identify the same relative asset roles.
 
\begin{figure}[!htbp]
\centering
\includegraphics[width=0.5\textwidth]{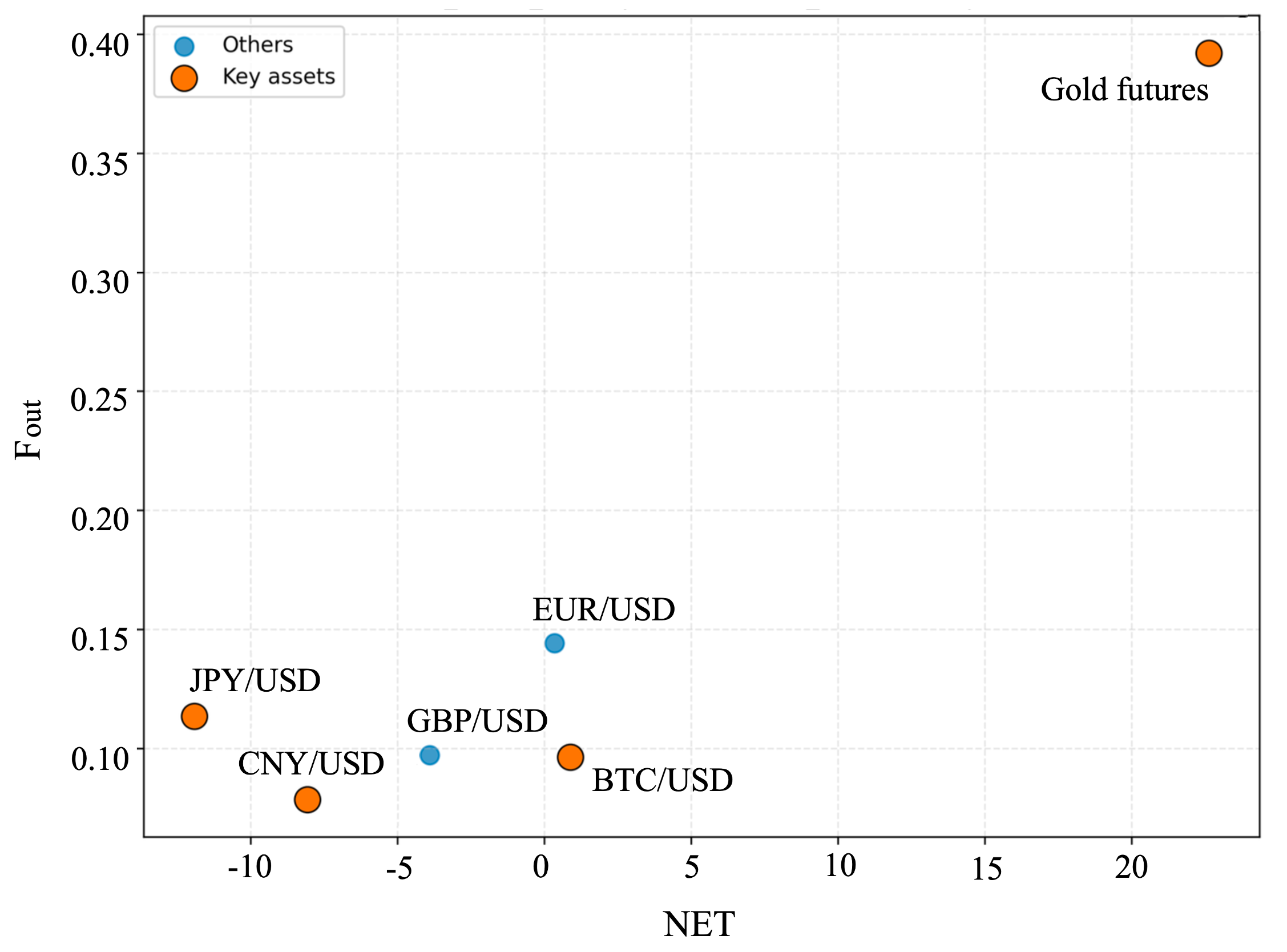}
\caption[Full-sample relationship between NET spillover and stationary departure]{
Full-sample relationship between NET spillover and the stationary departure measure $F_{\mathrm{out}}$ for the network excluding the dollar index. Assets farther to the right have stronger direct net spillover transmission, whereas assets higher on the vertical axis have greater stationary departure probability in the augmented reverse source-tracing chain.
}
\label{fig:fout_vs_net}
\end{figure}

The most notable feature in Figure~\ref{fig:fout_vs_net} is that gold futures are clearly separated from the other assets. Gold futures have a positive NET spillover of 22.64 and the highest $F_{\mathrm{out}}$ value of 0.392. Thus, they are both the dominant direct net transmitter and the leading stationary departure node under the selected reverse kernel.
 
By contrast, JPY/USD and CNY/USD have negative NET spillovers of -11.92 and -8.06, respectively, and lower stationary departure values. BTC/USD has a small positive NET value of 0.89, but its $F_{\mathrm{out}}$ value of 0.096 is much lower than that of gold futures. These comparisons describe positions in the augmented reverse chain and do not imply forward causal transmission.
 
Taken together, Figures~\ref{fig:rolling_fout_by_asset} and~\ref{fig:fout_vs_net} show that reverse-chain stationary departure is distinct from, but complementary to, direct spillover measures. The next subsection compares it with deterministic multistep upstream connectivity measured by Viral Centrality.

\subsection{Deterministic multistep upstream-connectivity roles}
 
This subsection examines multistep upstream connectivity using deterministically computed Viral Centrality (VC). While NET spillovers identify direct transmitter--receiver positions and $F_{\mathrm{out}}$ measures stationary departure in the reverse chain, VC is the fixed-point score defined by Eqs.~\eqref{eq:vc_initialization}--\eqref{eq:viral_centrality} on the same augmented reverse kernel. It is a cascade-inspired, closure-dependent topology diagnostic, not a Monte Carlo estimate or an exact expected cascade size.
 
Figure~\ref{fig:vc_vs_fout_fullsample} compares full-sample $F_{\mathrm{out}}$ and VC for the network excluding the dollar index. Gold futures have the highest $F_{\mathrm{out}}$ value of 0.392 but the lowest relative VC, 1.061. The two measures therefore summarize different features of the augmented reverse-kernel topology.
 
By contrast, CNY/USD has a low $F_{\mathrm{out}}$ value of 0.078 and a negative NET spillover of -8.06 but records the highest VC value, 1.955. BTC/USD and GBP/USD also have relatively high values of 1.897 and 1.888. Starting from these nodes, the reverse kernel can trace a broader set of upstream net-spillover suppliers; this does not mean that these assets spread shocks forward to that set.
 
This ordering follows from the kernel design. Each empirical transition of $\mathbf{P}(H)$ is oriented toward a net-spillover source, so $\mathrm{VC}_i$ summarizes approximate multistep upstream-source connectivity. Gold futures have zero incoming positive net-spillover strength in the full-sample empirical matrix. Their row is therefore replaced by the uniform restart closure, including its diagonal entry. The reported nonzero VC for gold partly reflects that closure; under the selected convention, gold nevertheless has the lowest relative VC. This statement is an empirical, closure-dependent comparison and is not a consequence of Proposition~\ref{prop:source_concentration}.

\begin{figure}[hb!]
\centering
\includegraphics[width=0.4\textwidth]{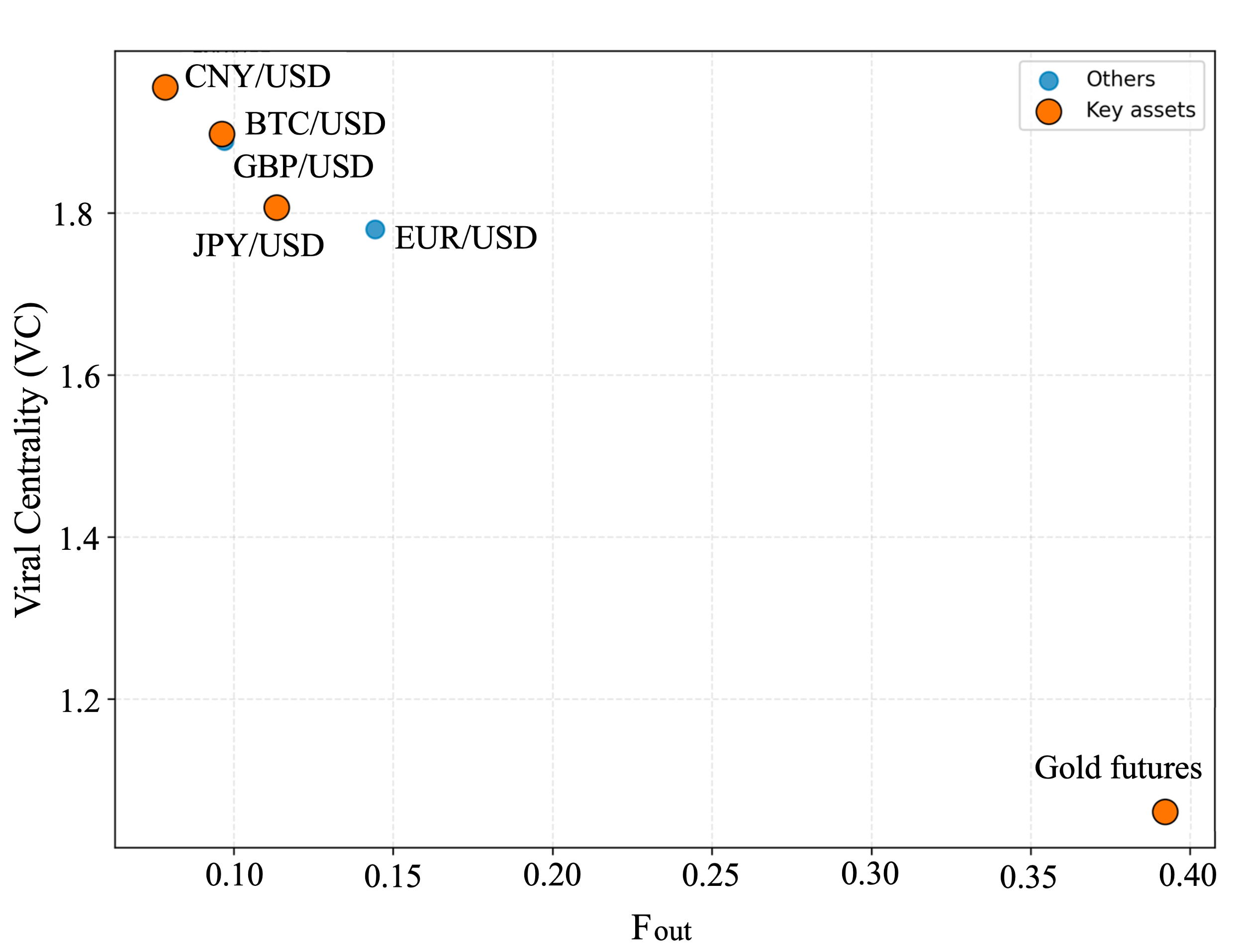}
\caption[Full-sample relationship between stationary departure and Viral Centrality]{
Full-sample relationship between the stationary departure measure $F_{\mathrm{out}}$ and deterministic Viral Centrality (VC) for the network excluding the dollar index. Assets farther to the right have greater stationary departure probability in the augmented reverse chain, whereas assets higher on the vertical axis have larger deterministic approximations to multistep upstream connectivity.
}
\label{fig:vc_vs_fout_fullsample}
\end{figure}
 
Figure~\ref{fig:rolling_vc_by_asset} reports the rolling deterministic VC values for each asset. Unlike the full-sample score computed from one full-sample kernel, the rolling plots show how the modeled upstream-connectivity score changes across estimated rolling kernels.
 
\begin{figure}[!htbp]
\centering
 
\begin{minipage}{0.32\textwidth}
\centering
\includegraphics[width=\textwidth,height=0.17\textheight,keepaspectratio]{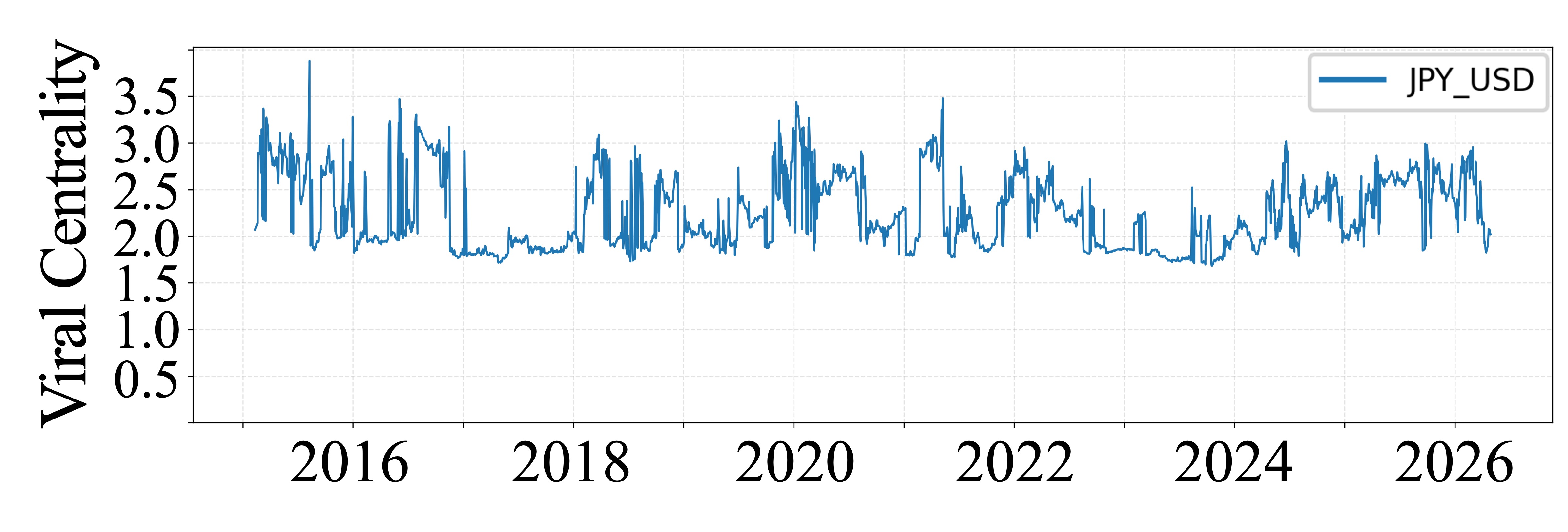}
 
{\small (a) JPY/USD}
\end{minipage}
\hfill
\begin{minipage}{0.32\textwidth}
\centering
\includegraphics[width=\textwidth,height=0.17\textheight,keepaspectratio]{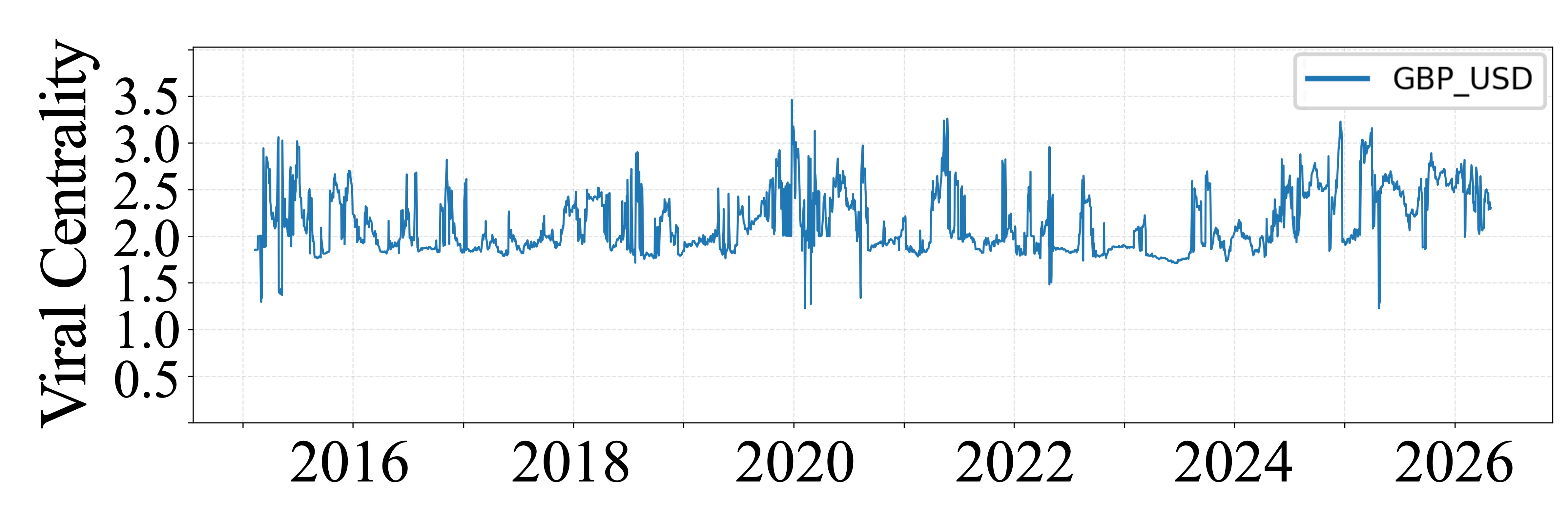}
 
{\small (b) GBP/USD}
\end{minipage}
\hfill
\begin{minipage}{0.32\textwidth}
\centering
\includegraphics[width=\textwidth,height=0.17\textheight,keepaspectratio]{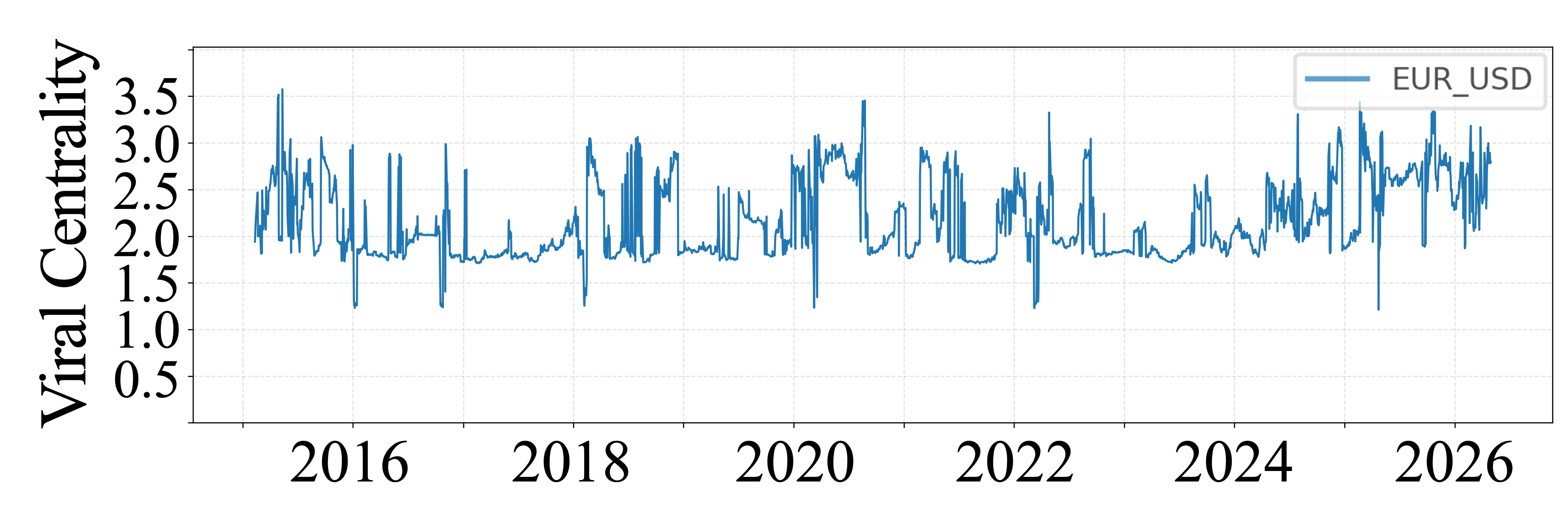}
 
{\small (c) EUR/USD}
\end{minipage}
 
\vspace{1mm}
 
\begin{minipage}{0.32\textwidth}
\centering
\includegraphics[width=\textwidth,height=0.17\textheight,keepaspectratio]{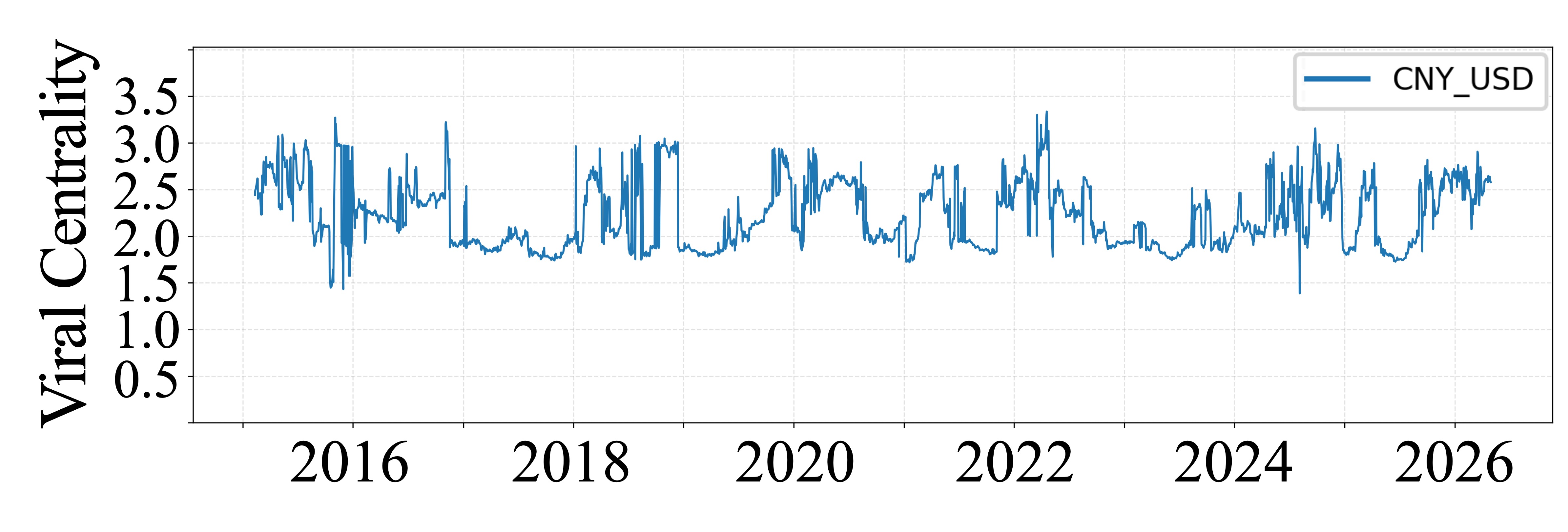}
 
{\small (d) CNY/USD}
\end{minipage}
\hfill
\begin{minipage}{0.32\textwidth}
\centering
\includegraphics[width=\textwidth,height=0.17\textheight,keepaspectratio]{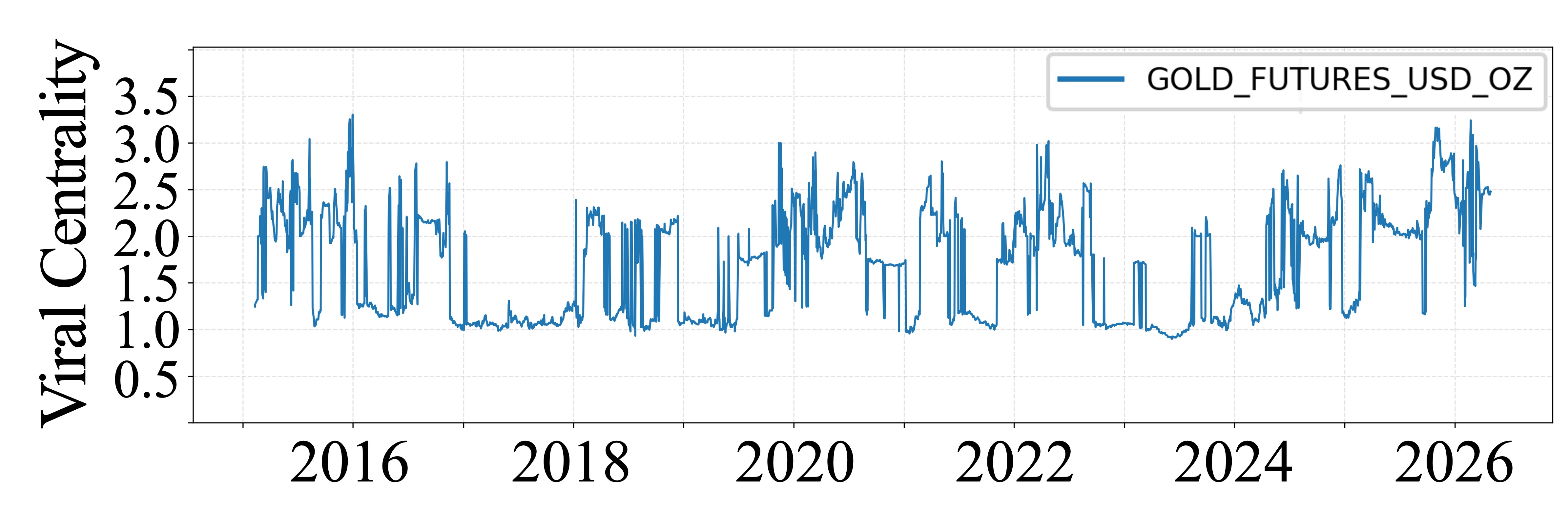}
 
{\small (e) Gold futures}
\end{minipage}
\hfill
\begin{minipage}{0.32\textwidth}
\centering
\includegraphics[width=\textwidth,height=0.17\textheight,keepaspectratio]{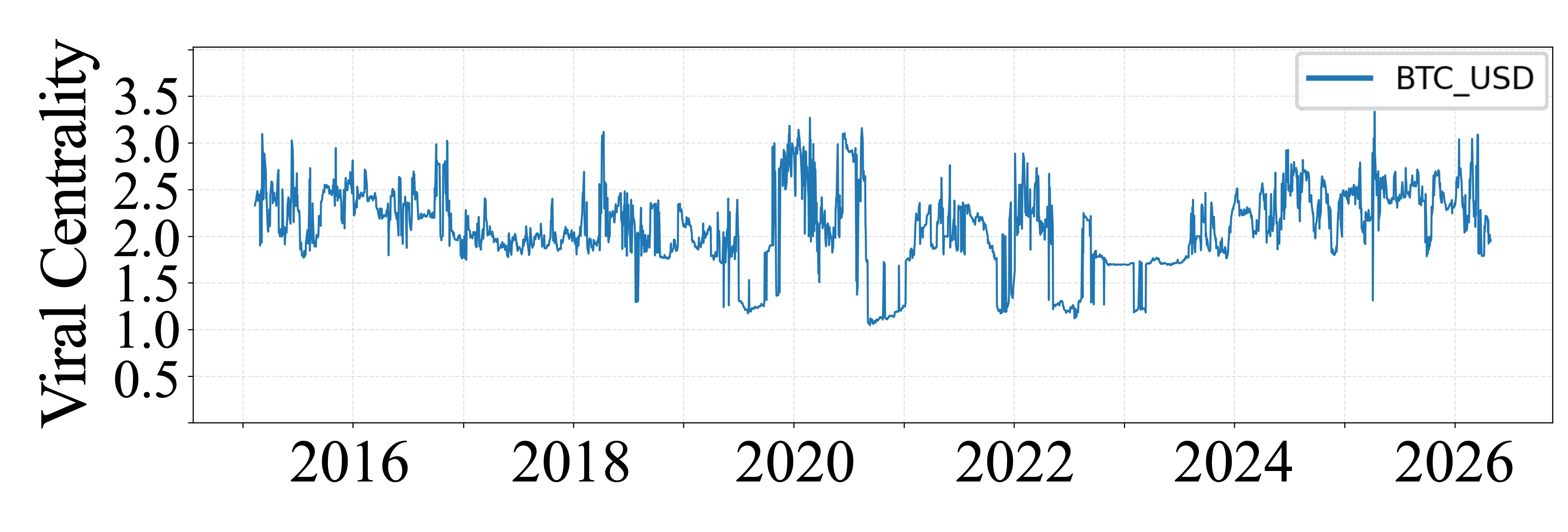}
 
{\small (f) BTC/USD}
\end{minipage}
 
\caption[Rolling Viral Centrality by asset]{
Rolling deterministic Viral Centrality (VC) by asset. The results are based on the augmented reverse source-tracing kernel for the network excluding the dollar index, with window size 100 and forecast horizon $H=10$. Higher values indicate larger algorithmic approximations to multistep upstream connectivity. Because the kernel may contain cycles and restart rows, VC is not an exact expected cascade size and is not evidence of realized causal contagion.
}
\label{fig:rolling_vc_by_asset}
\end{figure}

The rolling averages show that the ordering visible in
Figure~\ref{fig:rolling_vc_by_asset} is roughly the reverse of the
stationary-departure ordering: gold futures exhibit the lowest rolling-average
VC (1.690), the four exchange-rate variables cluster at the top of the range
with rolling averages between 2.134 and 2.248, and BTC/USD occupies an
intermediate level (2.088). Among the exchange rates, JPY/USD attains the
highest rolling-average VC (2.248), closely followed by CNY/USD (2.225).
The identity of the highest-scoring upstream-connectivity node therefore depends on the
estimand in the sense of Section~\ref{subsec:rolling_implementation}: the
time average of the rolling VC series is largest for JPY/USD, whereas the
full-sample transition matrix assigns the highest VC to CNY/USD.
The two
summaries are computed from different kernels and need not coincide; what is
common to both is that gold futures occupy the bottom of the VC ranking while
remaining the dominant node in direct transmission and reverse-chain stationary
departure. Detailed asset-level role comparisons for the auxiliary
U.S. Dollar Index network are reported in Supplementary Table~\ref{tabS:diffusion_role_with_without_dollar}.%

Taken together, the VC results show that TO, FROM, NET, and $F_{\mathrm{out}}$ do not exhaust the topology summarized by the reverse kernel. Gold futures dominate direct net transmission and reverse-chain stationary departure, whereas CNY/USD, BTC/USD, and GBP/USD have higher deterministic upstream-connectivity scores. The following subsection uses distribution-based state indicators to compare unusual values of these descriptive measures over time.

\subsection{Distribution-based diffusion-state diagnostics}
\label{subsec:diffusion_state_results}

As an additional diagnostic, we examine distribution-based diffusion-state
indicators constructed from rolling VC and rolling $F_{\mathrm{out}}$.
The detailed Weak EDS timelines and EDS\_score time-series results are
reported in Supplementary Figures~\ref{figS:weak_eds_timelines}
and~\ref{figS:eds_score_asset_heterogeneity}.

The results show that rolling VC and rolling $F_{\mathrm{out}}$ do not
jointly enter their upper tails under the conservative Joint EDS criterion.
This is consistent with the different constructions of the two diagnostics: a high reverse-chain stationary departure value need not coincide with a high deterministic upstream-connectivity score.

The continuous EDS\_score diagnostic nevertheless shows that diffusion roles
vary over time and across assets. Several exchange-rate nodes and BTC/USD
temporarily record larger combined standardized scores, whereas gold futures remain more prominent in direct transmission and stationary departure than in the VC diagnostic. These results show that the descriptive network measures vary over time and are not reducible to a single transmitter--receiver classification.

\subsection{Robustness to the U.S. Dollar Index}
\label{subsec:robustness_checks}
 
To examine whether the main findings are sensitive to the treatment of the broad dollar factor, this study compares the baseline network excluding the U.S. Dollar Index with an auxiliary network including the U.S. Dollar Index. The comparison is conducted under the same baseline setting of window size 100 and forecast horizon $H=10$. The purpose of this exercise is not to replace the baseline specification, but to assess whether the main diffusion-based interpretation remains robust when the dollar index is explicitly included as an additional node.
 
The auxiliary network including the U.S. Dollar Index exhibits a higher level
of system-wide connectedness than the baseline network excluding the dollar
index. The full rolling TCI comparison is reported in Supplementary
Figure~\ref{figS:rolling_tci_with_without_dollar}.
 
 
\begin{table}[H]
\centering
\caption{Robustness comparison between networks excluding and including the U.S. Dollar Index}
\label{tab:with_without_dollar_robustness}
\resizebox{0.98\textwidth}{!}{%
\begin{tabular}{l c l l l l c}
\toprule
Specification & Mean TCI & Main NET transmitter & Highest $F_{\mathrm{out}}$ asset 
& Highest $\pi$ asset & Highest VC asset & Joint EDS \\
\midrule
Without dollar index & 30.44 & Gold futures & Gold futures & Gold futures & JPY/USD & 0 \\
With dollar index    & 46.54 & U.S. Dollar Index & U.S. Dollar Index & U.S. Dollar Index & EUR/USD & 0 \\
\bottomrule
\end{tabular}
}
\vspace{1mm}
\begin{minipage}{0.98\textwidth}
\small
\textit{Note}: 
The comparison is based on the baseline setting with window size 100 and forecast horizon $H=10$. 
Mean TCI and the main NET transmitter are computed from rolling GFEVD spillover results. 
$F_{\mathrm{out}}$, the stationary distribution, and deterministic VC are computed from the augmented reverse source-tracing kernel constructed from the positive net spillover matrix. These quantities are conditional on row normalization and the uniform restart closure. 
All asset rankings in this table are based on time averages of rolling estimates; the full-sample VC ranking in Table~\ref{tab:full_sample_asset_roles} is computed from a different kernel and need not coincide. 
Joint EDS indicates whether any asset enters the conservative joint elevated diffusion state.
\end{minipage}
\end{table}
 
Table~\ref{tab:with_without_dollar_robustness} shows that the inclusion of the U.S. Dollar Index increases the overall connectedness of the network. The mean TCI rises from 30.44 in the baseline network excluding the dollar index to 46.54 in the auxiliary network including the dollar index. This indicates that the broad dollar factor strengthens system-wide spillover connectedness when it is explicitly included as a network node.

 
\subsection{Full-sample multidimensional asset role classification}
 
The preceding results show that direct net transmission, stationary departure in the reverse source-tracing chain, and deterministic multistep upstream connectivity do not necessarily identify the same asset roles.
This distinction is consistent with the broader connectedness literature, which emphasizes that the role of a financial asset is not fixed but may vary across market states, frequencies, and transmission channels.
For example, gold has been shown to play a time-varying role in cross-market connectedness, transmitting shocks over some horizons while receiving shocks over others \cite{shang2024quantile}.
Bitcoin has also been found to be connected with traditional financial assets, acting as a volatility transmitter during the COVID-19 period despite being only weakly connected over the full sample \cite{elsayed2022risk}.
Similarly, exchange-rate studies show that major currencies may serve as channels of shock propagation or as absorbers of spillovers from other currencies \cite{trancoso2023beyond}.
These findings suggest that a single direct NET measure may not be sufficient to characterize the role of each asset in a financial network.
 
To summarize these differences more explicitly, Table~\ref{tab:full_sample_asset_roles} reports the full-sample NET spillover, $F_{\mathrm{out}}$, and Viral Centrality (VC) for the baseline network excluding the dollar index under forecast horizon $H=10$.
This table provides a static benchmark for comparing the relative position of each asset across the three role dimensions.
 
\begin{table}[!htbp]
\centering
\caption{Full-sample asset roles based on direct spillover and source-tracing measures}
\label{tab:full_sample_asset_roles}
\scriptsize
\begin{tabular}{lrrrll}
\toprule
Asset & NET & $F_{\mathrm{out}}$ & VC & Direct NET role & Network role \\
\midrule
Gold futures & 22.64 & 0.392 & 1.061 & Dominant net transmitter & Highest stationary departure \\
BTC/USD & 0.89 & 0.096 & 1.897 & Weak net transmitter & High upstream score \\
EUR/USD & 0.35 & 0.144 & 1.780 & Near-balanced transmitter & Moderate upstream score \\
GBP/USD & -3.90 & 0.097 & 1.888 & Weak net receiver & High upstream score \\
CNY/USD & -8.06 & 0.078 & 1.955 & Clear net receiver & High upstream score \\
JPY/USD & -11.92 & 0.113 & 1.806 & Clear net receiver & Moderate upstream score \\
\bottomrule
\end{tabular}
 
\vspace{1mm}
\begin{minipage}{0.96\textwidth}
\small
\emph{Note}: The table reports full-sample measures for the network excluding the dollar index under forecast horizon $H=10$.
NET is defined as TO minus FROM.
$F_{\mathrm{out}}$ denotes stationary departure in the augmented reverse source-tracing chain, and VC denotes the deterministic fixed-point score in Eqs.~\eqref{eq:vc_initialization}--\eqref{eq:viral_centrality}.
The role classifications are descriptive and are based on the relative positions of assets in the full-sample NET, $F_{\mathrm{out}}$, and VC measures.
The VC role is conditional on row normalization and the uniform restart closure and should not be interpreted as an exact expected cascade size, a realized Joint EDS, or an externally validated cascade event.
\end{minipage}
\end{table}
 
Table~\ref{tab:full_sample_asset_roles} confirms that the asset ranking depends on the network summary considered.
Gold futures have the largest positive NET spillover and the highest $F_{\mathrm{out}}$, which supports their role as the dominant direct net transmitter and the leading stationary departure node in the selected reverse chain.
This result is consistent with prior evidence that gold is embedded in cross-market connectedness structures involving commodity, stock, currency, and macroeconomic indicators, and that its role may alternate between shock transmitter and receiver depending on the horizon, frequency, and market condition \cite{shang2024quantile,boubaker2023financial}.
However, the VC value of gold futures is the lowest among the six assets. Under the uniform restart closure, it is nonzero even though gold has zero incoming positive net-spillover strength. Gold therefore dominates direct transmission and reverse-chain stationary departure but not the relative deterministic upstream-connectivity score.
 
The exchange-rate variables show a different pattern.
CNY/USD and JPY/USD are clear net receivers in terms of NET spillover.
Although CNY/USD has the largest full-sample VC point estimate, CNY/USD, BTC/USD, and GBP/USD are better interpreted as a high-VC group than as a strict ranking; JPY/USD occupies a more moderate multistep position.
This result shows that an exchange-rate variable can be a receiver in direct net spillovers while the reverse source-tracing algorithm assigns it a high multistep upstream-connectivity score.
This interpretation is supported by studies showing that exchange-rate dynamics are inherently interconnected and that shocks can propagate through both direct and indirect channels \cite{tsiaplias2013multivariate,trancoso2023beyond}.
For CNY/USD, the high VC value means that a reverse-kernel walk initialized at this node can trace a relatively broad set of upstream net-spillover suppliers. This finding is broadly consistent with evidence that the Renminbi exchange-rate market is closely embedded in information linkages across China's bond and stock markets \cite{chen2023return}, but it is not evidence of a forward causal cascade.
For JPY/USD, the negative NET value is consistent with the view that the yen can absorb spillovers or operate as a safe-haven-related currency during periods of market stress, while still remaining an important currency node in global foreign-exchange dynamics \cite{ranaldo2010safe,trancoso2023beyond}.
 
EUR/USD and GBP/USD also illustrate the multidimensional nature of exchange-rate roles.
EUR/USD is close to balanced in direct NET spillover (a full-sample estimate; the rolling averages in Table~\ref{tab:directional_spillover_summary} place EUR/USD slightly on the receiver side) and has a moderate VC score, whereas GBP/USD is a weak net receiver but has relatively high VC.
This finding is consistent with currency-network evidence showing that the euro can serve as a major channel of shock propagation, while pound-related factors may absorb spillovers from other currencies \cite{trancoso2023beyond}.
Thus, major exchange rates need not appear as dominant direct transmitters in order to have broad upstream connectivity in the selected reverse-kernel topology.
 
BTC/USD occupies an intermediate position.
It has a small positive NET spillover and a relatively low $F_{\mathrm{out}}$, so it is not dominant in either direct transmission or reverse-chain stationary departure.
At the same time, its VC value is high relative to several other assets, indicating broad modeled access to upstream suppliers under the augmented reverse kernel.
This interpretation is consistent with prior evidence that Bitcoin is not isolated from traditional financial assets.
Elsayed et al.~\cite{elsayed2022risk} show that Bitcoin is connected with crude oil, gold, stocks, bonds, the U.S. dollar, and global uncertainty measures, and that Bitcoin becomes a net transmitter of volatility spillovers during the COVID-19 period, even though its full-sample connectedness is relatively weak.
Accordingly, BTC/USD occupies an intermediate role across the three descriptive network summaries rather than being either a dominant source or a purely peripheral asset.
 
Overall, Table~\ref{tab:full_sample_asset_roles} supports the main role-based interpretation of this study.
Direct spillover transmission, reverse-chain stationary departure, and deterministic upstream connectivity capture distinct summaries of the same estimated financial network.
This conclusion is consistent with the broader literature on financial contagion and spillovers, which emphasizes that shocks may be transmitted not only through direct links but also through indirect channels and intermediary mechanisms \cite{tsiaplias2013multivariate}.
Accordingly, asset roles should not be inferred from NET spillover alone, but can be compared using direct transmission, stationary departure, and deterministic upstream-connectivity measures, with the stated kernel and closure qualifications.

\section{Discussion}
\label{sec:discussion}

The empirical results support a role-based interpretation of cross-asset connectedness. Assets should not be characterized only as net transmitters or net receivers. Direct spillovers, stationary departure in the reverse source-tracing chain, and deterministic multistep upstream connectivity summarize related but distinct properties of the estimated network. An asset that dominates direct transmission need not rank highest under the other summaries, and a net receiver may have broad modeled access to upstream suppliers. These quantities are descriptive and conditional on the kernel construction.

From a methodological perspective, conventional TO, FROM, and NET measures remain useful for identifying the dominant direction of direct spillover transmission~\cite{diebold2012better,diebold2014network}. Embedding positive net spillovers into a row-stochastic kernel additionally permits a reverse walk that traces upstream net-spillover suppliers. Its stationary distribution, departure measure, and deterministic VC score characterize this normalized topology. They do not replace conventional spillover measures, restore the absolute magnitudes removed by row normalization, or establish forward causal shock propagation.

Gold futures are the dominant direct net transmitter and have the highest stationary departure value in the reverse chain. This pattern is consistent with literature relating gold to safe-haven demand, inflation expectations, and global risk conditions~\cite{ciner2013hedges,reboredo2013gold}. Gold nevertheless has the lowest relative deterministic VC score. That comparison shows that direct transmission, reverse-chain stationarity, and upstream-connectivity scores need not have the same ranking; it does not show that gold is weak in a forward physical cascade.

The exchange-rate variables provide a contrasting pattern. Some are net receivers or near-balanced nodes but have relatively high VC. Because the kernel reverses the empirical edge direction, this means that the algorithm initialized at those nodes traces a broad upstream supplier set. It does not identify them as forward cascade conduits. The contrast still supports the narrower conclusion that NET spillovers alone do not describe every property of the normalized network topology.

Bitcoin occupies an intermediate role. It is not a dominant direct source comparable to gold futures, and its reverse-chain stationary departure value is relatively low, but its deterministic VC score is high. This mixed profile is consistent with prior discussion of Bitcoin as both an alternative asset and an asset sensitive to liquidity conditions and risk appetite~\cite{dyhrberg2016bitcoin,urquhart2019bitcoin,bouri2017hedge}. Here it should be read as a cross-measure descriptive profile, not as proof that Bitcoin transmits multistep cascades.

The auxiliary specification including the U.S. Dollar Index provides an additional perspective on the organization of the diffusion network. The baseline specification excludes the dollar index not because dollar-related conditions are irrelevant, but because the main objective is to examine the internal diffusion topology among exchange rates, gold futures, and Bitcoin without allowing a broad dollar factor to enter the network as an explicit and potentially dominant node. This choice is also motivated by the fact that bilateral exchange rates share common currency risk factors and dollar-related systematic components~\cite{lustig2011common,verdelhan2018share}. If the dollar index is introduced as an auxiliary node, it may absorb part of the common currency component and alter the estimated diffusion hierarchy. Therefore, the auxiliary specification should not be interpreted as replacing the baseline network. Rather, it provides a robustness-oriented comparison for assessing whether the main diffusion-role patterns remain broadly consistent when a common dollar-related factor is explicitly included.

The distribution-based state indicators show that the two rolling diagnostics are episodic and asset-specific. Under the conservative joint criterion, no asset enters both upper tails simultaneously. Weak elevated states identify periods in which at least one diagnostic
is large relative to its asset-specific rolling distribution, whereas
the continuous EDS\_score summarizes the combined standardized level of
the two diagnostics. Because the inputs are reverse-chain $F_{\mathrm{out}}$ and closure-dependent deterministic VC, these indicators are descriptive diffusion-state diagnostics rather than
observed contagion states.

Several limitations should be noted. First, the results are based on return-based VAR-GFEVD spillovers and therefore describe forecast-error-variance connectedness rather than structural causal relations.
The estimated network should be interpreted as a statistical representation of shock transmission, not as evidence of underlying causal mechanisms.
Second, the Markov transition kernel is constructed from positive net spillovers.
This choice emphasizes the dominant direction of net transmission, but it abstracts from offsetting bilateral effects and negative net relationships.
Third, the uniform-row convention for nodes with zero incoming positive net-spillover strength, that is, pure net sources under the chosen orientation, is a modeling choice used to keep the transition matrix row-stochastic. It introduces augmented restart transitions, including diagonal entries, and affects both $F_{\mathrm{out}}$ and VC. Alternative closures may change their levels and rankings. Fourth, VC is a deterministic approximation on this augmented reverse kernel. With cycles it need not equal the exact independent-cascade expectation; because a row-stochastic kernel has spectral radius one, the small-spectral-radius accuracy condition is not available here. VC is therefore used only as a closure-dependent topology score. Fifth, the rolling results may depend on the selected window size and forecast horizon, although auxiliary specifications examine the main patterns.
In the seven-variable auxiliary network, estimating a VAR(2) within a 100-day rolling window leaves a relatively limited number of observations per estimated coefficient. Accordingly, the dollar-index results should be interpreted as supplementary robustness
evidence rather than as a more precisely estimated replacement for the baseline network.

Sixth, the series are observed at daily frequency, but their closing or settlement prices may be recorded at different venue-specific times.
As a result, the estimated lead--lag and directional spillover structure may partly reflect non-synchronous observation times rather than economic transmission alone. This issue is particularly relevant when comparing continuously traded exchange rates and Bitcoin with gold futures, whose daily value is based on an exchange-specific settlement convention.
Accordingly, the directional estimates, including the dominant transmitter role of gold futures, should be interpreted with this timing caveat in mind.
Finally, the proposed role categories should be interpreted as descriptive diffusion roles implied by the estimated network structure, not as fixed structural identities of the assets.

Overall, the contribution is a transparent comparison of direct spillovers with two source-tracing network summaries. Distinguishing direct transmission, reverse-chain stationary departure, and deterministic upstream connectivity reveals heterogeneous asset profiles not captured by the transmitter--receiver classification alone. The interpretation remains limited to statistical connectedness and the specified normalized, augmented kernel.

\section{Conclusion}
\label{sec:conclusion}
This study investigated cross-asset connectedness in a network of major exchange rates, gold futures, and Bitcoin. In addition to conventional net-transmitter and net-receiver measures~\cite{diebold2009measuring,diebold2012better,diebold2014network}, it compared direct spillovers with stationary departure and deterministic multistep upstream connectivity on an augmented reverse source-tracing kernel.

Return spillovers were estimated using a VAR-GFEVD framework, which provides ordering-invariant forecast-error variance shares~\cite{koop1996impulse,pesaran1998generalized}. Positive pairwise net spillovers were then converted into a row-stochastic reverse kernel: a step from a receiver moves toward one of its net-spillover suppliers. TO, FROM, and NET were compared with the stationary departure measure $F_{\mathrm{out}}$ and the deterministic Viral Centrality fixed point~\cite{fink2023centrality}. The latter approximates multistep upstream connectivity on the full augmented kernel; it is neither a Monte Carlo estimate nor an exact cascade expectation in a cyclic network.

Gold futures emerge as the dominant direct net transmitter and the highest-$F_{\mathrm{out}}$ node. They have the lowest relative VC score under the selected uniform-restart closure. Some exchange-rate receivers and Bitcoin have higher VC scores, meaning that the reverse-kernel recursion initialized at those nodes traces a broader set of upstream suppliers. This is not evidence that those assets are forward diffusion conduits. The main empirical conclusion is therefore a difference in rankings across well-defined descriptive measures, not a causal statement about realized shock cascades.

These findings show the value of comparing a direct NET measure with conditional source-tracing summaries. The framework links VAR-GFEVD connectedness with Markov-chain occupation and deterministic network recursion while keeping the three estimands distinct. It provides a more detailed descriptive characterization of asset positions without making structural causal claims.

Several limitations should be acknowledged. First, the estimates describe forecast-error-variance connectedness rather than structural causal relations. Second, retaining only positive pairwise net spillovers omits offsetting and opposite-direction components, while row normalization removes absolute edge magnitude. Third, the uniform row for nodes with zero incoming positive net-spillover strength is a restart closure that directly affects $F_{\mathrm{out}}$ and VC. Fourth, deterministic VC can overestimate exact independent-cascade reach in cyclic graphs, and the present row-stochastic kernel has spectral radius one. Fifth, rolling estimates depend on the window size and forecast horizon. The reported roles are therefore descriptive, specification-dependent network summaries rather than fixed asset identities.

Future research may compare alternative kernel orientations and dangling-row closures, retain information on absolute edge magnitude, and validate deterministic VC against exact or Monte Carlo independent-cascade reach. The framework can also be extended to broader asset systems and nonlinear or regime-dependent spillover models. Finally, the present PageRank comparison can be extended to other diffusion-based centralities to assess the robustness of the role classifications.

\section*{Acknowledgement}
This work of S. H. Choi and H. Choi was supported by the National Research Foundation of Korea(NRF) grant funded by the Korea government(MSIT) (No. 2022R1A5A1033624 and RS-2024-00342939). The work of S. Jang was supported by the National Research Foundation of Korea (NRF) grant funded by the Korea government (No. RS-2024-00464395). H. Lee is partially supported by the National Research Foundation of Korea (NRF) grant funded by the Korea government (MSIT) (No. RS-2024-00408003) and by a KIAS Individual Grant (AP103101) via the Center for AI and Natural Sciences at the Korea Institute for Advanced Study.


\clearpage
\appendix

\setcounter{section}{0}
\setcounter{table}{0}
\setcounter{figure}{0}
\setcounter{equation}{0}
\renewcommand{\thesection}{S\arabic{section}}
\renewcommand{\thesubsection}{\thesection.\arabic{subsection}}
\renewcommand{\thetable}{S\arabic{section}.\arabic{table}}
\renewcommand{\thefigure}{S\arabic{section}.\arabic{figure}}
\renewcommand{\theequation}{S\arabic{section}.\arabic{equation}}
\renewcommand{\theHsection}{S\arabic{section}}
\renewcommand{\theHsubsection}{\theHsection.\arabic{subsection}}
\renewcommand{\theHtable}{S\arabic{section}.\arabic{table}}
\renewcommand{\theHfigure}{S\arabic{section}.\arabic{figure}}
\renewcommand{\theHequation}{S\arabic{section}.\arabic{equation}}
\counterwithin*{table}{section}
\counterwithin*{figure}{section}

\begin{center}
{\Large Supplementary Material}\\[1em]
{\large Markovian Shock-Source Tracing and Multidimensional Asset Roles in Exchange Rates, Gold Futures, and Bitcoin}
\end{center}

\vspace{1em}

\section{Rolling directional spillover plots}
\label{supsec:rolling_directional_spillovers}

Figures~\ref{fig:to_spillover_asset_level}--\ref{fig:net_spillover_asset_level} summarize the time varying directional spillover patterns based on TO, FROM, and NET measures. Whereas the rolling TCI summarizes the overall level of system wide connectedness, these figures show how that connectedness is formed through the transmitter and receiver roles of individual assets. Because the y-axis scale is kept identical across assets within each figure, the relative magnitude of spillover intensity can be compared directly across assets.

\begin{figure}[!htbp]
\centering
 
\begin{minipage}{0.32\textwidth}
\centering
\includegraphics[width=\textwidth]{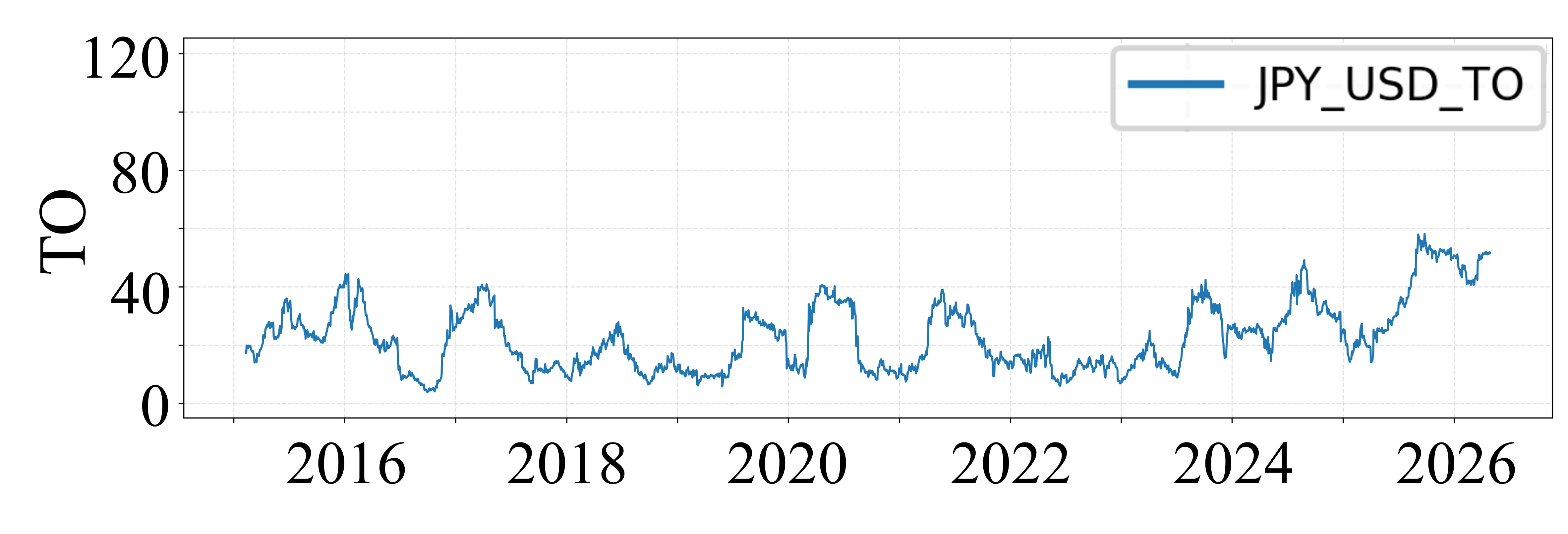}
 
{\small (a) JPY/USD}
\end{minipage}
\hfill
\begin{minipage}{0.32\textwidth}
\centering
\includegraphics[width=\textwidth]{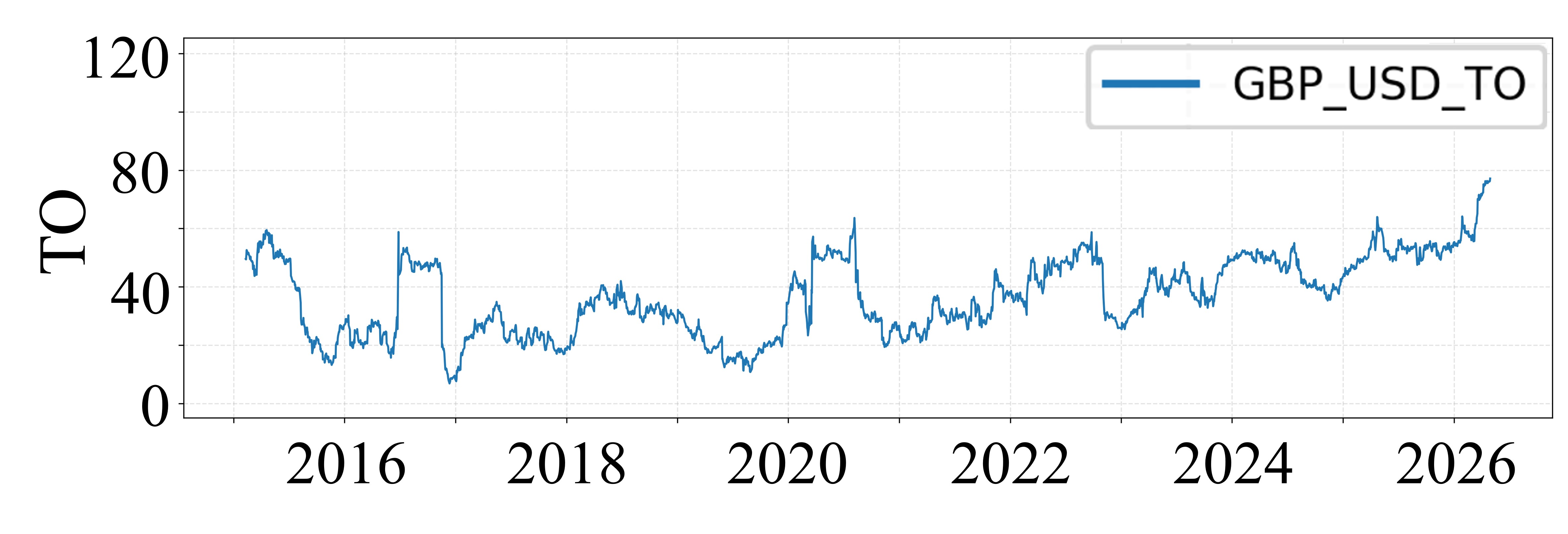}
 
{\small (b) GBP/USD}
\end{minipage}
\hfill
\begin{minipage}{0.32\textwidth}
\centering
\includegraphics[width=\textwidth]{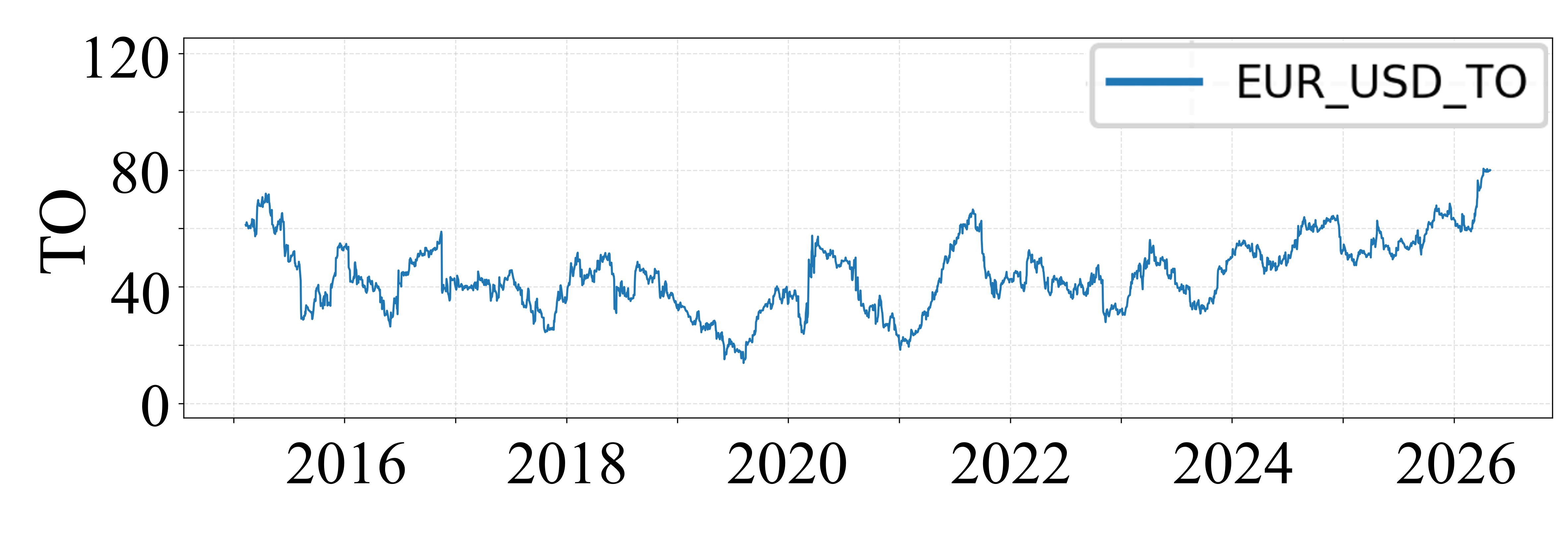}
 
{\small (c) EUR/USD}
\end{minipage}
 
\vspace{2mm}
 
\begin{minipage}{0.32\textwidth}
\centering
\includegraphics[width=\textwidth]{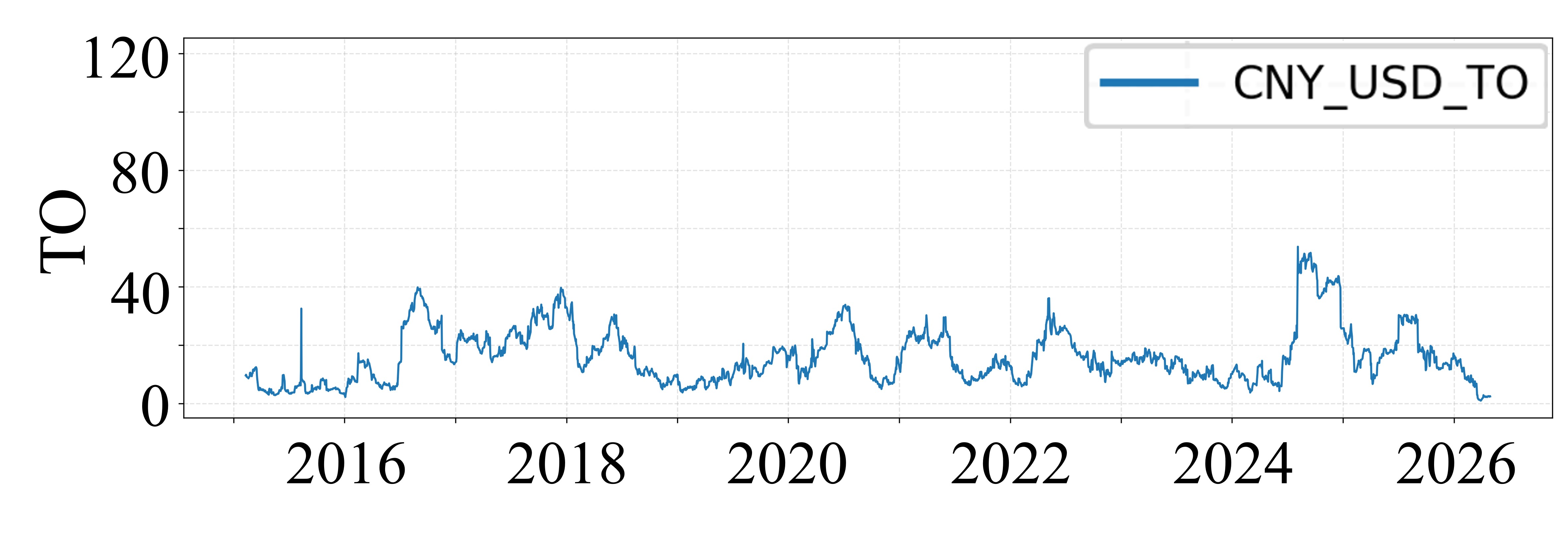}
 
{\small (d) CNY/USD}
\end{minipage}
\hfill
\begin{minipage}{0.32\textwidth}
\centering
\includegraphics[width=\textwidth]{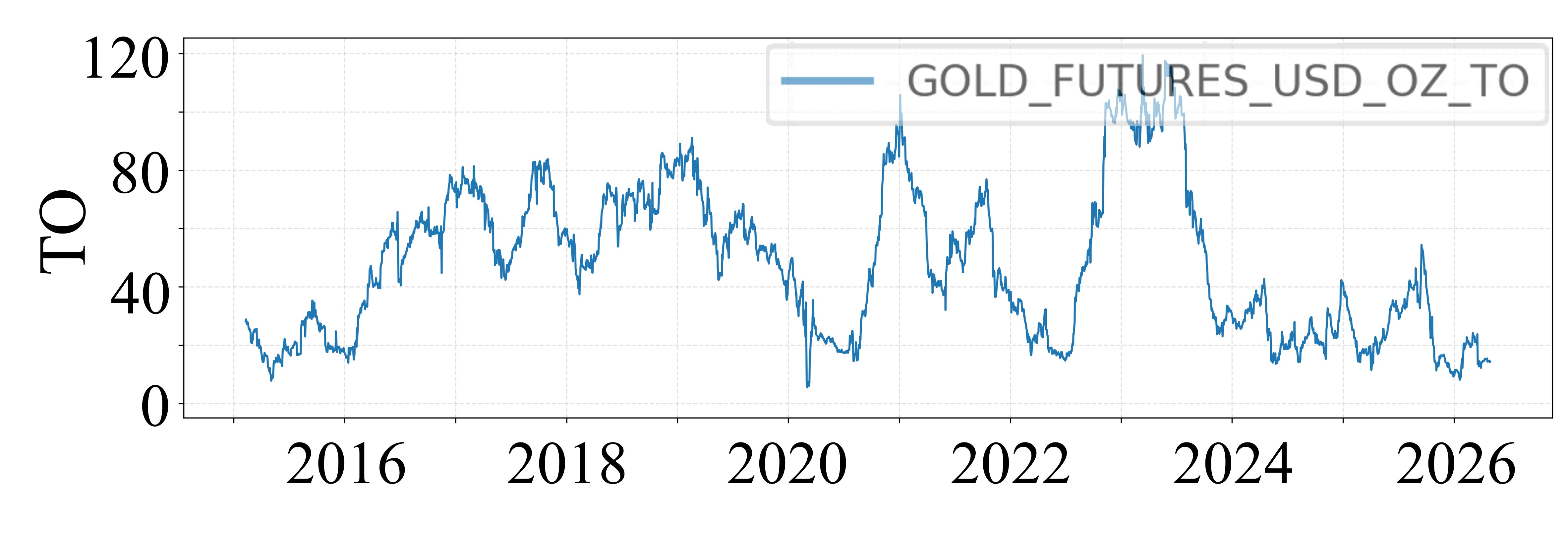}
 
{\small (e) Gold futures}
\end{minipage}
\hfill
\begin{minipage}{0.32\textwidth}
\centering
\includegraphics[width=\textwidth]{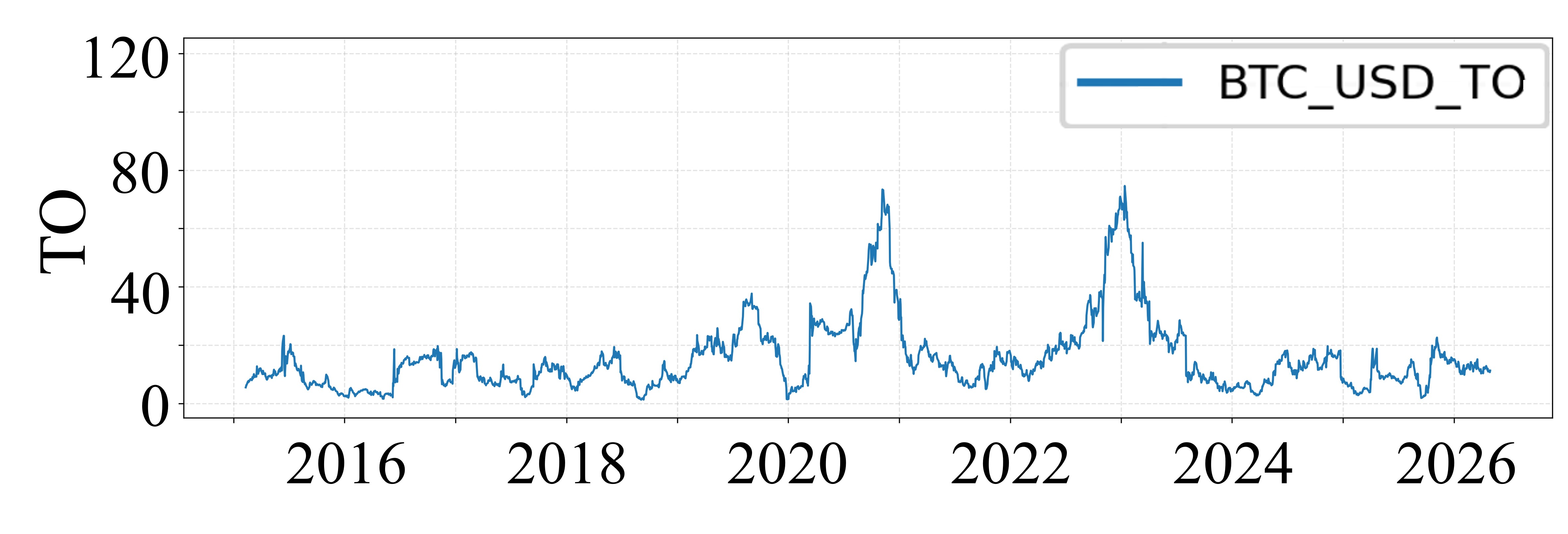}

{\small (f) BTC/USD}
\end{minipage}
 
\caption[Asset level TO spillovers]{
Asset level TO spillovers for the network excluding the dollar index under window size 100 and forecast horizon $H=10$. TO measures spillovers transmitted from each asset to other assets. The y-axis scale is kept identical across assets to allow direct comparison of transmission intensity.
}
\label{fig:to_spillover_asset_level}
\end{figure}

\begin{figure}[!htbp]
\centering
 
\begin{minipage}{0.32\textwidth}
\centering
\includegraphics[width=\textwidth]{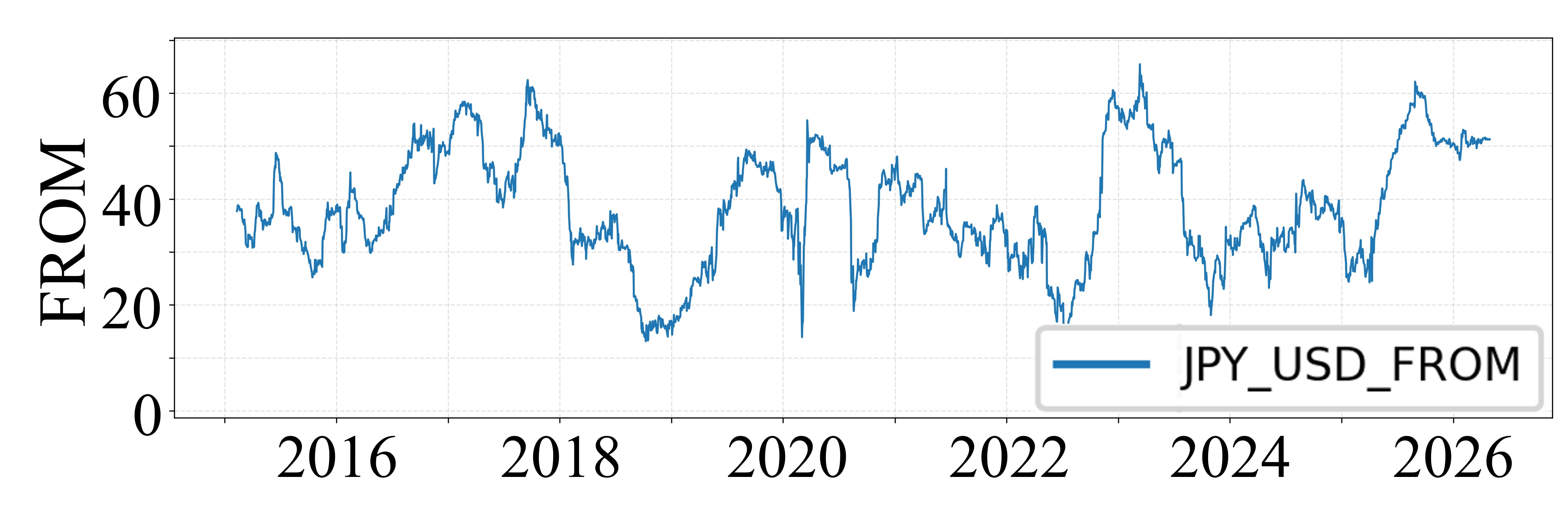}
 
{\small (a) JPY/USD}
\end{minipage}
\hfill
\begin{minipage}{0.32\textwidth}
\centering
\includegraphics[width=\textwidth]{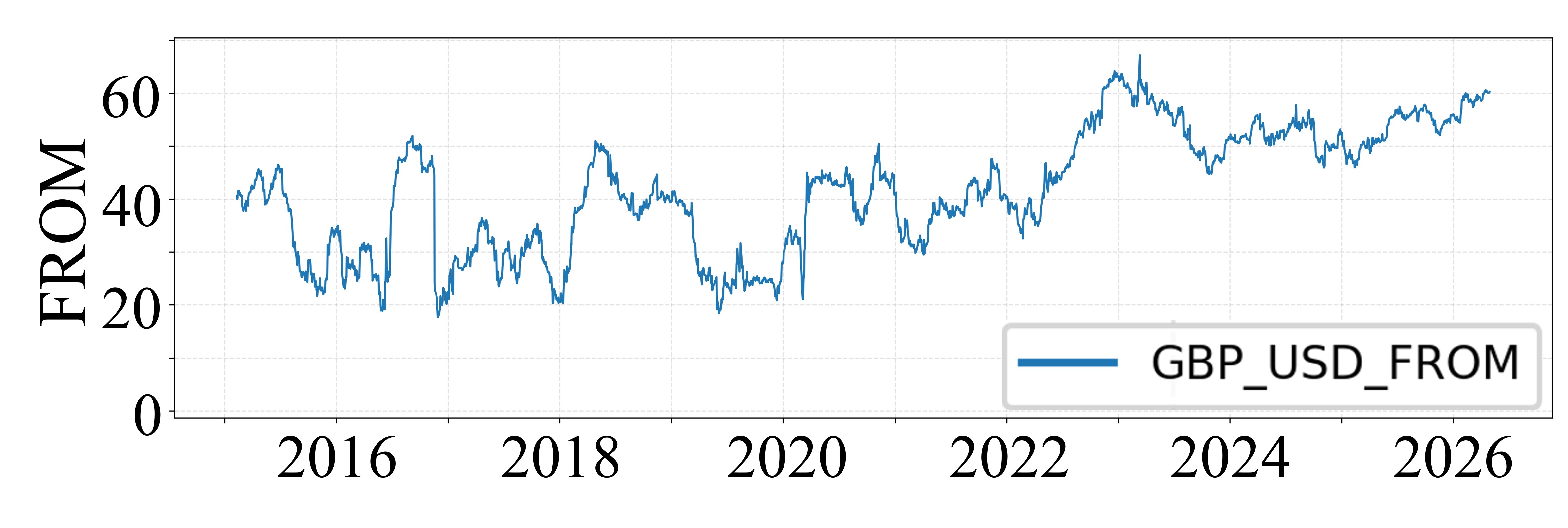}
 
{\small (b) GBP/USD}
\end{minipage}
\hfill
\begin{minipage}{0.32\textwidth}
\centering
\includegraphics[width=\textwidth]{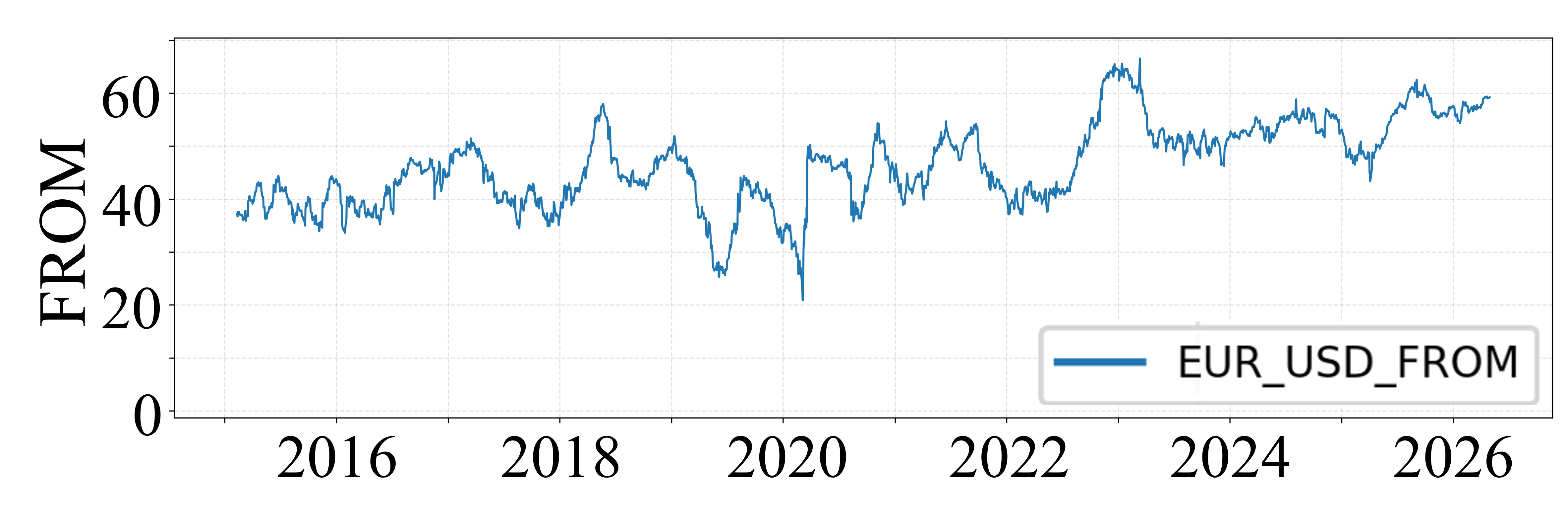}
 
{\small (c) EUR/USD}
\end{minipage}
 
\vspace{2mm}
 
\begin{minipage}{0.32\textwidth}
\centering
\includegraphics[width=\textwidth]{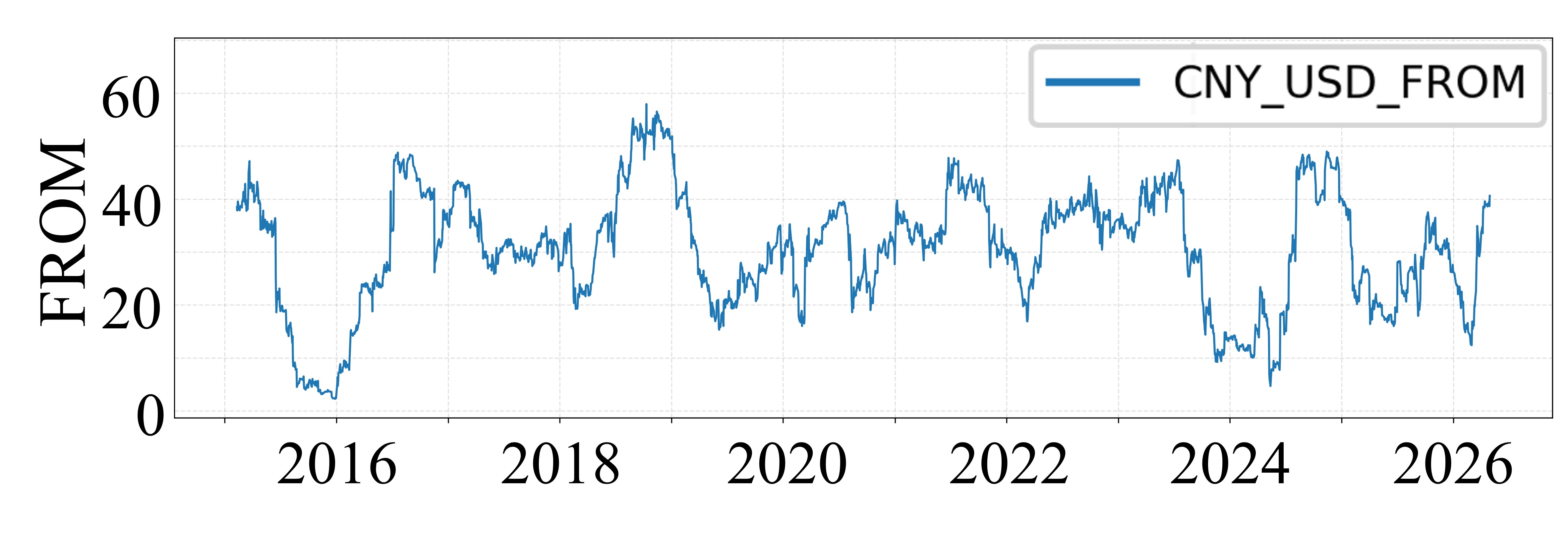}
 
{\small (d) CNY/USD}
\end{minipage}
\hfill
\begin{minipage}{0.32\textwidth}
\centering
\includegraphics[width=\textwidth]{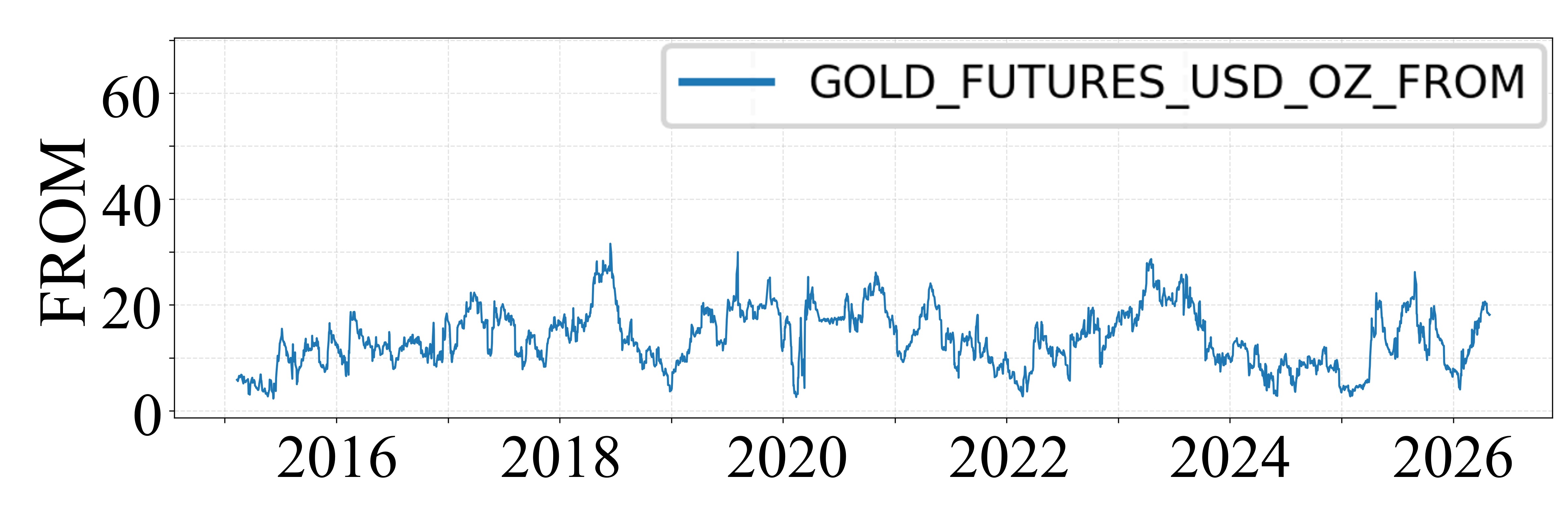}
 
{\small (e) Gold futures}
\end{minipage}
\hfill
\begin{minipage}{0.32\textwidth}
\centering
\includegraphics[width=\textwidth]{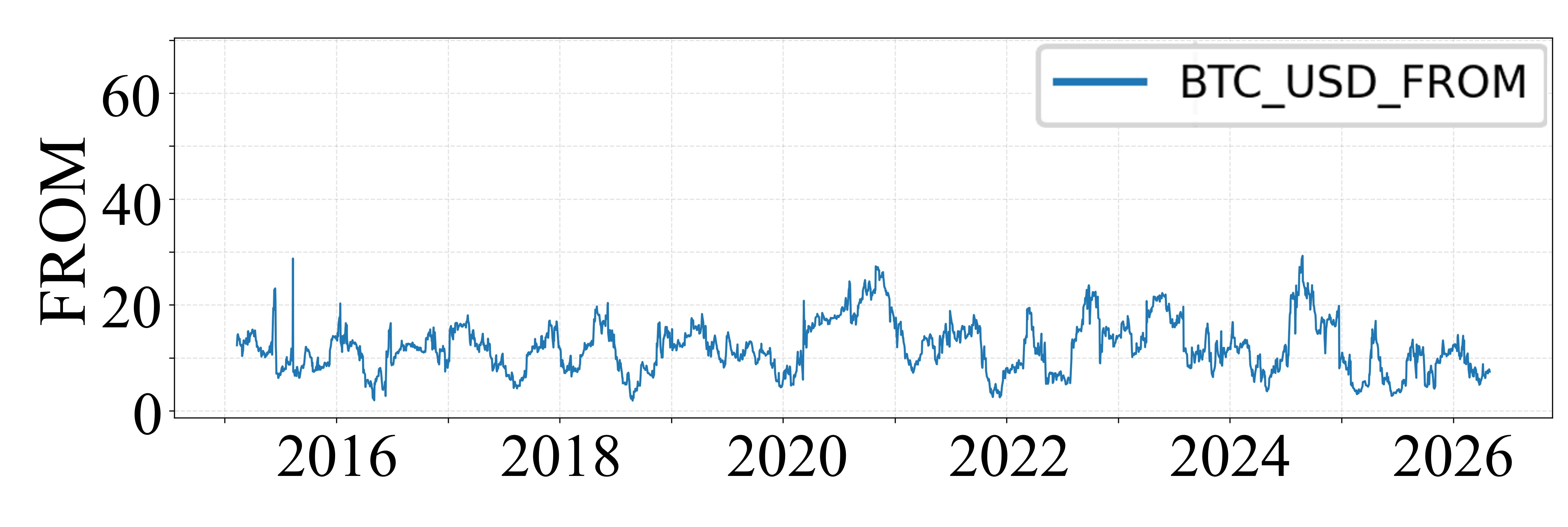}
 
{\small (f) BTC/USD}
\end{minipage}
 
\caption[Asset level FROM spillovers]{
Asset level FROM spillovers for the network excluding the dollar index under window size 100 and forecast horizon $H=10$. FROM measures spillovers received by each asset from other assets. The y-axis scale is kept identical across assets to allow direct comparison of exposure to incoming spillovers.
}
\label{fig:from_spillover_asset_level}
\end{figure}

Figure~\ref{fig:to_spillover_asset_level} reports TO spillovers, which measure the extent to which each asset transmits shocks to the other assets in the network. The most prominent pattern is the high and persistent transmission role of gold futures. Although several exchange rate variables and BTC/USD show temporary increases in TO spillovers, their direct transmission roles are less dominant and more time varying than that of gold futures. This confirms that gold futures are the main direct source of outward spillover transmission in the baseline network.

Figure~\ref{fig:from_spillover_asset_level} presents FROM spillovers, which measure the extent to which each asset receives shocks from other assets. The major exchange rate variables show relatively high FROM spillovers, indicating that currency pairs are important receivers of cross asset shocks in the baseline network. Gold futures and BTC/USD exhibit relatively lower FROM spillovers, suggesting that their exposure to incoming shocks from the rest of the network is more limited.

 
\begin{figure}[!htbp]
\centering
 
\begin{minipage}{0.32\textwidth}
\centering
\includegraphics[width=\textwidth]{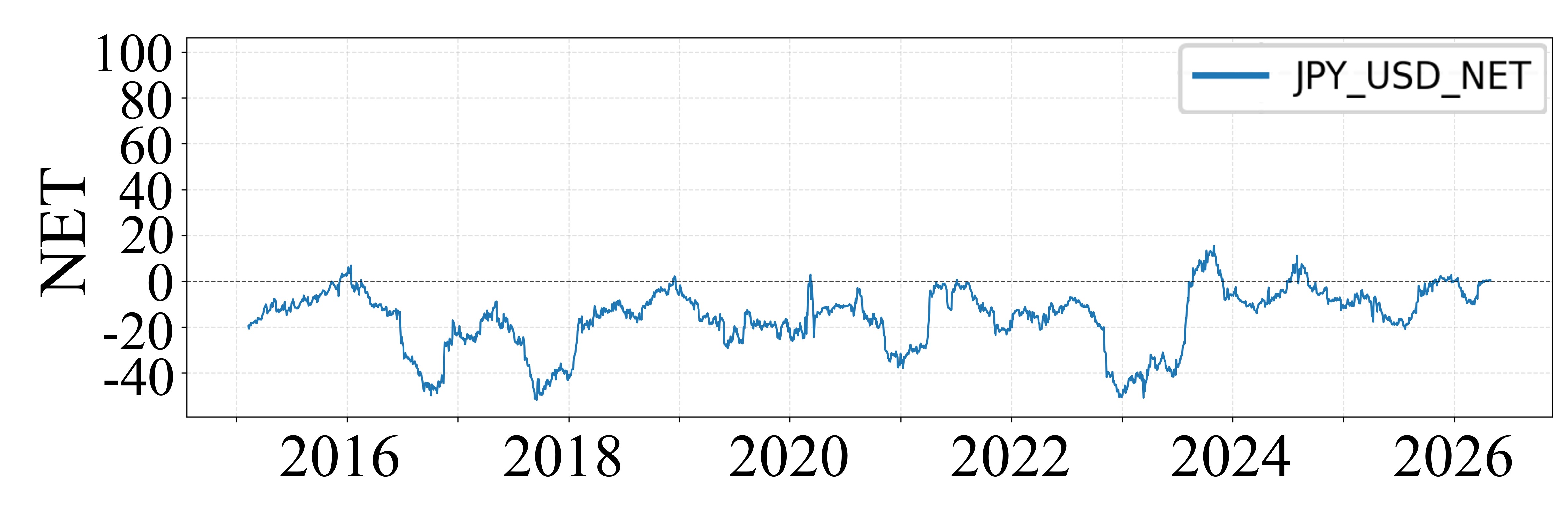}
 
{\small (a) JPY/USD}
\end{minipage}
\hfill
\begin{minipage}{0.32\textwidth}
\centering
\includegraphics[width=\textwidth]{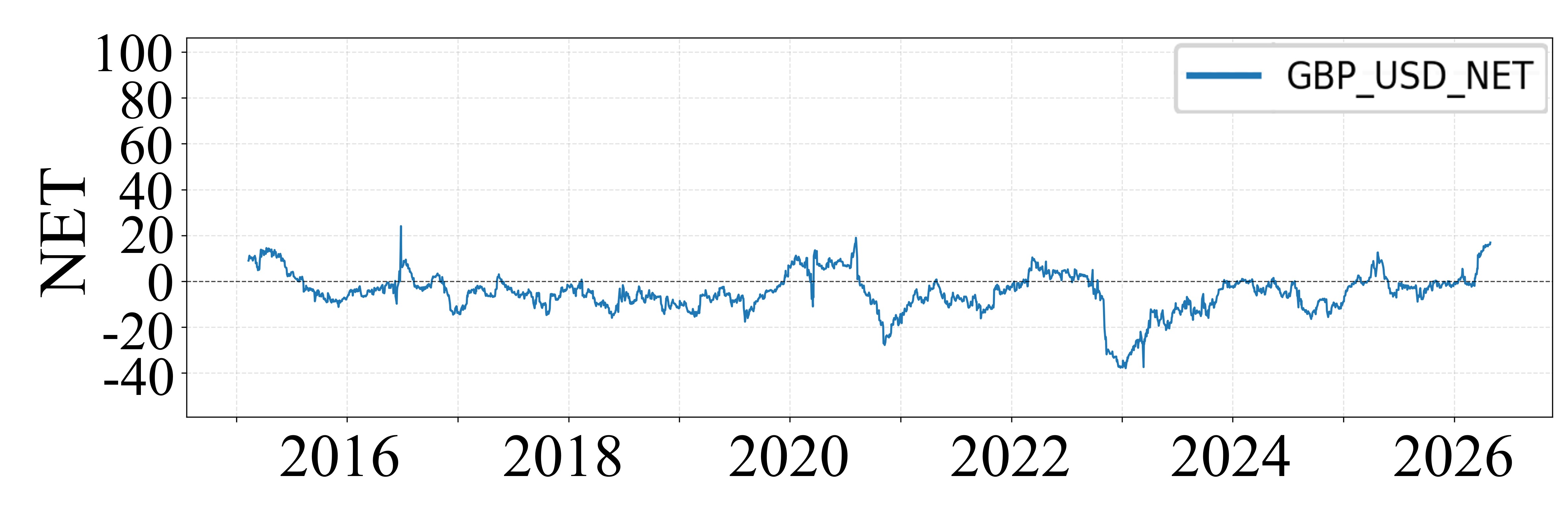}
 
{\small (b) GBP/USD}
\end{minipage}
\hfill
\begin{minipage}{0.32\textwidth}
\centering
\includegraphics[width=\textwidth]{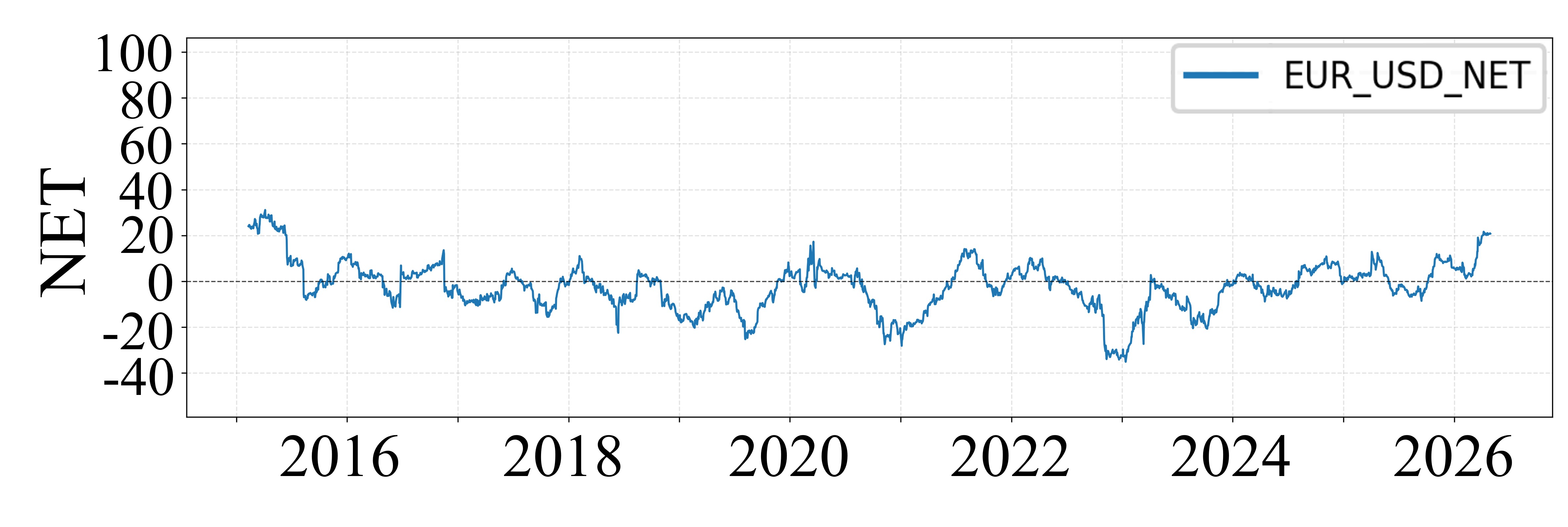}
 
{\small (c) EUR/USD}
\end{minipage}
 
\vspace{2mm}
 
\begin{minipage}{0.32\textwidth}
\centering
\includegraphics[width=\textwidth]{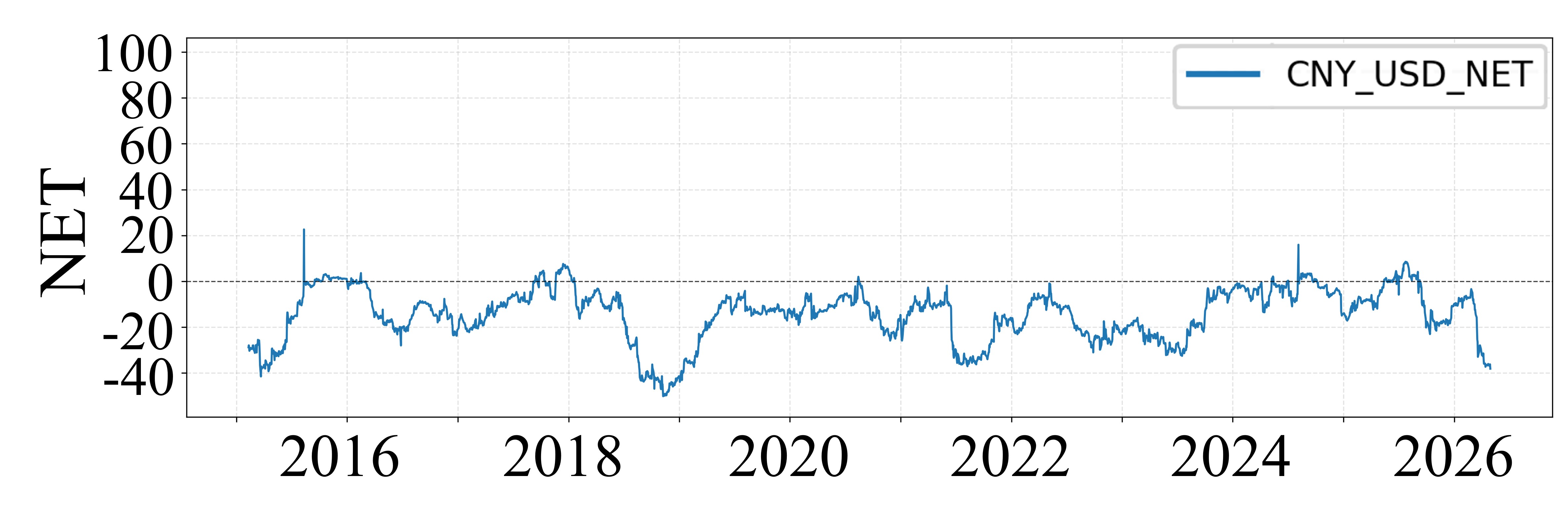}
 
{\small (d) CNY/USD}
\end{minipage}
\hfill
\begin{minipage}{0.32\textwidth}
\centering
\includegraphics[width=\textwidth]{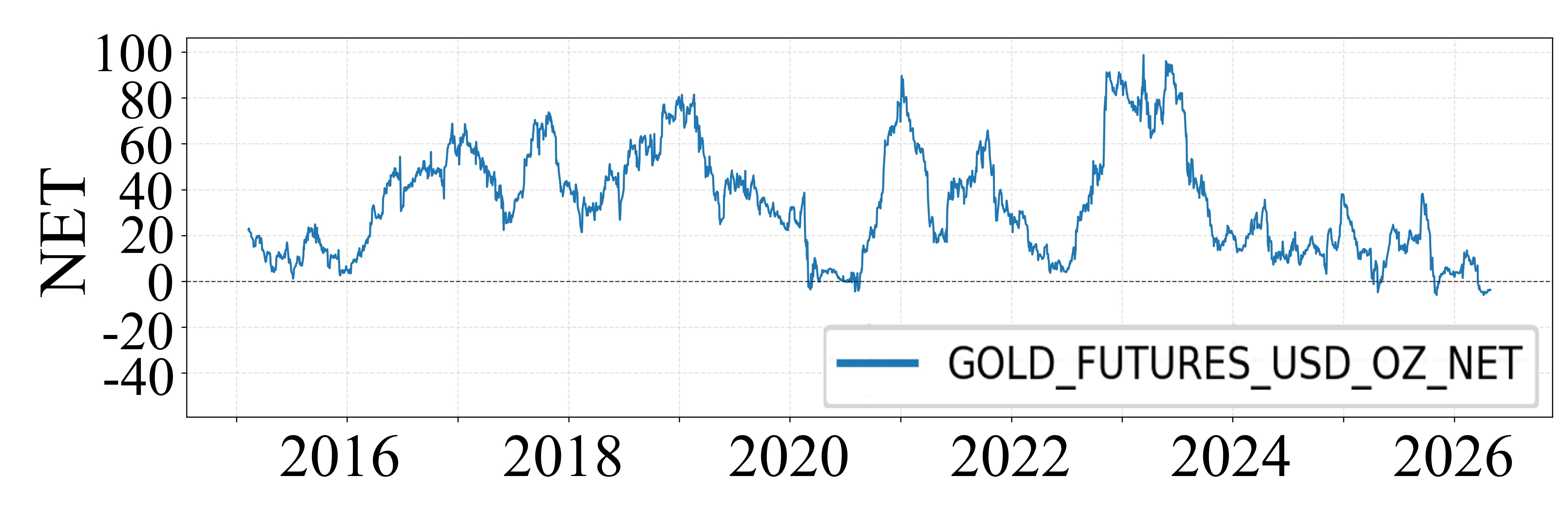}
 
{\small (e) Gold futures}
\end{minipage}
\hfill
\begin{minipage}{0.32\textwidth}
\centering
\includegraphics[width=\textwidth]{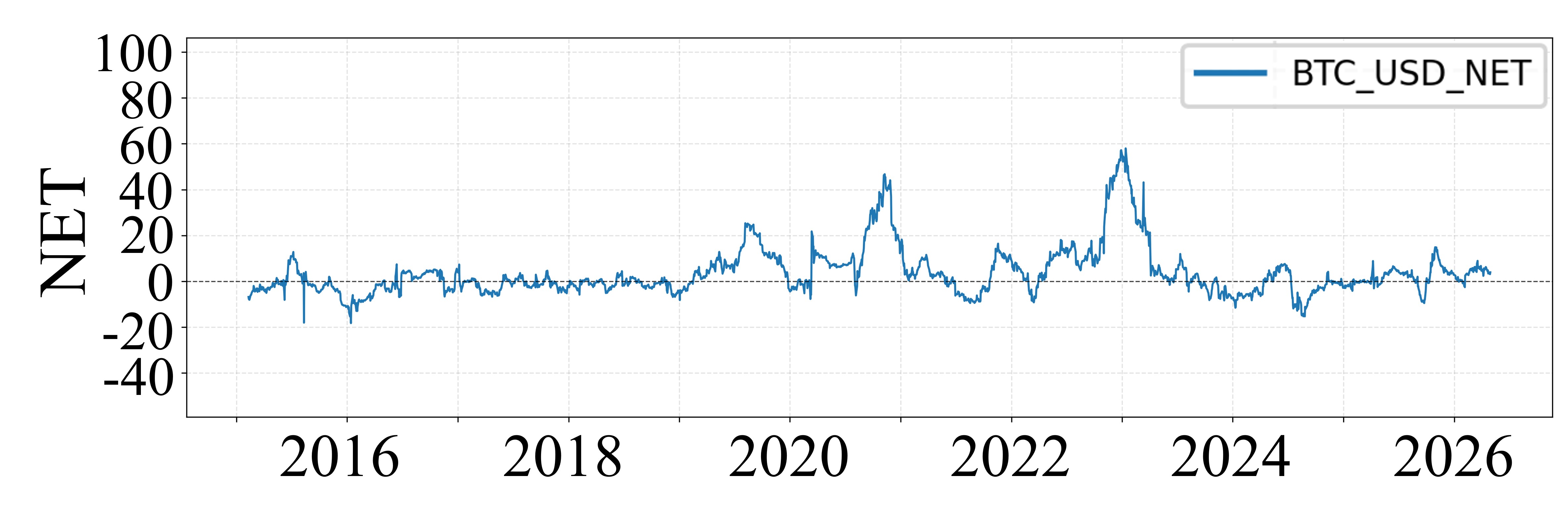}
 
{\small (f) BTC/USD}
\end{minipage}
 
\caption[Asset level NET spillovers]{
Asset level NET spillovers for the network excluding the dollar index under window size 100 and forecast horizon $H=10$. NET is defined as TO minus FROM. A positive NET value indicates a net transmitter, whereas a negative NET value indicates a net receiver. The y-axis scale is kept identical across assets to allow direct comparison of net transmission roles.
}
\label{fig:net_spillover_asset_level}
\end{figure}

Figure~\ref{fig:net_spillover_asset_level} reports NET spillovers, defined as TO minus FROM. Gold futures show positive NET spillovers for most of the sample period, confirming their role as the dominant direct net transmitter. CNY/USD and JPY/USD frequently display negative NET values, indicating clearer net receiving roles. EUR/USD and GBP/USD are closer to balanced but tend to remain on the receiver side. BTC/USD shows positive NET values in some periods, but it does not exhibit the persistent and dominant net transmitting pattern observed for gold futures.

Taken together, Figures~\ref{fig:to_spillover_asset_level}--\ref{fig:net_spillover_asset_level} show that direct spillovers in the baseline network are strongly asymmetric. Gold futures act as the dominant direct transmitter, whereas the major exchange-rate variables mainly function as receivers or near-balanced receivers of cross-asset shocks. Bitcoin occupies an intermediate position with a limited but nonzero transmitting role. This asymmetric topology provides the first layer of the asset-role classification developed in the main text.

\section{Distribution-based diffusion-state diagnostics}
\label{supsec:diffusion_state_diagnostics}

This section reports the detailed distribution-based diffusion-state
diagnostics that support the main-text summary. The diagnostics are based on
rolling deterministic Viral Centrality (VC) and the stationary departure
measure $F_{\mathrm{out}}$ for the network excluding the dollar index under
window size 100 and forecast horizon $H=10$. VC is a closure-dependent
fixed-point score for multistep upstream connectivity on the augmented reverse
source-tracing kernel, whereas $F_{\mathrm{out}}$ is the stationary
probability of occupying an asset and then moving to a distinct asset in that
chain. Neither measure is a forward causal shock-propagation estimate.

\begin{figure}[!t]
\centering
\captionsetup{font=small,skip=4pt}

\begin{minipage}{0.47\textwidth}
\centering
\includegraphics[width=\textwidth,height=0.105\textheight,keepaspectratio]{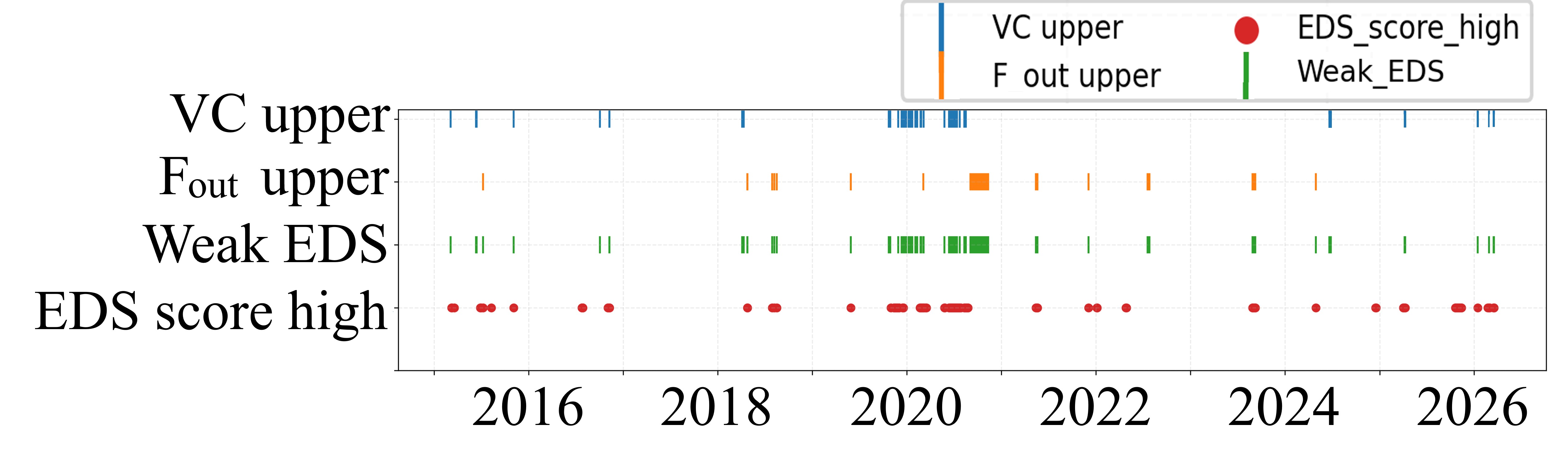}
{\small (a) BTC/USD}
\end{minipage}
\hfill
\begin{minipage}{0.47\textwidth}
\centering
\includegraphics[width=\textwidth,height=0.105\textheight,keepaspectratio]{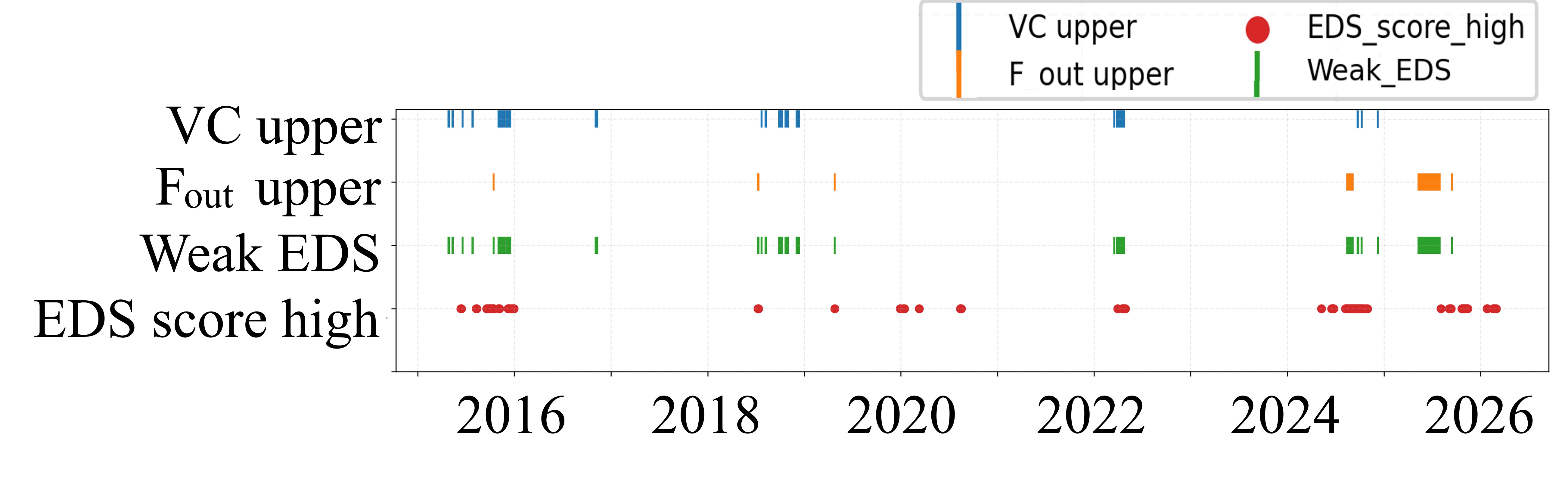}
{\small (b) CNY/USD}
\end{minipage}

\vspace{0.5mm}

\begin{minipage}{0.47\textwidth}
\centering
\includegraphics[width=\textwidth,height=0.105\textheight,keepaspectratio]{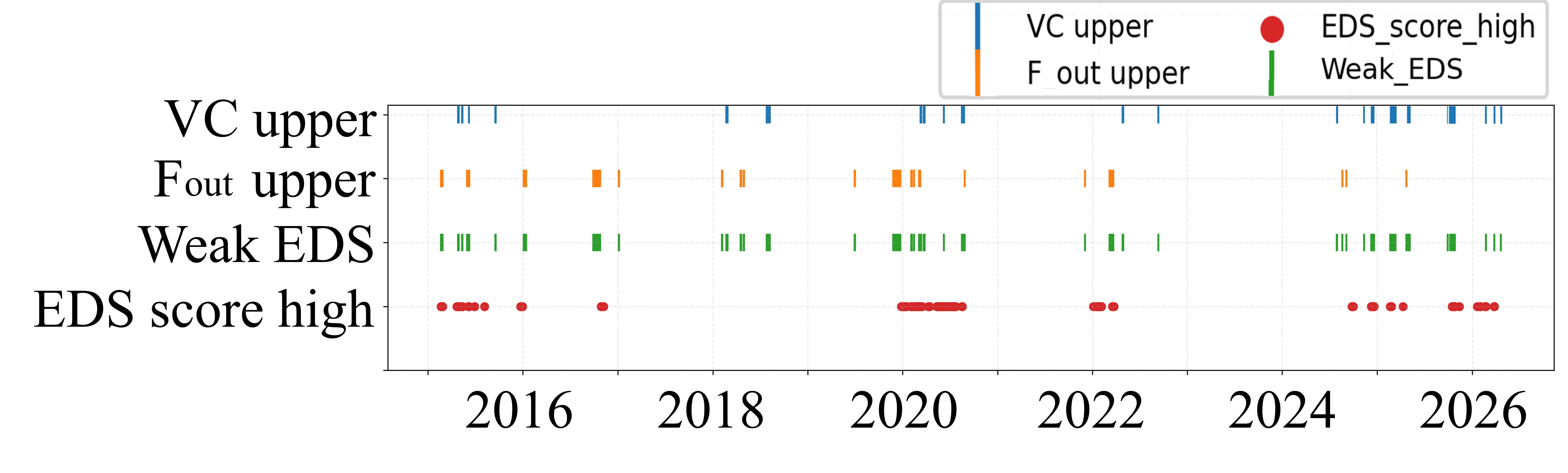}
{\small (c) EUR/USD}
\end{minipage}
\hfill
\begin{minipage}{0.47\textwidth}
\centering
\includegraphics[width=\textwidth,height=0.105\textheight,keepaspectratio]{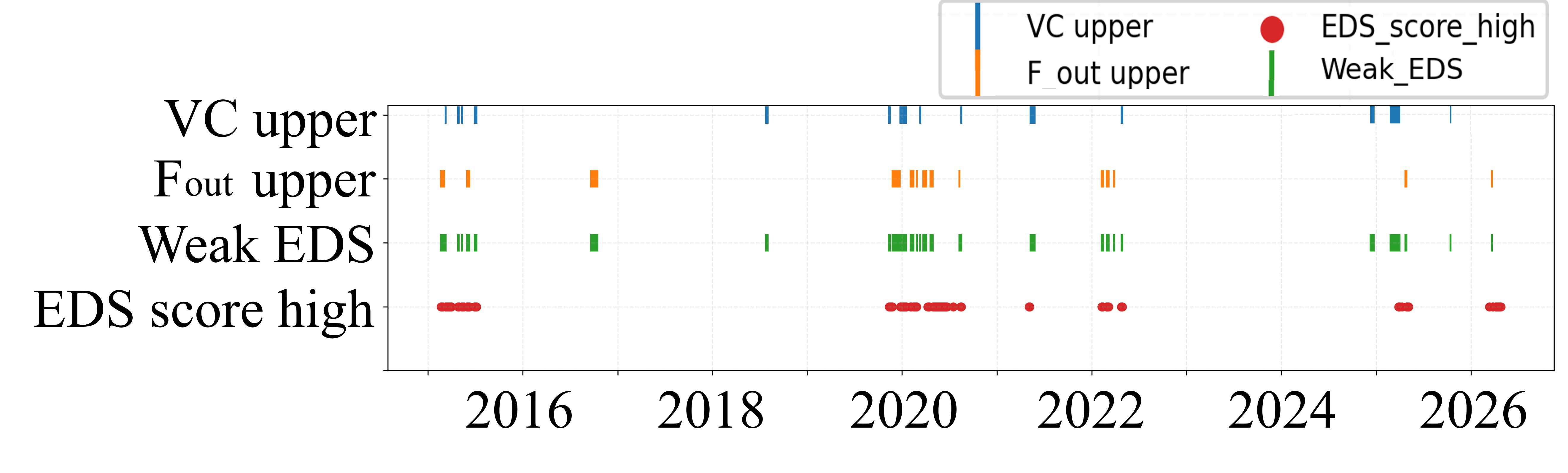}
{\small (d) GBP/USD}
\end{minipage}

\vspace{0.5mm}

\begin{minipage}{0.47\textwidth}
\centering
\includegraphics[width=\textwidth,height=0.105\textheight,keepaspectratio]{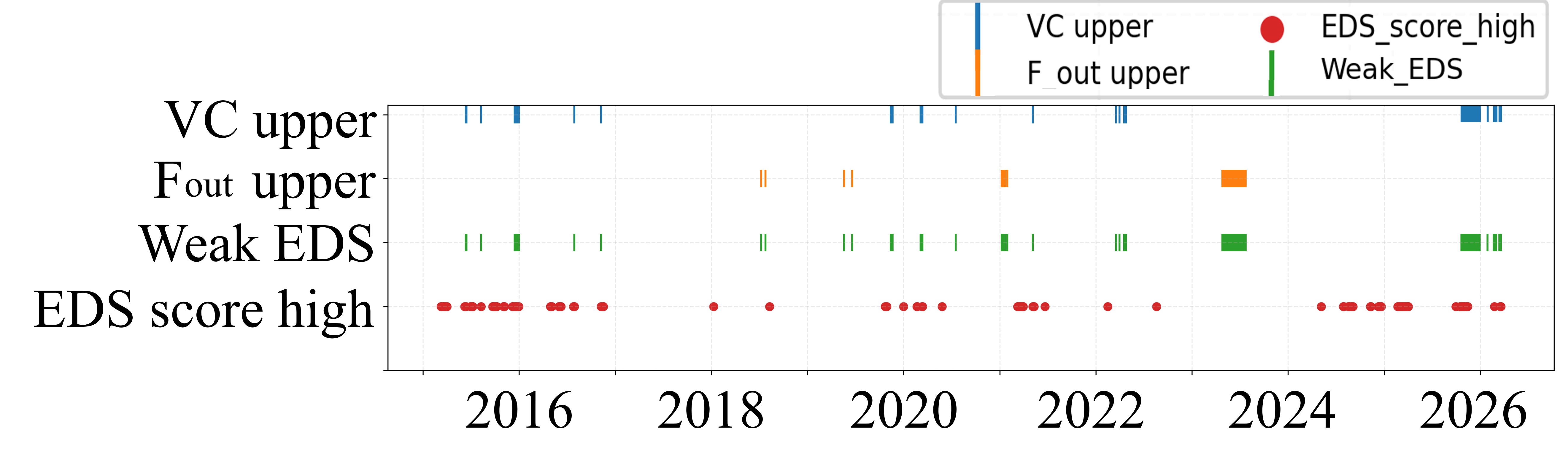}
{\small (e) Gold futures}
\end{minipage}
\hfill
\begin{minipage}{0.47\textwidth}
\centering
\includegraphics[width=\textwidth,height=0.105\textheight,keepaspectratio]{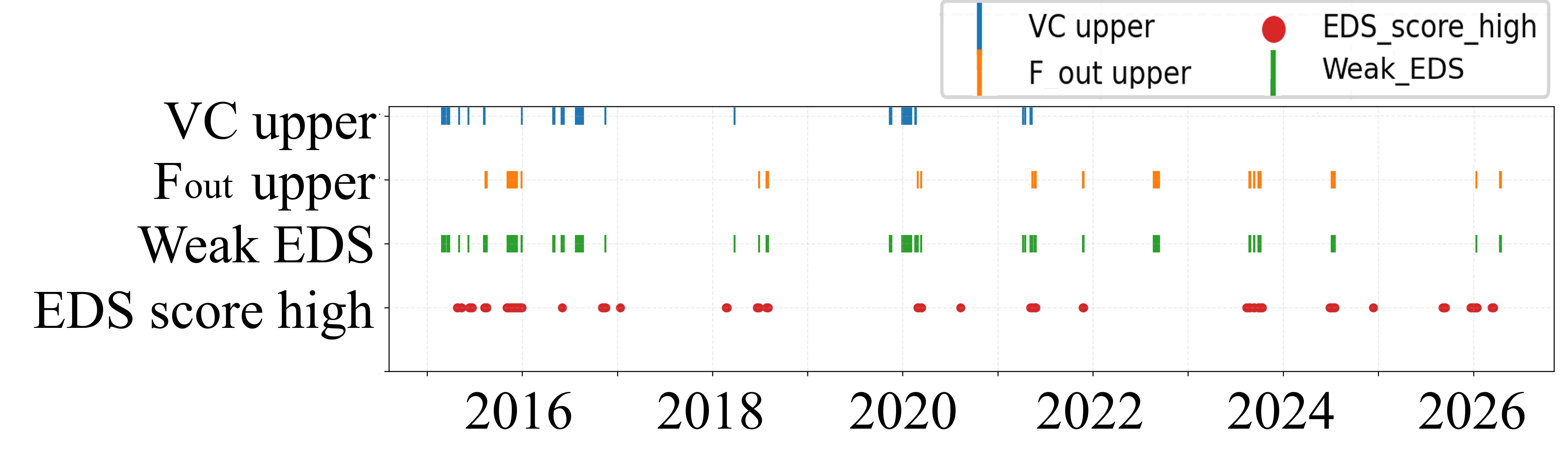}
{\small (f) JPY/USD}
\end{minipage}

\caption[Asset-level diffusion-state diagnostic exceedance timelines]{
Asset-level diffusion-state diagnostic exceedance timelines for the network excluding
the dollar index under window size 100 and forecast horizon $H=10$. For each
asset, the plotted rows indicate the rolling windows in which VC upper
exceedance, $F_{\mathrm{out}}$ upper exceedance, Weak EDS, and
EDS\_score\_high occur. Joint EDS is omitted from the plotted rows
because no Joint EDS events are observed for any asset over the sample
period.

Because the thresholds are
in-sample quantiles, the marginal exceedance frequencies are fixed by
construction at approximately 2.5\% for VC and $F_{\mathrm{out}}$ and 5\%
for the score-based indicator. The empirical content of the figure is
therefore the timing of the exceedance episodes. The absence of coincident VC
and $F_{\mathrm{out}}$ upper exceedances is a descriptive in-sample result,
not evidence of independence or the absence of contagion.
}
\label{figS:weak_eds_timelines}
\end{figure}

\vspace{-1mm}

Figure~\ref{figS:weak_eds_timelines} displays, for each asset, the windows in
which the individual diagnostic indicators fire. Because the marginal
exceedance frequencies are pinned down by the in-sample quantile thresholds,
the informative dimension of this figure is temporal: it shows when each
asset's diagnostic enters its own upper tail, not how often.

Two features stand out. First, consistent with the zero Joint EDS frequency,
the VC and $F_{\mathrm{out}}$ upper-exceedance marks never occur in the same
window for any asset, so the two diagnostics enter their upper tails through
disjoint episodes in this sample. Second, the exceedances appear as runs of
adjacent windows rather than isolated dates, reflecting the fact that
consecutive rolling windows share $W-1$ observations. These patterns are
descriptive and do not constitute a test of independence or causal contagion.

The timelines also reveal asset-level differences in the type and timing of
the episodes. BTC/USD shows concentrated episodes of Weak EDS and
EDS\_score\_high in particular periods, indicating that Bitcoin's combined
diagnostic score becomes high relative to its own history only under specific
estimated network conditions. Weak EDS episodes
are also observed for several exchange-rate variables, including JPY/USD and
CNY/USD, which is consistent with the preceding VC results showing that some
exchange-rate nodes can have relatively high deterministic upstream-connectivity
scores even when they are not dominant direct transmitters. This does not
identify these assets as forward cascade conduits. Gold futures exhibit
$F_{\mathrm{out}}$ upper exceedances in some periods, but these never
coincide with VC upper exceedances. Thus, their prominent direct-transmission
and reverse-chain stationary-departure roles do not imply simultaneous high
deterministic upstream connectivity.

\begin{figure}[!htbp]
\centering

\begin{minipage}{0.48\textwidth}
\centering
\includegraphics[width=\textwidth,height=0.105\textheight,keepaspectratio]{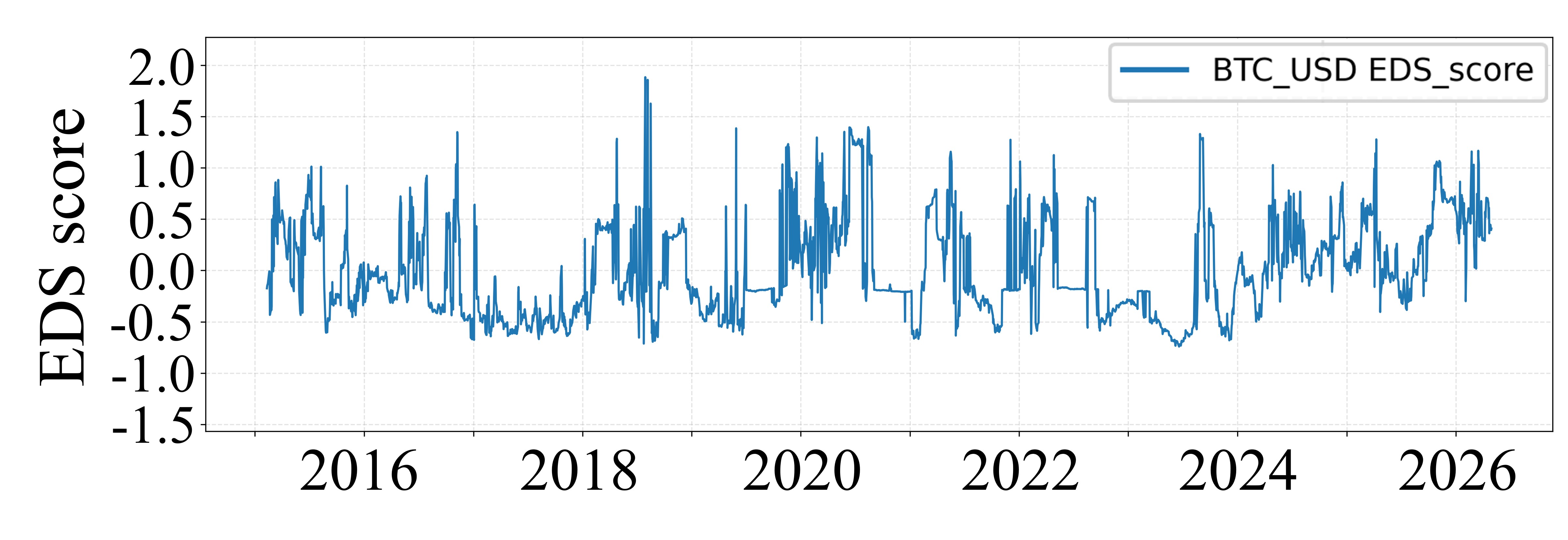}

{\small (a) BTC/USD}
\end{minipage}
\hfill
\begin{minipage}{0.48\textwidth}
\centering
\includegraphics[width=\textwidth,height=0.105\textheight,keepaspectratio]{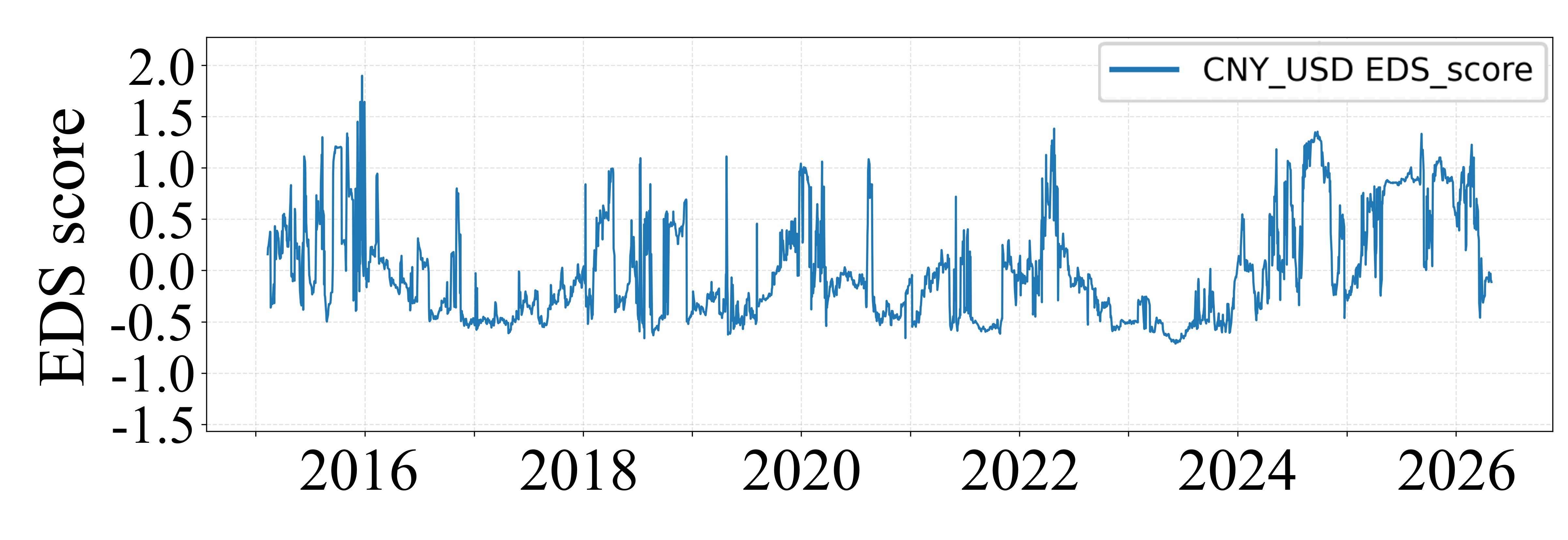}

{\small (b) CNY/USD}
\end{minipage}

\vspace{0.5mm}

\begin{minipage}{0.48\textwidth}
\centering
\includegraphics[width=\textwidth,height=0.105\textheight,keepaspectratio]{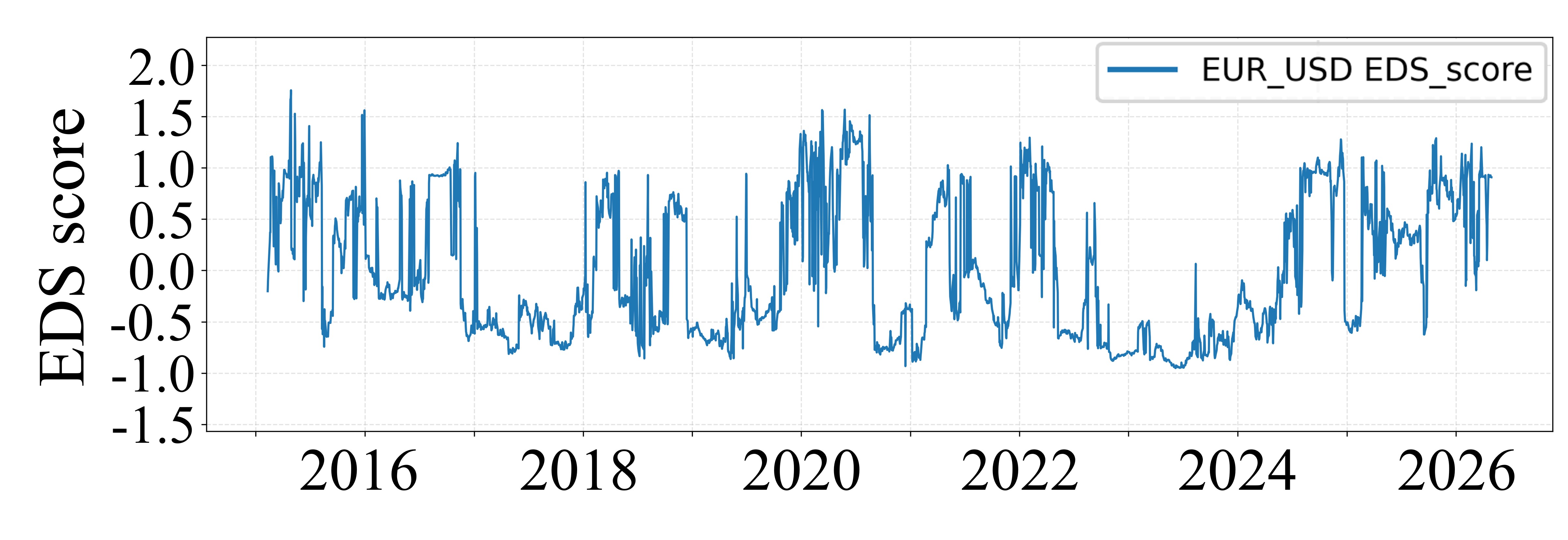}

{\small (c) EUR/USD}
\end{minipage}
\hfill
\begin{minipage}{0.48\textwidth}
\centering
\includegraphics[width=\textwidth,height=0.105\textheight,keepaspectratio]{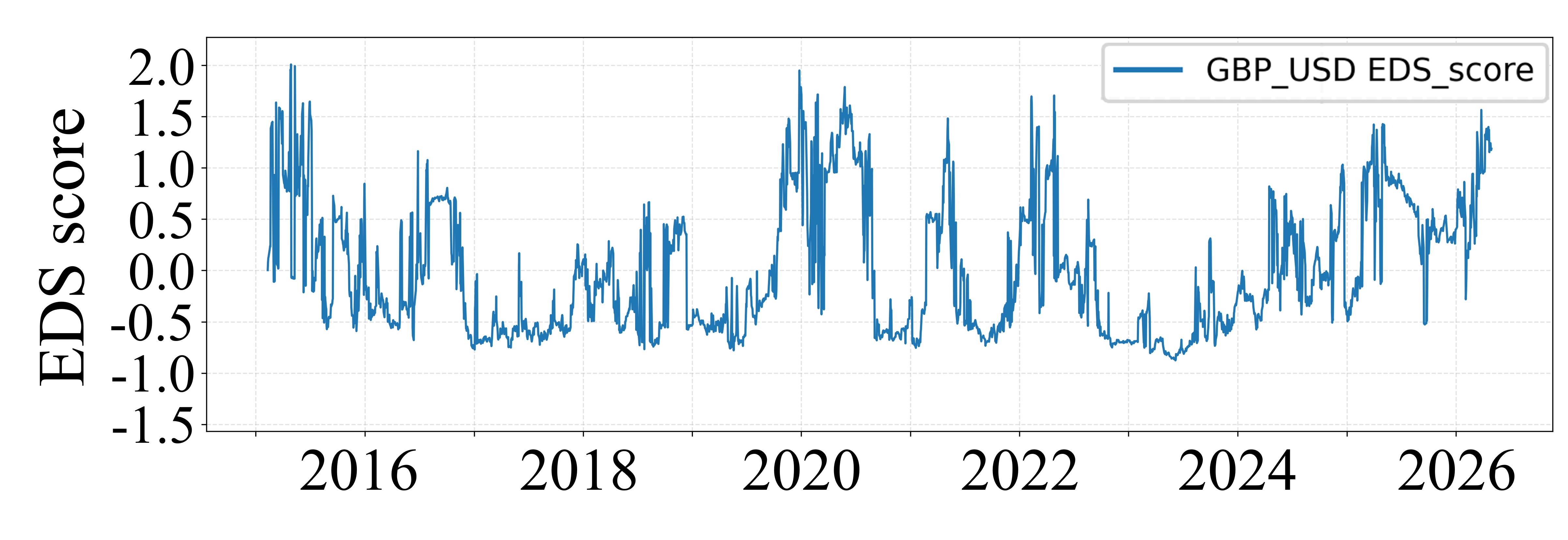}

{\small (d) GBP/USD}
\end{minipage}

\vspace{0.5mm}

\begin{minipage}{0.48\textwidth}
\centering
\includegraphics[width=\textwidth,height=0.105\textheight,keepaspectratio]{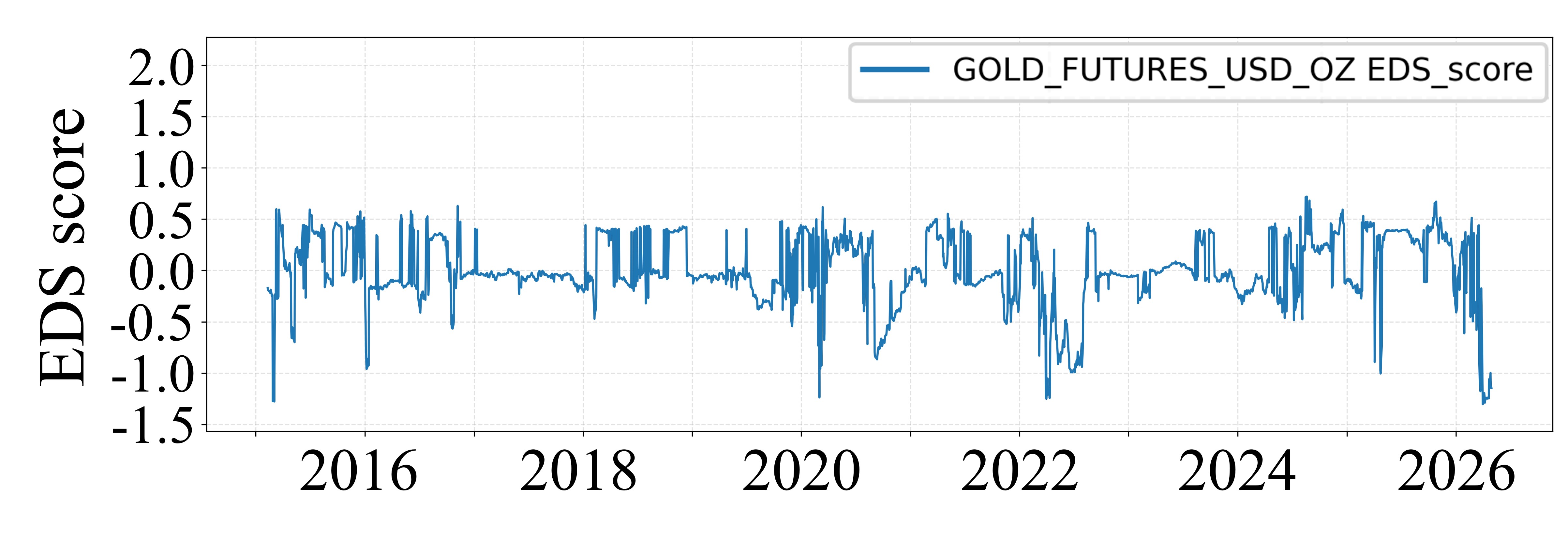}

{\small (e) Gold futures}
\end{minipage}
\hfill
\begin{minipage}{0.48\textwidth}
\centering
\includegraphics[width=\textwidth,height=0.105\textheight,keepaspectratio]{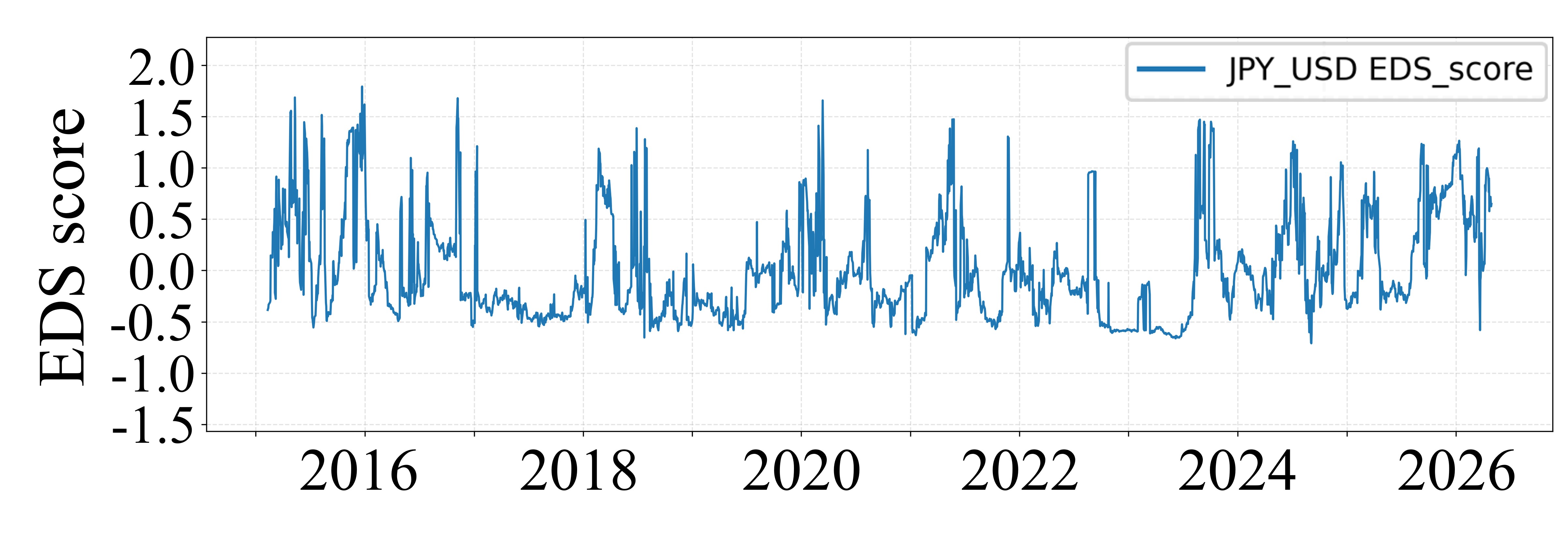}

{\small (f) JPY/USD}
\end{minipage}

\caption[Asset-level EDS\_score time series]{
Asset-level EDS\_score time series for the network excluding the dollar index
under window size 100 and forecast horizon $H=10$. EDS\_score combines standardized deterministic VC and standardized
reverse-chain stationary departure $F_{\mathrm{out}}$. A higher value
indicates a higher combined standardized level of the two diagnostics
relative to the asset's own in-sample distributions; it does not require
both diagnostics to be individually high. It is not a measure of forward
shock propagation, return performance, or causal contagion.
}
\label{figS:eds_score_asset_heterogeneity}
\end{figure}

Figure~\ref{figS:eds_score_asset_heterogeneity} reports the time-varying
EDS\_score for each asset. Since EDS\_score combines standardized VC and
standardized $F_{\mathrm{out}}$, a higher value indicates that the joint
diagnostic score of an asset increases relative to its own empirical
distribution. It is not a measure of return performance, forward shock spread,
or observed contagion.

The EDS\_score time-series results show asset-level variation in the
intensity and timing of the standardized diagnostics. GBP/USD, CNY/USD,
BTC/USD, JPY/USD, and EUR/USD record maximum EDS\_score values of
approximately 2.008, 1.898, 1.884, 1.794, and 1.757, respectively. These
values indicate temporary increases in the combined standardized diagnostics.
They do not mean that these assets become persistent direct net transmitters.

The EDS\_score pattern of gold futures is relatively smoother than those of
several other assets. This does not imply that gold futures have a weak direct
or stationary role. Gold futures are dominant in direct NET transmission and
reverse-chain stationary departure, whereas their deterministic VC is
relatively low. The score therefore combines two differently constructed
diagnostics.

Taken together, Figures~\ref{figS:weak_eds_timelines}
and~\ref{figS:eds_score_asset_heterogeneity} show two complementary aspects
of the diffusion-state diagnostics. Figure~\ref{figS:weak_eds_timelines} shows
that rolling VC and rolling $F_{\mathrm{out}}$ never enter their upper tails
in the same window for any asset, whereas
Figure~\ref{figS:eds_score_asset_heterogeneity} illustrates time-varying
changes in the combined standardized score. 
These results show that the
diffusion-state diagnostics vary over time and cannot be reduced
to a single direct transmitter--receiver classification.

Finally, the in-sample coverage of the 2.5th--97.5th percentile band is
approximately 95\% by construction, up to ties and discreteness; we report
it only as a consistency check of the threshold computation, not as evidence
of predictive calibration.

\section{Auxiliary U.S. Dollar Index network}
\label{supsec:dollar_index_network}

The auxiliary network including the U.S. Dollar Index is examined
under the same baseline setting of window size 100 and forecast horizon
$H=10$. This section reports the rolling TCI comparison and the detailed
asset-level directional and source-tracing results. All Markov and VC
quantities remain conditional on row normalization and the uniform restart
closure.

 
\begin{figure}[H]
\centering
\includegraphics[width=0.6\textwidth]{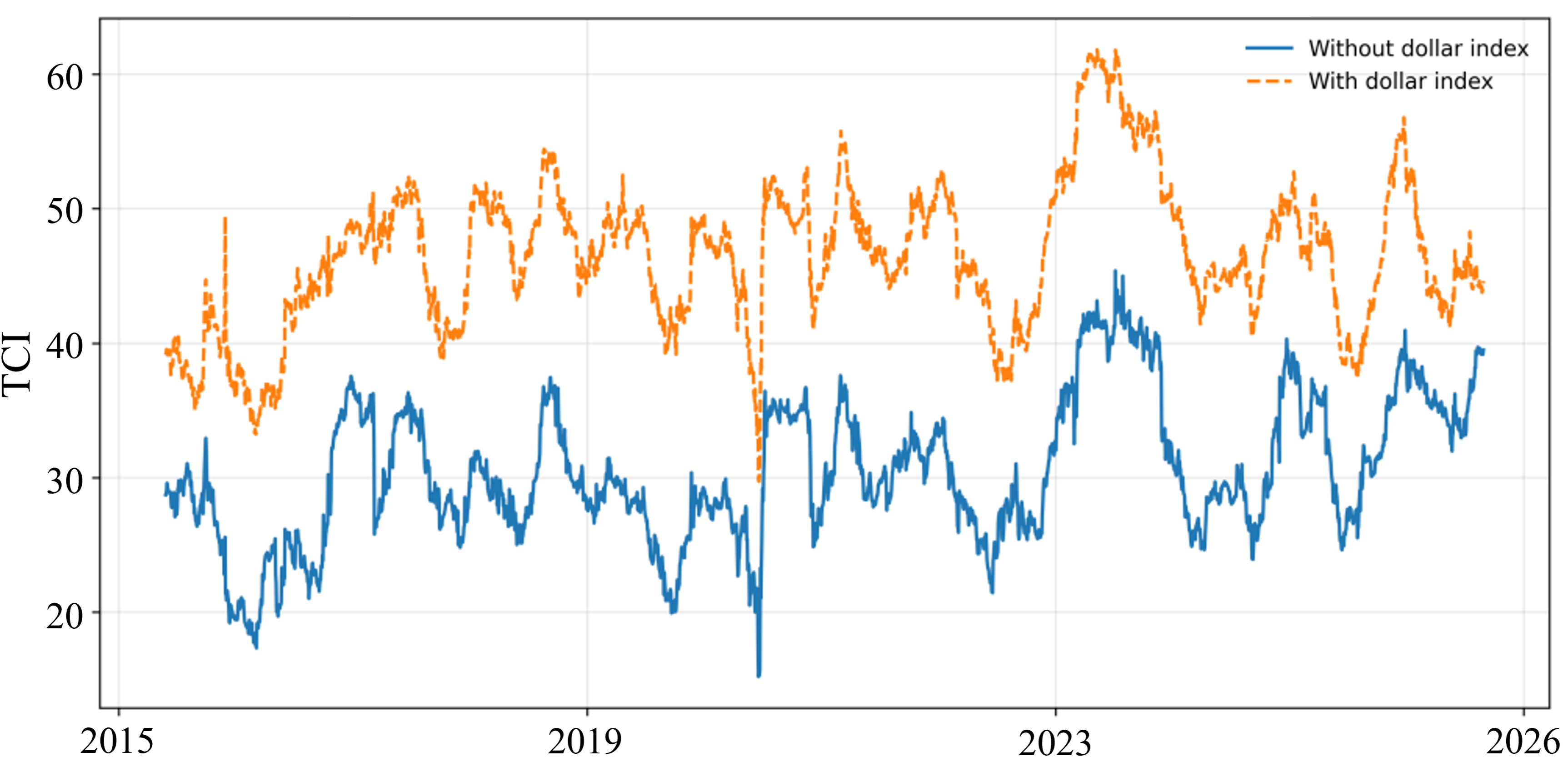}
\caption{Rolling TCI comparison between the networks excluding and including the U.S. Dollar Index. The figure reports the time-varying Total Connectedness Index under the same baseline setting of window size 100 and forecast horizon $H=10$.}
\label{figS:rolling_tci_with_without_dollar}
\end{figure}

Figure~\ref{figS:rolling_tci_with_without_dollar} shows that the network including the U.S. Dollar Index exhibits a higher level of system-wide connectedness than the baseline network excluding the dollar index. However, both networks display time-varying connectedness, indicating that cross-asset spillover intensity changes substantially over time regardless of whether the broad dollar factor is explicitly modeled.

\begin{table}[H]
\centering
\caption{Directional spillover comparison between networks excluding and including the U.S. Dollar Index}
\label{tabS:directional_spillover_with_without_dollar}
\scriptsize
\resizebox{0.70\textwidth}{!}{%
\begin{tabular}{l c c c c c c}
\toprule
 & \multicolumn{3}{c}{Without dollar index} 
 & \multicolumn{3}{c}{With dollar index} \\
\cmidrule(lr){2-4} \cmidrule(lr){5-7}
Asset & TO & FROM & NET & TO & FROM & NET \\
\midrule
U.S. Dollar Index & -- & -- & -- & 155.03 & 31.59 & 123.45 \\
Gold futures      & 47.51 & 13.71 & 33.81 & 66.59 & 31.50 & 35.10 \\
BTC/USD           & 15.49 & 11.95 & 3.54  & 25.46 & 21.20 & 4.25 \\
EUR/USD           & 43.98 & 46.38 & -2.40 & 23.26 & 84.66 & -61.40 \\
GBP/USD           & 36.68 & 41.58 & -4.90 & 17.58 & 62.17 & -44.59 \\
CNY/USD           & 16.14 & 30.47 & -14.33 & 16.76 & 40.30 & -23.55 \\
JPY/USD           & 22.84 & 38.55 & -15.71 & 21.13 & 54.39 & -33.27 \\
\bottomrule
\end{tabular}
}
\vspace{1mm}
\begin{minipage}{0.95\textwidth}
\small
\textit{Note}: 
TO denotes the total directional spillover transmitted by each asset to the other assets in the network, whereas FROM denotes the total directional spillover received from the other assets. NET is defined as TO minus FROM. The values are based on rolling GFEVD
estimates under window size 100 and forecast horizon $H=10$. NET is
computed from full-precision TO and FROM estimates; therefore, minor
discrepancies of 0.01 may arise when subtracting the rounded TO and FROM
values reported in the table. Dashes indicate that the U.S. Dollar Index
is not included in the baseline network excluding the dollar index.
\end{minipage}
\end{table}

Table~\ref{tabS:directional_spillover_with_without_dollar} shows that the direct transmission hierarchy changes when the dollar index is included. In the baseline network excluding the dollar index, gold futures are the strongest direct net transmitter, with a NET spillover of 33.81. In the auxiliary network including the dollar index, the U.S. Dollar Index becomes the dominant net transmitter, with a NET spillover of 123.45. Under this auxiliary specification, the dollar index becomes a dominant estimated node and changes the direction and intensity of the estimated direct spillovers; this comparison does not establish a structural causal mechanism. Nevertheless, gold futures continue to show a positive NET spillover and remain an important transmitter after the dollar index is included.
 
 
\begin{table}[H]
\centering
\caption{Source-tracing role comparison between networks excluding and including the U.S. Dollar Index}
\label{tabS:diffusion_role_with_without_dollar}
\scriptsize
\resizebox{0.65\textwidth}{!}{%
\begin{tabular}{l l c c}
\toprule
Specification & Asset & Mean $F_{\mathrm{out}}$ & Mean deterministic VC \\
\midrule
Without dollar index & Gold futures & 0.312 & 1.690 \\
Without dollar index & BTC/USD & 0.169 & 2.088 \\
Without dollar index & EUR/USD & 0.157 & 2.155 \\
Without dollar index & GBP/USD & 0.133 & 2.134 \\
Without dollar index & CNY/USD & 0.093 & 2.225 \\
Without dollar index & JPY/USD & 0.091 & 2.248 \\
\midrule
With dollar index & U.S. Dollar Index & 0.304 & 1.722 \\
With dollar index & Gold futures & 0.200 & 1.981 \\
With dollar index & BTC/USD & 0.152 & 2.137 \\
With dollar index & EUR/USD & 0.082 & 2.320 \\
With dollar index & JPY/USD & 0.079 & 2.235 \\
With dollar index & CNY/USD & 0.071 & 2.292 \\
With dollar index & GBP/USD & 0.068 & 2.304 \\
\bottomrule
\end{tabular}
}
\vspace{1mm}
\begin{minipage}{0.95\textwidth}
\small
\textit{Note}: 
$F_{\mathrm{out},i}=\pi_i(1-P_{ii})$ is the stationary probability of
occupying asset $i$ and then moving to a distinct asset in the augmented
reverse source-tracing chain. VC is computed deterministically on the full
augmented kernel, including restart and diagonal transitions. It is a
closure-dependent approximation to multistep upstream connectivity, not a
Monte Carlo estimate or an exact independent-cascade expectation. The values
are time averages of rolling estimates under window size 100 and forecast
horizon $H=10$.
\end{minipage}
\end{table}

Table~\ref{tabS:diffusion_role_with_without_dollar} shows that the rankings depend on the network summary considered. In the baseline network, gold futures have the highest mean $F_{\mathrm{out}}$ and stationary occupation probability, making them the leading stationary-departure node in the selected reverse chain. After the U.S. Dollar Index is added, the highest mean $F_{\mathrm{out}}$ and stationary occupation shift to the dollar index. This result is conditional on the selected orientation, normalization, and restart closure. Among these rolling means, the highest deterministic VC is observed for JPY/USD in the baseline network and for EUR/USD in the auxiliary network. Because the kernel is reverse oriented, these values indicate broad modeled connectivity to upstream net-spillover suppliers; they do not indicate forward shock spreading through cascade paths.

 
The qualitative conclusion is that direct transmission, reverse-chain stationary departure, and deterministic upstream connectivity assign different roles. The baseline specification examines the conditional source-tracing topology among exchange rates, gold futures, and Bitcoin, while the auxiliary specification shows how the broad dollar factor changes that estimated topology when it is included explicitly.

\section{PageRank-regularized source-tracing robustness}
\label{supsec:pagerank_robustness}

To examine whether the stationary-departure ranking depends on the
unregularized stationary distribution of $\mathbf{P}$, we use the
PageRank-regularized kernel defined in the main text with $\alpha=0.85$.
For comparability with $F_{\mathrm{out},i}=\pi_i(1-P_{ii})$, the reported
score is
\[
PR_{\mathrm{out},i}=PR_i(1-P_{ii}).
\]
Thus, $PR_{\mathrm{out}}$ is a hybrid PageRank-weighted original-kernel
departure score. It is not the stationary off-diagonal flux of the PageRank
kernel $\mathbf{G}_{\alpha}$, which would instead equal
$PR_i(1-G_{\alpha,ii})$.

Figure~\ref{figS:stationary_pagerank_scatter} compares the asset-level
rankings obtained from stationary departure $F_{\mathrm{out}}$ and the hybrid
PageRank-weighted score $PR_{\mathrm{out}}$ under the
baseline specification excluding the U.S. Dollar Index, with window size
$W=100$ and forecast horizon $H=10$. The two rankings are identical. Gold
futures remain the leading reverse-chain stationary-departure node under both
measures, while BTC/USD retains an intermediate position. The major
exchange-rate variables remain below gold futures and BTC/USD under both
scores.

\begin{figure}[H]
\centering
\includegraphics[width=0.6\textwidth]{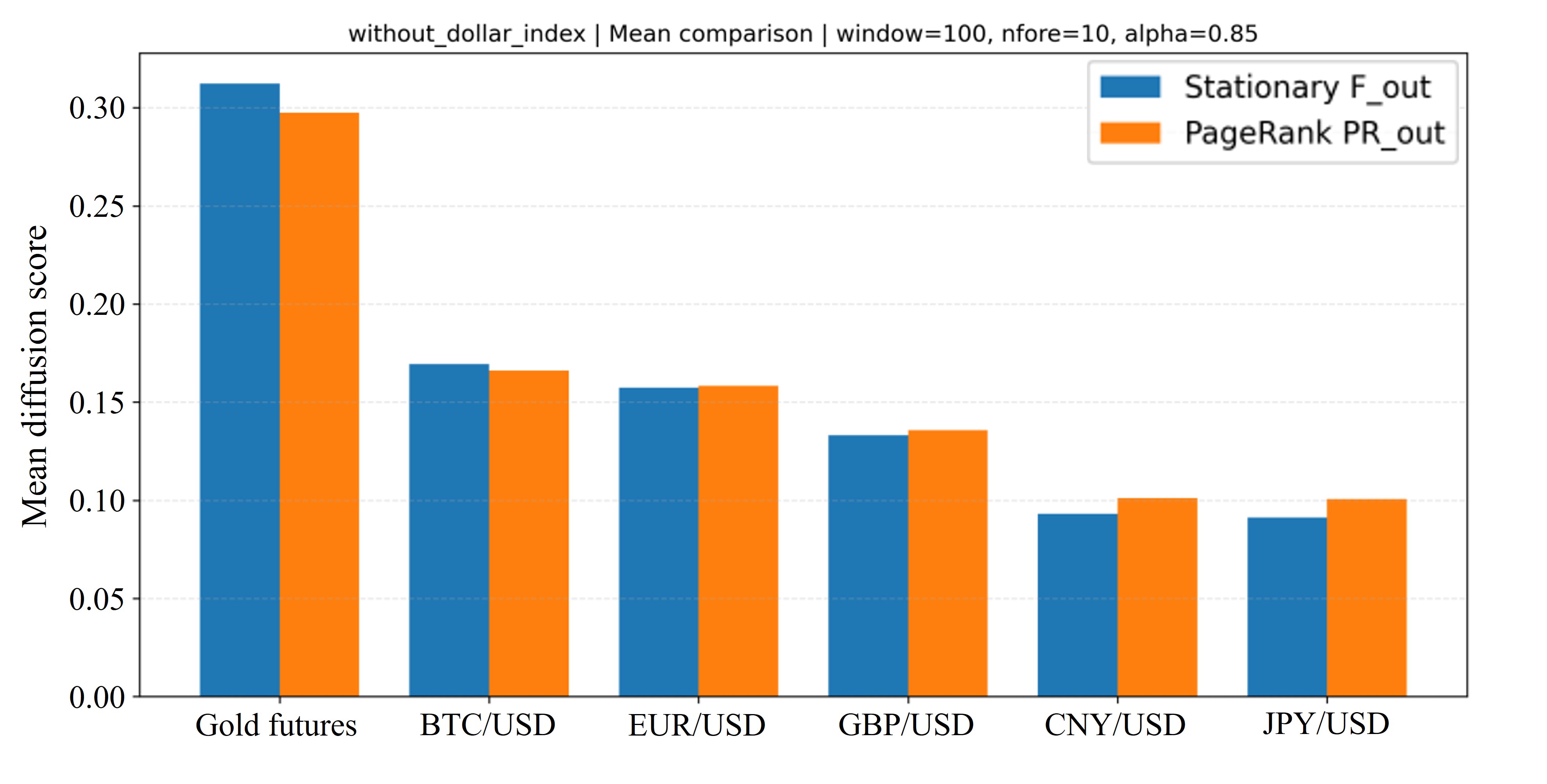}
\caption{Comparison of average reverse-chain stationary departure $F_{\mathrm{out}}$ and the hybrid PageRank-weighted original-kernel departure score $PR_{\mathrm{out}}$. The figure is based on the baseline network excluding the U.S. Dollar Index, with $W=100$, $H=10$, and $\alpha=0.85$. The similar values and identical ranking show that the stationary-departure ordering is robust to PageRank-regularized occupancy weighting. Neither score represents forward economic shock outflow.}
\label{figS:stationary_pagerank_scatter}
\end{figure}

Figure~\ref{figS:stationary_pagerank_scatter} further compares the average
values of $F_{\mathrm{out}}$ and $PR_{\mathrm{out}}$. For every asset, the
paired bars for the two measures are nearly indistinguishable. In the
baseline specification, the Spearman rank correlation and Pearson
correlation between the two scores are both equal to 1.000 after rounding to
three decimals. PageRank regularization therefore changes the scale only
slightly while preserving the reported asset ordering.

To evaluate whether this similarity is specific to the baseline setting, we
repeat the comparison over alternative rolling-window sizes
$W\in\{52,100,150\}$ and forecast horizons $H\in\{5,10,15\}$.
Table~\ref{tabS:stationary_pagerank_robustness} summarizes the correlations
between the two source-tracing departure scores across all specifications.

\begin{table}[H]
\centering
\caption{Robustness of stationary-departure and PageRank-weighted departure rankings}
\label{tabS:stationary_pagerank_robustness}
\small
\begin{tabular}{lcccc}
\toprule
Asset set & Minimum & Mean & Median & Maximum \\
\midrule
Including U.S. Dollar Index  & 1.000 & 1.000 & 1.000 & 1.000 \\
Excluding U.S. Dollar Index  & 1.000 & 1.000 & 1.000 & 1.000 \\
\bottomrule
\end{tabular}

\vspace{1mm}

\begin{minipage}{0.92\textwidth}
\small
\textit{Note}: The table reports Spearman rank correlations between the
average reverse-chain stationary-departure ranking based on
$F_{\mathrm{out}}$ and the average hybrid PageRank-weighted original-kernel
departure ranking based on $PR_{\mathrm{out},i}=PR_i(1-P_{ii})$ across alternative
rolling-window sizes $W\in\{52,100,150\}$ and forecast horizons
$H\in\{5,10,15\}$, with $\alpha=0.85$.
\end{minipage}
\end{table}

\section{Full Pairwise Connectedness Matrix (GFEVD)}
\label{supsec:gfevd_matrix}
\begin{table}[H]
\centering
\caption{Full pairwise connectedness matrix (GFEVD)}
\label{tabS:gfevd_full_matrix}
\small
\setlength{\tabcolsep}{6pt}
\renewcommand{\arraystretch}{1.15}
\begin{adjustbox}{max width=\textwidth}
\begin{tabular}{lrrrrrr}
\toprule
 & JPY/USD & GBP/USD & EUR/USD & CNY/USD & Gold futures & BTC/USD \\
\midrule
JPY/USD      & 73.34 & 4.03  & 11.36 & 0.93  & 10.34 & 0.00 \\
GBP/USD      & 3.59  & 65.18 & 23.59 & 3.26  & 3.70  & 0.67 \\
EUR/USD      & 9.25  & 21.79 & 60.21 & 2.84  & 5.63  & 0.28 \\
CNY/USD      & 1.25  & 4.76  & 4.52  & 84.68 & 4.58  & 0.21 \\
Gold futures & 0.64  & 0.26  & 0.66  & 0.09  & 97.72 & 0.64 \\
BTC/USD      & 0.01  & 0.08  & 0.01  & 0.15  & 0.66  & 99.08 \\
\bottomrule
\end{tabular}
\end{adjustbox}

\vspace{4pt}
\begin{minipage}{0.95\textwidth}
\footnotesize
\emph{Note.} Values are reported in percentage terms. Rows denote the variables whose generalized forecast-error variance is decomposed, while columns denote the variables contributing to that forecast-error variance. Diagonal elements represent own-variable contributions, and off-diagonal elements represent cross-variable spillover contributions. These are statistical GFEVD contributions, not identified structural shocks. The matrix is computed from the full-sample GFEVD for the network excluding the dollar index with forecast horizon $H=10$. Row sums may differ from 100 by 0.01 because of rounding.
\end{minipage}
\end{table}

\section{Full signed pairwise net-spillover matrix}
\label{supsec:net_matrix}
\begin{table}[H]
\centering
\caption{Full signed pairwise net-spillover matrix}
\label{tabS:full_signed_net_matrix}
\small
\setlength{\tabcolsep}{6pt}
\renewcommand{\arraystretch}{1.15}
\begin{adjustbox}{max width=\textwidth}
\begin{tabular}{lrrrrrr}
\toprule
 & JPY/USD & GBP/USD & EUR/USD & CNY/USD & Gold futures & BTC/USD \\
\midrule
JPY/USD      & 0.00  & -0.44 & -2.11 & 0.32  & -9.70 & 0.01 \\
GBP/USD      & 0.44  & 0.00  & -1.80 & 1.50  & -3.44 & -0.59 \\
EUR/USD      & 2.11  & 1.80  & 0.00  & 1.68  & -4.97 & -0.27 \\
CNY/USD      & -0.32 & -1.50 & -1.68 & 0.00  & -4.49 & -0.06 \\
Gold futures & 9.70  & 3.44  & 4.97  & 4.49  & 0.00  & 0.02 \\
BTC/USD      & -0.01 & 0.59  & 0.27  & 0.06  & -0.02 & 0.00 \\
\bottomrule
\end{tabular}
\end{adjustbox}

\vspace{4pt}
\begin{minipage}{0.95\textwidth}
\footnotesize
\emph{Note.} Values are reported in percentage-point terms. A positive value indicates that the row variable is a net pairwise transmitter to the column variable, whereas a negative value indicates that the row variable is a net pairwise receiver from the column variable. The source-tracing kernel reverses this transmission direction, so a transition $i\to j$ traces an upstream net-spillover supplier rather than a forward transmission target. For a pure net source, the uniform row is an augmented restart closure; its transitions, including the diagonal entry, are retained in both $F_{\mathrm{out}}$ and deterministic VC. The matrix is computed from the full-sample pairwise net connectedness results for the network excluding the dollar index with forecast horizon $H=10$.
\end{minipage}
\end{table}

\section{Robustness checks across rolling-window sizes and forecast horizons}
\label{supsec:rolling_window_robustness}
This appendix reports robustness checks based on alternative rolling-window sizes.
The main analysis uses a rolling window size of 100 daily observations. To examine whether the main findings are sensitive to the window-length choice, additional results are reported using shorter and longer rolling windows. The shorter window captures temporary changes in connectedness more sensitively, whereas the longer window smooths short-run fluctuations and reflects more persistent spillover patterns.
In addition, Table~\ref{tabS:eds_lds_robustness_grid} summarizes the Joint EDS and LDS
frequencies over the full specification grid $W\in\{52,100,150\}$ and $H\in\{5,10,15\}$.

\begin{figure}[H]
\centering
\includegraphics[width=0.90\textwidth]{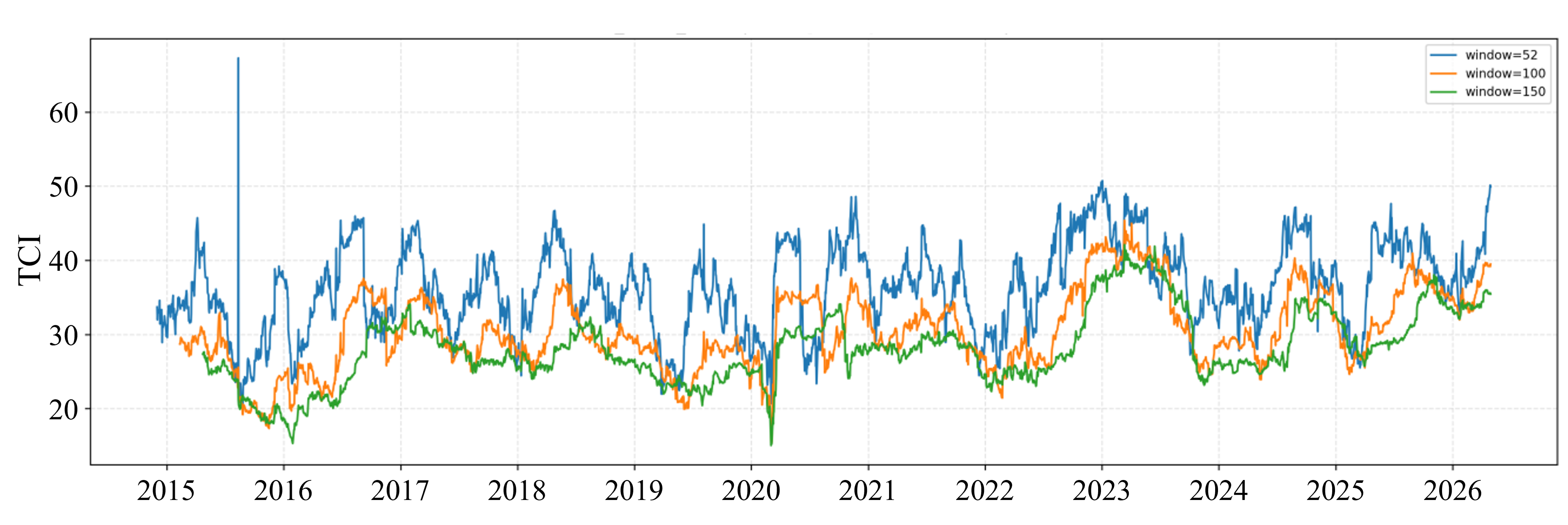} \caption{Rolling TCI comparison across alternative rolling-window sizes. 
The figure compares the time-varying TCI of the network excluding the dollar index across the baseline window size of 100 and alternative window sizes used for robustness checks.
The forecast horizon is fixed at $H=10$.} \label{fig:appendix_rolling_tci_window_comparison}
\end{figure}

\begin{table}[H]
\centering
\caption{Joint EDS and LDS diagnostic frequencies across rolling-window sizes and forecast horizons}
\label{tabS:eds_lds_robustness_grid}
\small
\begin{tabular}{p{0.35\textwidth}cc}
\toprule
Specification grid
& Maximum Joint EDS frequency
& Maximum LDS frequency \\
\midrule
$W\in\{52,100,150\}$,
$H\in\{5,10,15\}$,
all assets and asset sets
& 0.000
& 0.000 \\
\bottomrule
\end{tabular}

\vspace{1mm}

\begin{minipage}{0.94\textwidth}
\small
\textit{Note}: Each value is the maximum asset-level state frequency across
all analyzed asset sets, rolling-window sizes, and forecast horizons.
Joint EDS requires simultaneous upper-tail exceedances of VC and
$F_{\mathrm{out}}$, whereas LDS requires simultaneous lower-tail
exceedances. Both state frequencies are zero in every specification.
Because the thresholds are estimated in sample and adjacent rolling windows
overlap, these values are descriptive frequencies. They do not establish
statistical independence, intrinsic antagonism between the diagnostics,
predictive calibration, or the absence of financial contagion.
\end{minipage}
\end{table}

\bibliographystyle{elsarticle-num} 
\bibliography{ref}

\end{document}